%% file: main.tex
\documentclass[11pt]{article}
\usepackage{sabee}
\patchcmd{\abstract}{\small}{\normalsize}{}{}

\AtEveryBibitem{%
  \clearlist{language}%
  \clearlist{location}
}
\AtBeginBibliography{%
  \setlength{\emergencystretch}{2em}%
}
\AtEveryBibitem{ %
    \clearfield{day}
    \clearfield{month}
    \clearfield{series}
    \clearfield{venue}
    \clearname{editor}
    \clearlist{publisher}
    \clearlist{location} %
    \clearfield{venue}
    \clearfield{issn}
    \clearfield{isbn}
    \clearfield{urldate}
    \clearfield{eventdate}
    \clearfield{pages}
    \clearfield{number}
    \clearfield{volume}
}

\hypersetup{
  pdftitle={A classical oracle separation between QMA and QMA(2)},
  pdfauthor={Sabee Grewal and Saachi Mutreja and Kabir Tomer}
}
\allowdisplaybreaks
\setlist{
    itemsep=2pt plus .5pt minus .5pt,
    topsep=4pt plus 1pt minus 1pt,
    parsep=0pt,
    partopsep=0pt,
    labelindent=.5em,
    leftmargin=*
  }

\newcommand{\sgn}{\operatorname{sgn}}
\newcommand{\RE}{\mathsf{RE}}
\newcommand{\MIP}{\mathsf{MIP}}
\newcommand{\QMAT}{\mathsf{QMA(2)}}
\newcommand{\NEXP}{\mathsf{NEXP}}

\newcommand{\QMA}{\mathsf{QMA}}

\newcommand{\QCMA}{\mathsf{QCMA}}

\newcommand{\SBQP}{\mathsf{SBQP}}

\begin{document}
\title{A classical oracle separation between $\QMA$ and $\QMAT$}
\author{%
  Sabee Grewal\thanks{IBM Research and Columbia University. \href{mailto:sabee@ibm.com}{\texttt{sabee@ibm.com}}}
  \and
  Saachi Mutreja\thanks{Columbia University. \href{mailto:sm5540@columbia.edu}{\texttt{sm5540@columbia.edu}}}
    \and 
    Kabir Tomer\thanks{UIUC \href{mailto:ktomer2@illinois.edu}{\texttt{ktomer2@illinois.edu}}
  }
}
\date{}
\maketitle

\begin{abstract}
We exhibit a classical oracle relative to which $\QMA \neq \QMAT$. The separation is based on a separable Hamiltonian problem defined by signed edge constraints on a graph. The YES instances have zero-energy product states, while the NO instances, constructed using the antisymmetric subspace, have entangled zero-energy states but no low-energy product states.
For the $\QMA$ lower bound, we use compressed-oracle techniques to track the information learned by a quantum computation
making few queries to a YES instance. 
The key step is to map the resulting superpositions of partial oracle records to the NO family without significantly changing the verifier's behavior.
\end{abstract}

\section{Introduction}\label{sec:intro}

Quantum information can affect proof systems in seemingly unpredictable ways.
In a stunning example, we know that $\MIP^* = \RE$, meaning that entanglement among multiple provers enables a polynomial-time classical verifier to verify undecidable languages~\cite{JNV+20}.
This is all the more surprising given that, in the early study of the model, showing that $\MIP^* \supseteq \MIP$ was itself a breakthrough result~\cite{IV12}.
By contrast, the computational power of quantum interactive proofs and quantum refereed games matches that of their classical counterparts~\cite{Sha92,FK97,GW07,JJuW11}.

In this work, we investigate a resource with no direct classical analogue: \emph{unentanglement}.
Its simplest complexity-theoretic manifestation is the complexity class $\QMAT$, in which a polynomial-time quantum verifier receives two quantum proofs that are promised to be unentangled. 
Compared with $\QMA$, the essential difference is a restriction on Merlin: the prover is not allowed to create entanglement across some fixed bipartition of the proof space.
Whether this restriction increases verification power is closely tied to our understanding of entanglement and separability~\cite{JWL26}, and thus understanding the power of $\QMAT$ is regarded as one of the major open problems in quantum complexity theory.

Despite being introduced around a quarter century ago~\cite{KMY03}, we still cannot locate $\QMAT$ in the complexity landscape beyond the essentially trivial containments $\QMA \subseteq \QMAT \subseteq \NEXP$.
Worse yet, there isn't even a clear consensus as to where its power should lie: some fully believe that $\QMA = \QMAT$, while others suspect that $\QMAT = \NEXP$. 

Let us discuss both possibilities, starting with the possibility that $\QMAT=\NEXP$.
There are nearby settings in which unentanglement provably yields substantial additional verification power~\cite{bassirian2024quantummerlinarthurinternallyseparable,grewal2025unentanglementpostmeasurementbranchingquantum}.
However, these results exploit additional structure or resources beyond unentanglement, and therefore do not say anything about the power of unentanglement on its own.
More broadly, known routes to showing $\QMAT = \NEXP$ appear to require strong structural assumptions about the proofs, and removing these assumptions would seem to require the verifier to certify such structure from otherwise arbitrary unentangled proofs. 
There is reason to be pessimistic about this approach. In particular, even the seemingly basic task of testing whether an unknown quantum state is a low-degree phase state---the quantum analogue of one of the central objects in classical PCP constructions---was recently shown to require many copies of the state~\cite{alrabiah2026lowdegreetestsquantumstates}. 

Turning to the other end of the spectrum,
the evidence distinguishing $\QMAT$ from $\QMA$ has recently become more substantial.
In particular, Bostanci, Grewal, Haferkamp, Huang, Hwang, Natarajan, and Nirkhe~\cite{BGH+26} recently exhibited a \emph{quantum oracle} relative to which $\QMA \neq \QMAT$.
Their separation gives a black-box problem specified by an unknown unitary, for which two unentangled proofs enable efficient verification, whereas no polynomial-query $\QMA$ verifier can succeed. 
Thus, in this quantum black-box setting, the promise of unentanglement is a genuine computational resource. 

This oracle separation can be interpreted as a barrier to proving $\QMA = \QMAT$.\footnote{In addition to presenting barriers, an oracle separation must isolate a setting where two computational models behave differently, which can itself yield insight into the source of their differences. We will discuss this more shortly.} 
The separation of~\cite{BGH+26} rules out any generic transformation of a $\QMAT$ protocol into a $\QMA$ protocol that continues to work in the presence of arbitrary quantum oracles.
More concretely, it rules out a natural ``disentangle-then-verify'' strategy: take a single $\QMA$ witness, efficiently transform it into two unentangled proofs, and then run the $\QMAT$ verifier.

A quantum-oracle separation leaves open whether the two classes can still be separated relative to a \emph{classical} oracle. 
A recent example shows how this distinction can matter a lot for $\QMA$. Aaronson exhibited a quantum oracle relative to which $\QMA \neq \QMA_1$~\cite{Aar09}, whereas the recent proof that $\QMA = \QMA_1$ holds relative to \emph{every} classical oracle~\cite{grewal2026qmaperfectcompleteness,liu2026achievingperfectcompletenessone} (see also~\cite{agarwal2026nonstandardoraclesboundederrorcomplexity} for related examples). 
Indeed, a classical-oracle separation imposes a stronger relativization barrier on any proof that $\QMA=\QMAT$. 
Such separations can also be substantially harder to obtain. For example, the quantum-oracle separation between $\QMA$ and $\QCMA$~\cite{AK07} preceded its classical-oracle separation by nearly two decades~\cite{BHNZ26}.
One reason for this difficulty is that an arbitrary quantum oracle can hide a complicated quantum object behind a single black-box unitary, whereas representing the same phenomenon by a classical oracle necessarily exposes some classical description of that object, which an algorithm may be able to exploit.
This raises the natural question of whether the separation between $\QMA$ and $\QMAT$ survives relative to a classical oracle.
Our main result shows that it does.

\begin{theorem}
\label{thm:main}
There exists a classical oracle $A$ such that
\[
\QMA^A\neq\QMAT^A.
\]
\end{theorem}

Our construction builds on the symmetric--antisymmetric geometry underlying the quantum-oracle separation of~\cite{BGH+26}. Their oracle tests membership in a hidden subspace of a bipartite system. In the YES case, the subspace is spanned by a product state; in the NO case, it is a locally rotated copy of the antisymmetric subspace on four registers, two on each side, which stays a constant distance from every product state. A product witness therefore certifies the YES case, while~\cite{BGH+26} show via the unitary polynomial method~\cite{she-yuen-23} that no $\QMA$ verifier can distinguish the two cases.

The main obstacle to a classical-oracle analogue is to implement this membership test with a classical oracle that does not reveal the subspace. We start from a basic observation: the uniform superposition over a set is invariant under permutations of that set, so an oracle that applies random permutations of a hidden set lets a quantum verifier test for this state without learning the set. This idea appears in Lutomirski's component mixers~\cite{Lut11} and in the graph-oracle problems of Natarajan and Nirkhe~\cite{NN24} and Liu, Mutreja, and Yuen~\cite{LMY25}. Our main new ingredient is to use \emph{signed} superpositions and signed permutations, which let a classical oracle test for a hidden copy of the antisymmetric subspace. Since a classical oracle reveals every sign it returns, we also randomize the signs so that individual signs look the same in both cases.

Instantiating the construction of~\cite{BGH+26} with these tests yields a separable Hamiltonian problem on a signed graph. The vertices are the basis states of the two proof registers, the oracle gives access to the edges through signed permutations, and each edge asks the amplitudes at its endpoints to be equal or opposite. The question is whether the minimum energy over \emph{product} states is small or large. Our hard YES instances are frustration-free and admit a zero-energy product state. Strikingly, our hard NO instances are also frustration-free: they admit zero-energy states, but these states have antisymmetric structure and are necessarily entangled; indeed, we prove that every product state has energy bounded below by a constant. Two unentangled Merlins can therefore convince a verifier that makes a single query, so the problem is in $\QMAT$.

The main technical challenge is proving the $\QMA$ lower bound. The existence of entangled zero-energy states in our NO instances only rules out the natural strategy of estimating the energy; a proof of \cref{thm:main} must rule out an arbitrary $\QMA$ verifier. Our proof is a hybrid argument in the spirit of~\cite{BBBV}, built on the compressed-oracle framework of Zhandry~\cite{Zha19} and its recent extension to random permutations~\cite{Car25,CM26}. This framework represents what a few quantum queries reveal about a random oracle as a superposition of small partial records. We show that the records produced by queries to a YES oracle remain compatible with the NO family, and we handle witnesses that depend on the oracle with the standard technique of guessing the witness~\cite{MW05,Aar12}. We outline the proof in \cref{sec:techoverview}.

We conclude with two open questions suggested by our construction.
First, our signed-graph formulation is closely related to the graph-oracle problems studied by Natarajan and Nirkhe~\cite{NN24} and Liu, Mutreja, and Yuen~\cite{LMY25}.  
It is natural to ask whether the additional signed structure used here can be combined with those constructions to obtain a simple graph-based classical-oracle separation between $\QMA$ and $\QCMA$.  
Such a separation is now known~\cite{BHNZ26}, but the existing construction is based on a spectral variant of Forrelation; a graph-theoretic separation could give a substantially different and potentially simpler explanation for the advantage of quantum witnesses over classical ones.

Second, we highlight a striking difference between our separation and the quantum-oracle separation of~\cite{BGH+26}.  
Their lower bound applies even to $\SBQP$, and hence gives a quantum oracle relative to which $\QMAT \not\subseteq \SBQP$.
By contrast, our separating problem admits an $\SBQP$ algorithm, and therefore cannot yield such a separation relative to a classical oracle. 
Does there exist a classical oracle $A$ for which $\QMAT^A\not\subseteq\SBQP^A$?
Or does $\QMAT^A \subseteq \SBQP^A$ hold relative to every classical oracle $A$?  
In our view, resolving this question is one of the most interesting directions for understanding the power of unentanglement.

\pdfbookmark[1]{Acknowledgments}{acknowledgments}
\paragraph*{Acknowledgments.}
We thank Henry Yuen for many valuable conversations. SM would like to thank Aaron (Louie) Putterman for setting up a harness and running the prompts which eventually generated the initial draft of this oracle separation. 
SG is supported by the Herman Goldstine Memorial Postdoctoral Fellowship.

\pdfbookmark[1]{Tool and computational resource disclosure}{tool-disclosure}
\paragraph*{Tool and computational resource disclosure.}
The original proof of this result was obtained through approximately 10 rounds of interaction with ChatGPT 6 Astra in Work mode. 
Directing the model to the expander 
distinguishing problem of~\cite{NN24} and its simplified formulation, Components, in~\cite{LMY25} appears to have been an important step in this process (and, in fact, the AI-generated construction used perfect matchings like these problems do, instead of cycle covers, which are used in the present paper).
The authors subsequently verified the mathematical arguments independently and revised and reorganized the proofs and exposition.
They take full responsibility for the correctness, attribution, and presentation of this work. For transparency, the original AI-generated manuscript is available \href{https://drive.google.com/file/d/1keazrUzaBOkmGGnFRV_goQm275ReXG3e/view?usp=sharing}{here}.

\section{Technical Overview}
\label{sec:techoverview}
 
\subsection{Constructing the classical oracle}

\paragraph{The quantum-oracle construction.}
Our starting point for the classical oracle separation is the problem used by~\cite{BGH+26} for the quantum oracle separation, which is a restriction of the Entangled Subspace problem of She and Yuen~\cite{she-yuen-23}. Consider four registers, each with computational basis indexed by a set of labels $[N]$, where the first two registers belong to subsystem $A$ and the last two belong to subsystem $B$. The oracle is a reflection about a hidden subspace of this space, so a single query lets the verifier measure whether a state lies in the subspace.
In particular their YES and NO subspaces are defined by the corresponding projectors
\[
P_{\mathrm{YES}}=(U_A\otimes U_B)\ketbra{+,+}{+,+}(U_A\otimes U_B)^\dagger
\qquad\text{and}\qquad
P_{\mathrm{NO}}=(U_A\otimes U_B)\,\Pi_{\mathrm{alt}}\,(U_A\otimes U_B)^\dagger,
\]
where $U_A$ and $U_B$ are independent Haar-random unitaries, $\ket{+,+}$ is the uniform superposition over the four registers, and $\Pi_{\mathrm{alt}}$ is the projector onto the antisymmetric subspace of the four registers. As a result the YES subspace is spanned by a single product state, while every product state has squared overlap at most $1/3$ with the NO subspace~\cite[Lemma~4.1]{BGH+26}. Our overall goal is to replace these quantum testing oracles with classical ones.
\paragraph{Testing signed superpositions with permutations.}
A natural way to let a classical oracle test a quantum state goes back to quantum money from knots~\cite{FGHLS12}, where part of the verification checks that the money state is invariant under a classical Markov chain. Lutomirski formalized this approach through component mixers~\cite{Lut11}, and related graph constructions appear in prior work on oracle separations~\cite{NN24,LMY25}. The basic idea is as follows.
 
Let $S$ be a hidden set of strings, with a hidden sign $\phi(x)\in\{\pm1\}$ for each $x\in S$, and consider the signed uniform superposition
\[
\ket{S^\phi}:=\frac{1}{\sqrt{|S|}}\sum_{x\in S}\phi(x)\ket{x}.
\]
A signed permutation of $S$ is a permutation $\pi$ of $S$ together with a sign $\sigma(x)\in\{\pm1\}$ for each $x\in S$; it acts as the unitary $W\ket{x}:=\sigma(x)\ket{\pi(x)}$. When the signs are set to be $\sigma(x):=\phi(x)\phi(\pi(x))$, then $W$ preserves $\ket{S^\phi}$ for every $\pi$. Conversely, it can be shown that a state on $S$ that is invariant under these signed permutations for every $\pi$ must have amplitudes proportional to $\phi$ throughout $S$, and hence must be proportional to $\ket{S^\phi}$. We can therefore test for $\ket{S^\phi}$ by checking invariance under signed permutations. Given classical oracle access to $\pi$, to its inverse, and to the signs $\sigma(x)$, a quantum algorithm can apply $W$ and use a simple interference test that rejects a state $\ket{\psi}$ with probability $\frac14\bigl\|\ket{\psi}-W\ket{\psi}\bigr\|^2$ (see \cref{prop:qma2-upper}). If the oracle provides many independent random signed permutations of $S$ and the verifier tests a random one of them, then every state far from $\ket{S^\phi}$ is rejected with constant probability. This allows a verifier to use classical oracles to test for $\ket{S^\phi}$ without explicitly requiring $S$ or $\phi$ themselves.
 
The same reasoning applies to a hidden partition $X=S_1\sqcup\cdots\sqcup S_t$ with signs $\phi$ on $X$. Random signed permutations that preserve each part of the partition test membership in
\[
\operatorname{span}\{\ket{S_1^\phi},\ldots,\ket{S_t^\phi}\}.
\]
They force the amplitudes to be proportional to $\phi$ within each part, while allowing arbitrary superpositions across parts. This subspace test has been considered in previous works using only positive superpositions, where $\phi\equiv1$; we extend the test to use signed superpositions, for reasons we explain later.
 
\paragraph{Replacing Haar-random unitaries with a signed superposition encoding.}
The quantum oracle separation consisted of a simple subspace obfuscated using Haar random unitaries. While the underlying subspaces are easy to test, undoing the obfuscation involves inverting these unitaries, and therefore the overall test seems to require a quantum oracle. The key idea behind the construction is to replace the random unitaries with a different obfuscation procedure. In particular, we map each pair of labels to a random signed superposition over a hidden subset of strings. These superposition states are natural stand-ins for Haar random states, and are in fact indistinguishable from them to polynomially bounded algorithms~\cite{GTB23,JMW23}, while also being amenable to testing via classical oracles. 

 
Concretely, set $N=2^k$, and for each party $P\in\{A,B\}$, let $\mathcal{X}_P$ be a set of $N^3$ strings. The encoding $E_P$ partitions $\mathcal{X}_P$ into sets $C^P_{a,b}$, one for each pair of labels $(a,b)\in[N]^2$, and maps each pair of labels to the signed superposition over its set:
\[
E_P\ket{a,b}:=\frac{1}{\sqrt{|C^P_{a,b}|}}\sum_{x\in C^P_{a,b}}h_P(x)\ket{x}.
\]
The sets are chosen by assigning each string a uniformly random pair of labels, and the signs $h_P(x)\in\{\pm1\}$ are uniformly random. Since the sets are disjoint, $E_P$ is an isometry. Since $E_A\otimes E_B$ acts separately on $A$ and $B$, it maps product states to product states. The encoded states are supported on the set $V:=\mathcal{X}_A\times\mathcal{X}_B$ of pairs of strings, where $E_A\otimes E_B$ maps $\ket{a,b}\ket{c,d}$ to the signed superposition over $C^A_{a,b}\times C^B_{c,d}$ with signs $\chi(x_A,x_B):=h_A(x_A)h_B(x_B)$. For intuition, we assume for now that all sets $C^P_{a,b}$ have the same size.
 
\paragraph{Testing the encoded subspaces.}
Replacing $U_A\otimes U_B$ by $E_A\otimes E_B$ in the projectors $P_{\mathrm{YES}}$ and $P_{\mathrm{NO}}$, the verifier should test membership in the span of $(E_A\otimes E_B)\ket{+,+}$ in the YES case, and in the image of the antisymmetric subspace under $E_A\otimes E_B$ in the NO case. We now show that signed permutations, which a classical oracle can provide, implement both tests.

\paragraph{The YES test.} 
In the YES case, the encoded state is a single signed superposition, over all of $V$:
\[
(E_A\otimes E_B)\ket{+,+}=\frac{1}{\sqrt{|V|}}\sum_{x\in V}\chi(x)\ket{x}=\ket{h_A}\ket{h_B},
\]
where $\ket{h_P}$ is the signed superposition over $\mathcal{X}_P$ with signs $h_P$. The oracle can therefore provide $N$ independent, uniformly random permutations $\pi_1,\ldots,\pi_N$ of $V$ and their inverses, along with access to the signs
\[
\sigma_p(x):=\chi(x)\chi(\pi_p(x)).
\]
The verifier picks a random signed permutation thus defined and checks that the state is invariant under the corresponding operation. The state $\ket{h_A}\ket{h_B}$ passes the test while states far from it are rejected with constant probability. 
 
\paragraph{The NO test.}
In the NO case, we must test a whole subspace. The basis vectors for the antisymmetric subspace can be indexed by four-element sets of form $Q=\{q_1<q_2<q_3<q_4\}\subseteq[N]$, where the corresponding basis state is defined as:
\[
\ket{\mathrm{Alt}_Q}:=\frac{1}{\sqrt{24}}\sum_{\tau\in S_4}\operatorname{sgn}(\tau)\,\ket{q_{\tau(1)},q_{\tau(2)}}_A\ket{q_{\tau(3)},q_{\tau(4)}}_B.
\]
Thus every ordering of the four labels appears once, with a positive or negative sign according to the parity of the permutation producing that ordering. Under the encoding map, $\ket{\mathrm{Alt}_Q}$ becomes a signed superposition over the hidden set
\[
\mathcal{B}_Q:=\bigcup C^A_{\ell_1,\ell_2}\times C^B_{\ell_3,\ell_4},
\]
where the union is over the $24$ orderings $(\ell_1,\ell_2,\ell_3,\ell_4)$ of $Q$. Its sign at $x = (x_A, x_B)$ is $\chi(x)\varepsilon(x)$, where $\varepsilon(x)$ is the sign of the permutation that sorts the four labels of $x$, and $\chi(x)$ is defined above as $h_A(x_A)h_B(x_B)$. That is, under our assumption of equal set sizes, $(E_A\otimes E_B)\ket{\mathrm{Alt}_Q}$ equals
\[
\ket{\psi_Q}:=\frac{1}{\sqrt{|\mathcal{B}_Q|}}\sum_{x\in\mathcal{B}_Q}\chi(x)\varepsilon(x)\ket{x}.
\]
 
The sets $\mathcal{B}_Q$, which we call blocks, are disjoint, so the encoded antisymmetric subspace is the span of signed superpositions over the parts of a partition, and the partition version of the test applies. For each $p\in[N]$, the permutation $\pi_p$ independently permutes every block uniformly at random, with signs $\sigma_p(x)=\chi(x)\varepsilon(x)\chi(y)\varepsilon(y)$, where $y=\pi_p(x)$.
Invariance under the induced signed permutations forces the amplitudes to be proportional to $\chi\varepsilon$ within each block, while allowing arbitrary superpositions across blocks, so the states that pass every test are exactly the superpositions of the states $\ket{\psi_Q}$. The remaining elements of $V$, whose four labels are not all distinct, form a ``remainder" block. Antisymmetric states cannot be supported on such labels, so we set $\pi_p$ to also permute the remainder block uniformly at random, with signs $\sigma_p(x)=-\chi(x)\chi(y)$, which penalize every state supported there. In summary, the NO signs are 
\[
\sigma_p(x):=\chi(x)\chi(y)\,s(x,y),\qquad\text{where}\qquad
s(x,y):=\begin{cases}\varepsilon(x)\varepsilon(y), & x,y\in\mathcal{B}_Q,\\ -1, & x,y\text{ in the remainder block.}\end{cases}
\]
Testing for antisymmetry this way necessitates the use of negative signs in the permutations; this is the primary reason we need to extend from permutations to signed permutations. The lower bound additionally requires the superpositions to also be signed, since if $\chi(x)$ was always $1$ the test for YES instances would only have positive signs, while NO instances would also have negative signs. Adding the signs to the superpositions helps mask the signs in the permutations.

Since the encoding acts separately on $A$ and $B$, the product-overlap bound of~\cite[Lemma~4.1]{BGH+26} carries over: every product state has squared overlap at most $1/3$ with the span of the states $\ket{\psi_Q}$, whereas entangled states can lie entirely inside it.
 
In both cases, the oracle consists only of $N$ signed permutations of $V$, and the verifier runs the same invariance test. Averaged over the random permutations, applying a signed permutation projects each block onto its satisfying state, with a negative sign on the remainder block. That is, the averaged operator $\overline{\Gamma}$ is
\[
\ketbra{h_A,h_B}{h_A,h_B}\quad\text{(YES)}\qquad\text{or}\qquad
\sum_{Q}\ketbra{\psi_Q}{\psi_Q}-\ketbra{\psi_{\mathrm{rem}}}{\psi_{\mathrm{rem}}}\quad\text{(NO)},
\]
where $\ket{\psi_{\mathrm{rem}}}$ is the normalized state with amplitudes proportional to $\chi(x)$ on the remainder block. Apart from the remainder term, which is always a penalty, these are exactly the projectors $P_{\mathrm{YES}}$ and $P_{\mathrm{NO}}$, with $E_A\otimes E_B$ in place of $U_A\otimes U_B$. So, on average, the verifier's invariance test is the membership test of~\cite{BGH+26} for the encoded subspaces.
 
 
\subsection{The signed-graph Hamiltonian problem}
These oracles define a Hamiltonian problem on a signed graph. Its vertices are the elements of $V$, and each $p\in[N]$ specifies a signed cycle cover: every vertex $x$ has a directed edge to $\pi_p(x)$ with sign $\sigma_p(x)$. The oracle provides local access to this graph: given $p$ and a vertex, one can query its neighbor and the corresponding sign, in either direction. More formally, the oracle takes a cycle-cover index $p$, a vertex $v$, and a direction bit $z$, and returns
\[
O(p,z,v)=\begin{cases}\bigl(\pi_p(v),\sigma_p(v)\bigr), & z=0,\\[2pt] \bigl(\pi_p^{-1}(v),\sigma_p(\pi_p^{-1}(v))\bigr), & z=1.\end{cases}
\]
Thus $z$ simply toggles between querying the forward and backward neighbor of $v$.
 
We now turn the signed graph into a Hamiltonian. Consider a single edge
\[
x\xrightarrow{\;\sigma\;}y,\qquad\sigma\in\{\pm1\}.
\]
If a state $\sum_{x\in V}\alpha_x\ket{x}$ has amplitudes $\alpha_x$ and $\alpha_y$ at the two vertices in this edge, then we want this edge to contribute no energy precisely when $\alpha_x=\sigma\alpha_y$. We therefore associate to the edge the nonnegative penalty
\[
\frac14|\alpha_x-\sigma\alpha_y|^2,
\]
which vanishes only when $\alpha_x=\sigma\alpha_y$. The Hamiltonian simply averages this penalty over all signed edges in the graph. For the same state $\ket{\psi}$, its energy is
\[
\bra{\psi}H\ket{\psi}=\frac{1}{4N}\sum_{p\in[N]}\sum_{x\in V}\bigl|\psi(x)-\sigma_p(x)\psi(\pi_p(x))\bigr|^2,
\]
which is exactly the probability that the invariance test, for a uniformly random $p$, rejects $\ket{\psi}$. Consequently, a zero-energy state is exactly an assignment of amplitudes that satisfies every signed edge constraint simultaneously.
 
Our problem is a \emph{separable} Hamiltonian problem: the energy is minimized only over product states across the bipartition $A|B$. 
Accordingly, we define the minimum separable energy
\[
\lambda_{\mathrm{sep}}(H):=\min_{\substack{\ket{\psi} = \ket{\phi_A} \otimes \ket{\phi_B},\\ \|\ket{\psi}\|=1}}\bra{\psi}H\ket{\psi}.
\]
 The promise problem is to distinguish, given oracle access to the signed graph, whether $\lambda_{\mathrm{sep}}(H)\le a$ or $\lambda_{\mathrm{sep}}(H)\ge b$, for constants $0\le a<b$.
 
In our YES case, $\ket{h_A}\ket{h_B}$ satisfies every edge exactly, so $\lambda_{\mathrm{sep}}(H)=0$. In the NO case, low energy states must be close to superpositions of states $\ket{\psi_Q}$. However, since these states are antisymmetric, any linear combination of these states is entangled. The $\QMAT$ upper bound is therefore straightforward: the two Merlins provide the two factors of a product state, and Arthur runs the invariance test for a uniformly random cycle cover, using one oracle query (\cref{prop:qma2-upper}).
 
Two details are needed to obtain a constant gap in the NO case. First, the random sets $C^P_{a,b}$ are only approximately equal in size, so $\ket{\psi_Q}$ is only close to $(E_A\otimes E_B)\ket{\mathrm{Alt}_Q}$; a Chernoff bound shows that this distortion is small. Second, $H$ averages $N$ independent random cycle covers, so by matrix concentration it is close to its average over the permutations, which is $\frac12(\id-\overline{\Gamma})$ for the averaged operator $\overline{\Gamma}$; the negative remainder term in $\overline{\Gamma}$ only raises the energy. Together, these show that with high probability every product state has energy at least $1/4$, while the states $\ket{\psi_Q}$ still have energy exactly zero (\cref{proposition_goodNO}).

\subsection{The \texorpdfstring{$\QMA$}{QMA} lower bound}
The discussion above explains why the natural $\QMAT$ verifier does not immediately give a $\QMA$ protocol. Its soundness relies on the fact that the two Merlins provide a product state; a single $\QMA$ prover has no such restriction and may instead send an arbitrary entangled state across the two registers. Indeed, our NO instances have entangled states that satisfy the signed-edge tests perfectly. Thus the interference test used for the $\QMAT$ upper bound is not sound against arbitrary $\QMA$ witnesses.
 
This by itself, however, does not prove a separation. It only shows that one particular $\QMA$ verification strategy fails. To prove the lower bound, we must therefore show that no polynomial-query $\QMA$ verifier can distinguish the YES instances from the NO instances.

\label{subsubsec:overview-lower-bound}
We prove this using a \emph{transfer}: a map that takes the state of a YES oracle, after a verifier has queried it a few times, to the state of a NO oracle that is consistent with everything the verifier has seen.
The classical warm-up below shows, in the simplest setting, why a transfer suffices and which properties it needs: it must keep everything the algorithm has already seen, and making a query before or after it must give nearly the same result.
We then explain what changes for quantum queries, and finally why a transfer with these properties also rules out $\QMA$ verifiers, whose witness may depend on the entire oracle.

\paragraph{Warm up: $\mathsf{BPP}$ lower bound} Before explaining our lower bound, we first warm up by proving a $\mathsf{BPP}$ lower bound. Several steps here are more than this weaker bound needs, but each of them has a quantum counterpart. In particular, we lazily sample the oracle only in the form that compressed oracles support for quantum queries: as independent random permutations and independent random signs.

Let \(\mathcal D_+\) and \(\mathcal D_-\) denote the hard YES and hard NO
distributions described above, and consider a BPP algorithm making \(q\) queries. We want to argue that such an algorithm cannot distinguish YES instances from NO instances. We accomplish this by a
hybrid argument.
To start with, instead of sampling a whole YES oracle in advance, we choose the oracle parameters as they
are needed, that is, as the algorithm makes queries, and we keep earlier choices fixed. This is called \emph{lazy
sampling}. It will let us describe the computation using small databases, while leaving most of the oracle unspecified.
The main idea behind lazy sampling is to maintain a \emph{permutation database}
\(D=(D_p)_{p\in[N]}\), consisting of one table \(D_p\) for each
permutation \(\pi_p\). Each table  \(D_p\) is a partial bijection: its entries
have distinct inputs and distinct outputs. The lazy sampling procedure also includes a sign table
\(H\). This table stores the coordinate-sign values \(h_A(u)\) and \(h_B(v)\) of the queried vertices.
Before the algorithm makes any queries, all tables are empty.

When the algorithm queries the permutation \(\pi_p\) at \(x\), if \(D_p\) already contains a record of the form \(x\mapsto y\), then the output is $y$.
Otherwise,  \(y\) is chosen uniformly from the set of all vertices not yet used as
outputs in \(D_p\), and the edge \(x\mapsto y\) is added to the table \(D_p\). For an
inverse query at \(y\), if there exists a record of the form $x\rightarrow y$ in \(D_p\), then we output $x$;
otherwise, \(x\) is chosen uniformly from the vertices not yet used as
inputs in \(D_p\), and \(x\mapsto y\) is added to \(D_p\).

Each query also returns the sign of the queried edge.
Recall that in the YES construction this sign is
$
\sigma_p(x)=\chi(x_A,x_B)\chi(y_A,y_B).
$
The lazy sampler does the following- it looks up the signs of $x_A,x_B,y_A$ and $y_B$  in \(H\). For all coordinates $u$ such that $h(u)\notin H$ the lazy sampler chooses an independent uniform sign from \(\{\pm1\}\)
and stores it in \(H\) as $h(u)$. The sign of $u$ is now stored in $H$ for future queries and is not resampled again.
Now we will assemble these ideas to show how the lazy sampler replaces
the YES queries with NO queries one by one in a sequence of hybrids. In order to do so, we define a
randomized map \(\mathcal T\) that acts on \((D,H)\), and transforms it into a permutation database and sign table $(D_{\theta},H')$ for a NO oracle. Crucially, this transformation should remain
consistent with all the query answers and signs already stored in $(D,H)$.

In a bit more detail, the output of \(\mathcal T\) consists of a randomly sampled
pair-labeling \(\theta=(c_A,c_B)\), a NO permutation database
\(D_\theta=(D_{p,\mathcal B})_{p,\mathcal B}\), and a coordinate-sign table \(H'\).
Here \(\mathcal B\) ranges over the blocks determined by \(\theta\), and
\[
 D_{p,\mathcal B}:=\{x\mapsto y\in D_p:x,y\in\mathcal B\}.
\]
The lazy sampler must sample \(\theta\) such that the two endpoints of every entry already present in
\(D\) lie in the same block. Thus the tables \(D_{p,\mathcal B}\) partition the
entries of \(D_p\), and each \(D_p\) remains a partial bijection.
 Now, the table \(H'\) only assigns signs to the same coordinates as \(H\), and the signs corresponding to all the unqueried coordinates, that is coordinates not present in $H$, remain unspecified.  
However, the sampler is allowed to sample sign values that are different from the ones stored in $H$, as long as the final edge sign is consistent with $(D,H)$. More concretely, let \(h'_A,h'_B\) denote the new sign values for all coordinates in $H$. Then,
 we require, for every recorded entry \(x\mapsto y\in D_p\),
\[
 h'_A(x_A)h'_B(x_B)h'_A(y_A)h'_B(y_B)s_\theta(x,y)
 =\sigma_p(x),
\]
where \(\sigma_p(x)\) is the previously returned edge sign and
\(s_\theta(x,y)\) denotes the factor \(s(x,y)\) in the NO
construction for this \(\theta\). This condition preserves the 
recorded permutation entry and the edge sign.
After performing the transfer $\mathcal{T}$, the 
NO queries made by the $\mathsf{BPP}$ algorithm are answered by lazy sampling as above, using
\(D_\theta\) and \(H'\). For a fresh forward or inverse entry in the NO case, the sampler 
chooses an unused output or input, respectively, ensuring that it lies within the
block of the queried vertex. It samples any unassigned coordinate signs
uniformly, and store them in \(H'\). To return the edge sign, it returns the product of the
four coordinate signs multiplied by \(s_\theta(x,y)\).

For \(j=0,\ldots,q\), we can define a sequence of hybrids by answering the first \(j\) queries by YES lazy sampling, applying
\(\mathcal T\) to \((D,H)\),  and answering the remaining queries by NO lazy sampling. The first hybrid corresponds to sampling a random NO oracle while the last corresponds to sampling a random YES oracle. We will show that consecutive hybrids are nearly indistinguishable, which completes this warm-up lower bound.
Observe that hybrid $j$ and hybrid $j+1$ differ only in whether
the query \(j\) is answered before or after applying
\(\mathcal T\). So we will show that the distributions obtained in either case are very close.

\paragraph{Transferring a database to the NO family.}
First, we call an entry in $D$ \emph{removable} if one of  its endpoints has two coordinates that
occur nowhere else in $D$. The transfer map is constructed by repeatedly
removing such entries, and this process is called \emph{peeling}.
We say that a vertex sampled during a query is \emph{fresh} if neither of its
coordinates occurs in the queried vertex or in an earlier entry.
If, during the construction of a database by querying, every newly
sampled vertex was fresh, then removing entries in reverse of the  
order that they were queried, peels the database completely. We define \(\mathcal T\)
on  databases with this property. This is because, with very high probability the databases created by the queries made by the $\mathsf{BPP}$ algorithm will satisfy this property. 
To apply \(\mathcal T\), we first peel the database, keeping an
ordered list of the removed entries, their edge signs, and the endpoint 
used for each removal. We are left with an empty database and this
list. We now construct a NO database by reinserting the entries in
reverse order. 
At each step, we assign pair-labels and coordinate
signs as we described above, so that the entry's endpoints lie in the same NO block and its
edge sign agrees with the recorded answer.

Now consider the next entry $(x\rightarrow y)$ that needs to be reinserted.  The transfer samples any unassigned pair-labels and signs of
\(x\) independently from their original distributions. Then it samples
the two pair-labels of \(y\) conditioned on
placing it in the same block as $x$. Then we choose one of  $y$'s coordinate
signs, say $h'_A(y_A)$, uniformly and then choose $h'_B(y_B)$ such that $h'_A(x_A)h'_B(x_B)h'_A(y_A)h'_B(y_B)s_\theta(x,y)
 =\sigma_p(x)$. Both coordinates of \(y\) are new at this step, so these choices
preserve all entries already reinserted.

Once all entries have been reinserted, we can complete \(\theta\) by sampling
all unassigned pair-labels independently from their original distributions. Once $\theta$ is well defined, we group the entries into
\(D_\theta\) and keep the newly assigned coordinate signs in \(H'\).

It can be shown that the distribution produced by this process does
not depend on the peeling order, or on the endpoint used for removal/reinsertion, as long as it is consistent.

\paragraph{Comparing the two query orders.}
Let \(D\) be the current permutation database and let \(H\) be the sign table.  Now, fix a forward query \((p,x)\), where \(x\)
has no recorded image in \(D_p\). We will compare the probabilities of $(\theta, y)$ when we first query then transfer, and when we transfer first and then query. Fix a candidate outcome \((\theta,y)\), where the coordinates of
\(y\) occur neither in \(D\) nor in \(x\). Let 
\(e=(p;x\mapsto y)\), and let \(\mathcal B\) be the block of
\(x\) under \(\theta\).
We assume \(y\in\mathcal B\) for the calculations below since otherwise both experiments give probability zero to \((\theta,y)\).
Now, observe that the distribution over pair labels produced by \(\mathcal T\) depends only
on the permutation database \(D\). So let \(\mu_D(\theta)\) and
\(\mu_{D\cup\{e\}}(\theta)\) denote the probabilities of sampling \(\theta\) when transferring
\(D\) and \(D\cup\{e\}\), respectively. In the first experiment, where we query then transfer,
the YES query chooses \(y\) from \(N^6-|D_p|\) unused outputs,
then the transfer samples the labels. In the second experiment, where we transfer then query,
the transfer samples the labels first, then the NO query chooses
\(y\) from \(|\mathcal B|-|D_{p,\mathcal B}|\) unused outputs.
The probabilities of our fixed outcome are therefore
\[
 \begin{aligned}
 \text{Query, then transfer:}\quad&
 \frac{\mu_{D\cup\{e\}}(\theta)}{N^6-|D_p|},\\[1ex]
 \text{Transfer, then query:}\quad&
 \frac{\mu_D(\theta)}{|\mathcal B|-|D_{p,\mathcal B}|}.
 \end{aligned}
\]

To compare these expressions, we first relate \(\mu_{D\cup\{e\}}(\theta)\) and \(\mu_D(\theta)\). Since \(y\) is fresh, we can peel \(e\) first from
\(D\cup\{e\}\). Thus, in both hybrids, we first reinsert all
entries of \(D\), using the same pair-label sampling steps as we described above. Next,
in both hybrids, we sample any unassigned pair-labels at \(x\)
independently from their original distribution. Since the coordinates
of \(y\) occur neither in \(D\) nor in \(x\), its labels are
still unassigned.
Now, observe that for \(D\) , the two labels at \(y\) are sampled independently. However, for \(D\cup\{e\}\), they are sampled conditioned on placing
\(y\) in the block of \(x\). All other labels follow the same
sampling rules in both hybrids.

Let \(q_\theta(x)\) be the probability that two independent uniform
pair-labels at \(y\) place it in the block of \(x\), with the
labels at \(x\) fixed to their values in \(\theta\).
Conditioning on this event divides the probability of each permitted
pair of labels by \(q_\theta(x)\). Since our candidate \(\theta\)
puts \(y\) in \(\mathcal B\), we obtain
\[
 \mu_{D\cup\{e\}}(\theta)
 =\frac{\mu_D(\theta)}{q_\theta(x)}.
\]
Substituting into the first expression above,
\[
 \frac{\Pr[\text{outcome }(\theta,y)\text{ in experiment }1]}
      {\Pr[\text{outcome }(\theta,y)\text{ in experiment }2]}
 =\frac{|\mathcal B|-|D_{p,\mathcal B}|}
        {(N^6-|D_p|)q_\theta(x)}.
\]
A straightforward Chernoff bound shows that \(|\mathcal B|\) is close to
\(N^6q_\theta(x)\) w.h.p., in which case the ratio is close to 1, showing that the probabilities for $(\theta, y)$ in both experiments are close.

Next we compare the sign tables and the returned edge sign in both hybrids. Fix the
labels $\theta$ and all the signs of all  coordinates, except $(y_A,y_B)$. The new edge sign is
\(\sigma=cab\), where \(a,b\) are the fresh signs at \(y\) and
\(c\in\{\pm1\}\) is fixed. If we query first, then as we described above,\(\sigma,a\) is sampled
independently and uniformly, and then  \(b:=c\sigma a\). Now, if we apply the transfer map first, then we 
first sample \(a,b\) independently and uniformly, and then set
\(\sigma=cab\). Each triple satisfying \(\sigma=cab\) has
probability \(1/4\) in either ordering, so the distributions agree.

Finally, inverse queries have the
same argument with inputs and outputs exchanged. As a result, both experiments are negligibly close and therefore indistinguishable, completing the lower bound.
In words, the transfer keeps everything the algorithm has already seen, and making a query before or after it gives nearly the same result; these are the two properties we will need for quantum queries as well.

\paragraph{A witness that specifies oracle entries.}
Before we extend to the quantum setting, we consider how giving the algorithm access to an oracle dependent witness changes the argument. We extend the lower bound to $\mathsf{MA}$ verifiers that receive a classical witness chosen after the
oracle is fixed. As a first step, suppose its only effect is to fix
polynomially many signed permutation entries. We will return to this issue in a moment. 
For now, assuming these entries
are consistent with a YES oracle, we can sample the remaining oracle
parameters from the YES distribution conditioned on them.
Now, the first issue is that these entries need not form a database that peels completely:
the witness can fix entries with shared coordinates, and unlike in the $\mathsf{BPP}$ lower bound, we can no longer argue that this is unlikely, because the witness chooses these entries instead of sampling them. Thus, our transfer map should be able to handle such databases as well.
We modify the transfer map as follows- peel as many entries as possible, and let the remaining entries be database
\(C\). Then, we assign the pair-label \((1,2)\) to every \(A\)-coordinate
in \(C\), and \((3,4)\) to every \(B\)-coordinate in \(C\).
Every remaining entry $(x \rightarrow y)\in C$ now lies in \(\mathcal B_{\{1,2,3,4\}}\),
and its additional NO sign factor is $s_{\theta}(x,y)=\varepsilon(x)\varepsilon(y)=+1$, since both endpoints have labels $(1,2,3,4)$. For the signs, we let $h_{A}(x_A)=h'_{A}(x_A)$, and let $h_{B}(x_B)=h'_{B}(x_B)$. Since the entry $x \rightarrow y$ came from
a YES oracle, this preserves its edge sign as $s_{\theta}(x,y)=1$. Next, we reinsert the
peeled entries as before, and sample the remaining labels independently. Now, the comparison between querying before and after the transfer still applies, same as the $\mathsf{BPP}$ lower bound argument.

Now, we return to our assumption that the witness only fixes polynomially many signs and entries. In general, the dependence of the witness on the oracle may be complex and not easily captured by such fixing arguments. Nevertheless we will show that proving the transfer for this limited sort of witness state allows us to prove a full lower bound. Before we discuss this in detail, we show how to extend our lower bound technique to the quantum setting.

\paragraph{Quantum queries and coherent transfer}
We now turn to the main difficulty in extending
the classical transfer to quantum algorithms.
Classically, we construct a randomized map from
partial YES databases to compatible NO databases.
Its correctness follows by showing that querying
before or after the transfer gives approximately
the same probability distribution. Quantumly,
however, the verifier can query in superposition, and
different branches of its state can later interfere,
so producing the right probabilities is not enough.
We therefore need a linear transfer map, acting on the
joint state of the verifier and the oracle, that
approximately commutes with quantum oracle queries,
even on arbitrary superpositions of partial
oracle information.

To construct this map, we use compressed-oracle techniques
 \cite{Zha19,Car25,CM26}. These allow us to perform a quantum analogue of lazy sampling as in the classical case above. We are, however, restricted to lazy sampling distributions for which compressed oracle techniques are known, which in our case are random functions and random permutations. This is why our classical lower bound kept the oracle in the form of independent random permutations and independent random signs.

We use these techniques as a black box, through the following interface.
As before, the oracle's parameters are kept in a register that the verifier cannot access, now in superposition.
After a fixed change of basis on this register, the joint state of the verifier and the oracle is a superposition of \emph{records}, entangled with the verifier's registers.
A record is $|D,S\rangle$ in the YES case and $|\theta, D,S\rangle$ in the NO case.
Here, $D$ is a collection of partial permutation tables, as in the classical lazy sampler.
The set $S$ replaces the sign table, but it is written in the Fourier basis: it lists the coordinates whose signs the verifier's state depends on, rather than the values of those signs.
A query acts like a step of lazy sampling: it reads the entry for the queried vertex if $D$ has one, and otherwise adds one, in superposition over the possible answers.
For an edge $e=(p;x\to y)$, let $\partial e$ be the set of coordinates appearing an odd number of times among $x_A,x_B,y_A,y_B$; the YES edge sign $\chi(x)\chi(y)$ is the product of the signs of the coordinates in $\partial e$.
Returning this sign multiplies the state by that product, which in the Fourier basis replaces $S$ by $S\triangle\partial e$.
Finally, after $t$ queries every record has $O(t)$ entries and coordinates, so the records stay small.
 

We now want to construct a quantum analogue of the classical
transfer. Recall that classically, we reweight the
distribution of hidden NO pair labelings $\theta$ so
that the permutation entries in $D$ are compatible
with the NO block structure. Quantumly, we must
implement this reweighting coherently, rather than
sampling a labeling separately for each possible
record. We therefore define a linear completion map
\[
W|D,S\rangle
=
\sum_{\theta}
\sqrt{\nu(\theta)}\,a_{D,S}(\theta)
|\theta,D_\theta,S\rangle,
\]
where $\nu$ is the original distribution of hidden
NO labelings and $D_\theta$ retains the same
permutation entries, organized into NO blocks.
The amplitudes $a_{D,S}(\theta)$ implement the reweighting.
The first difficulty is constructing these
amplitudes consistently.
Recall that answering the sign of a recorded edge
$e=(p;x\to y)$ changes $S$ to $S\triangle\partial e$; the
corresponding NO query additionally introduces
the sign $s_\theta(x,y)$. Therefore, for
completion to commute with recorded sign queries,
we require
\[
a_{D,S\triangle\partial e}(\theta)
=
s_\theta(x,y)a_{D,S}(\theta)
\]
for every  record $(D,S)$ and every
$e\in D$. Classically, we can compensate for the
additional NO signs by reassigning the signs to coordinates in order to satisfy the edge sign constraint. Quantumly, we must construct a single
family of amplitudes satisfying these identities
simultaneously.

We construct these amplitudes through a modification
of the classical peeling procedure. Given $(D,S)$,
we repeatedly remove permutation entries, as in
the classical argument, but also update the sign
subset $S$. More concretely, an edge $e=(p;x\to y)$
can be peeled at $y$ when its two coordinates
occur nowhere else in the remaining database
and have the same membership in the current
subset $S_{\mathrm{cur}}$. Writing
$b_e\in\{0,1\}$ for this common membership,
we remove $e$ and update
\[
S_{\mathrm{cur}}
\longmapsto
S_{\mathrm{cur}}\triangle(b_e\partial e).
\]
Each peeled entry contributes a weight $w_e(\theta)$, the amplitude version of the classical reinsertion step, which places $y$ in the block of $x$.
We also include the factor $s_\theta(x,y)^{b_e}$
in the completion amplitude. Thus, when both
coordinates of $y$ belong to $S_{\mathrm{cur}}$,
their sign dependence is transferred to the
coordinates of $x$, together with the appropriate
NO sign correction.

However, unlike the classical argument, it is not clear that 
different peeling orientations result in the
same completion amplitude. For example, suppose
$D$ contains only one edge $e=(p;x\to y)$,
whose four coordinates are distinct, and let
\[
S=\{y_A,y_B\},
\qquad s_\theta(x,y)=-1.
\]
Peeling at $y$ gives $b_e=1$, leaving the free
sign subset $F_x=\{x_A,x_B\}$ and introducing
a factor $-1$. Peeling at $x$ gives $b_e=0$,
leaving $F_y=\{y_A,y_B\}$ and introducing
no additional sign. Thus, using only $w_e(\theta)$
and these peeling signs would give
$-w_e(\theta)$ and $+w_e(\theta)$ for the
same record.

Now, note that we did not have this issue in the classical argument, because both orientations generated the same
\emph{uniform conditional distribution} over coordinate sign assignments satisfying the recorded edge
constraint. Quantumly, however, our two peeling calculations
must yield exactly the same \emph{amplitude}. We explain how to achieve this in Section~\ref{subsec:database-transfer}.

Now, even after defining consistent amplitudes, we must
show that the transfer map approximately commutes with
quantum queries, which coherently insert or remove
permutation entries and act on arbitrary
superpositions of possible answers. Thus, we use the compressed-oracle representation
to construct a linear transfer map $ T$ satisfying
\[
\|(O_-T-TO_+)|\psi\rangle\|
\leq \text{negl}\||\psi\rangle\|,
\]
for all states $|\psi\rangle$ supported on small records,
where $O_+$ and $O_-$ are the YES and NO quantum
oracle queries.
We cannot ask for more than this: YES and NO oracles differ globally, since NO permutations preserve the hidden blocks, so no map commutes with queries on every state.

\paragraph{Removing the assumption on the witness.} Now we will tie the last loose end in the argument: dealing with arbitrary witnesses.

For the argument below, we use the following property. Start with
the purified YES distribution and an oracle-independent state of the
verifier and any reference registers. After \(t\) queries, the joint
state has a compressed representation as a superposition of records,
each involving \(O(t)\) permutation entries and sign coordinates.
Intuitively, if the state is prepared from scratch by querying the oracle, it has a structure similar to the oracle entry fixing witness from the previous section, and is compatible with our transfer map. 
Suppose a $\QMA$ verifier uses \(m\) witness qubits and makes \(q\)
queries, with completeness \(2/3\) and soundness \(1/3\).
An actual witness can depend on the entire oracle. Its initial joint
state with the oracle need not have a small-record representation.
Instead of starting with that witness, we use the standard trick of starting with the maximally
mixed state \(I/2^m\) on the \(m\)-qubit witness register,
independent of the YES oracle. This gives weight \(2^{-m}\) to
every witness state, including one that satisfies completeness.

For a fixed oracle \(\omega\), let \(A(\omega)\) be the positive
operator defined by
\[
 \Pr[\text{the verifier accepts }|w\rangle]
 =\langle w,A(\omega)w\rangle,
 \qquad 0\preceq A(\omega)\preceq I.
\]
Let \(A_+\) and \(A_-\) apply \(A(\omega)\) controlled by the
YES and NO parameter registers, respectively.\(A(\omega)\)    uses at most \(2q\) oracle
queries , and we keep the resulting states unnormalized.

Now, let \(\phi_0\) be the tensor product of the purified YES distribution
and a purification of the maximally mixed witness, and let
\(\phi_j=A_+^j\phi_0\).
The vector \(\phi_j\) is prepared using at most \(2qj\) queries
and postselection. For a polynomially bounded number \(r\) of applications,
the probability that all the projections succeed is
\[
 p_r:=\|\phi_r\|^2
 =2^{-m}\mathbb E_{\omega\sim\mathcal D_+}
       \operatorname{tr}\!\left(A(\omega)^{2r}\right)
 \ge 2^{-m}\left(\frac23\right)^{2r}.
\]
The inequality follows from completeness: every YES acceptance
operator has an eigenvalue at least \(2/3\). The exponent is \(2r\)
because the success probability is the squared norm after applying
\(A_+^r\).

\paragraph{Bounding each application using NO soundness.}
Now, the transfer construction gives a linear map \(T\) from YES to NO
oracle-parameter registers. 
Let \(O_+\) and \(O_-\) denote the
oracle queries controlled by the respective parameter registers.
Then, so far we have that, for every polynomial query bound, there is a negligible
\(\varepsilon=\varepsilon(k)\) such that the following holds for
every unnormalized vector \(\phi\) supported on short records and prepared from the purified YES
distribution :
\[
 (1-\varepsilon)\|\phi\|\le\|T\phi\|\le\|\phi\|,
 \qquad
 \|(O_-T-TO_+)\phi\|\le\varepsilon\|\phi\|.
\]

Crucially, these bounds must hold for every such \(\phi\).
The reason is that \(\phi_j\) is postselected: its norm can be as small as about \(2^{-m}\), so it can be concentrated on query histories that are unlikely under the YES distribution, such as two answers that share a coordinate.
This is why the transfer has to handle every small record, including databases that do not peel completely.

Now, take \(r=m+1\) and use the bounds we stated above with query bound \(2qr\).
Next, fix \(j<r\) and let \(\phi=\phi_j\). Preparing \(\phi\) and
then applying \(A_+\) uses at most \(2q(j+1)\le2qr\) queries.
The bounds therefore apply to every intermediate state and to
\(A_+\phi\). Thus,
\[
 \|(A_-T-TA_+)\phi\|\le2q\varepsilon\|\phi\|.
\]
On every oracle satisfying the NO promise, soundness gives
\(A(\omega)\preceq I/3\).
This also bounds a joint state in which the witness depends on the
NO parameters, because it holds for every witness at each such
parameter choice. The negligible weight on oracles that fail the NO
promise therefore gives
\(\|A_-T\phi\|\le(1/3+o(1))\|\phi\|\).
Using the lower bound on \(\|TA_+\phi\|\), we obtain
\[
 \begin{aligned}
 (1-\varepsilon)\|A_+\phi\|
 &\le\|TA_+\phi\|\\
 &\le\|A_-T\phi\|+2q\varepsilon\|\phi\|\\
 &\le\left(\frac13+o(1)\right)\|\phi\|.
 \end{aligned}
\]
Thus \(\|A_+\phi\|\le b\|\phi\|\), where
\(b=1/3+o(1)\).

Applying this inequality successively to \(\phi_0,\ldots,\phi_{r-1}\)
gives \(p_r\le b^{2r}\). Combining this with completeness gives
\[
 2^{-m}\left(\frac23\right)^{2r}\le b^{2r},
 \qquad\text{hence}\qquad
 \frac23\le2^{m/(2r)}b\le\sqrt2\,b.
\]
Since \(b=1/3+o(1)\), this is impossible for sufficiently large
\(k\). The contradiction rules out the $\QMA$ verifier without any
restriction on how its witness depends on the oracle.

\section{Preliminaries}

We assume familiarity with standard notions from quantum computation and quantum complexity theory; see, e.g.,~\cite{NC10,KSV02} for background.
For vectors, $\norm{\cdot}$ denotes the Euclidean norm, and for operators it denotes the operator norm.
A function $\mathsf{negl}(k)$ is \emph{negligible} if, for every constant $c>0$, $\mathsf{negl}(k)\le k^{-c}$ for all sufficiently large $k$.

We will use the following two standard concentration inequalities.

\begin{lemma}[Multiplicative Chernoff bound]
\label{lem:chernoff}
Let $X=\sum_{i=1}^m X_i$, where $X_1,\ldots,X_m$ are independent Bernoulli random variables, and let $\mu=\E[X]$. Then, for every $0<\eps\le1$,
\[
\Pr\left[\abs{X-\mu}\ge \eps\mu\right]
\le
2\exp\left(-\frac{\eps^2\mu}{3}\right).
\]
\end{lemma}


\begin{theorem}[Matrix Hoeffding inequality~{\cite[Theorem~1.3]{Tro12}}]
\label{thm:matrix-hoeffding}
Let $X_1,\ldots,X_N$ be independent Hermitian $d\times d$ random
matrices. Suppose that, for some $L>0$ and every $i$,
$$
\mathbb E X_i=0,
\qquad
\|X_i\|\le L
\quad\text{almost surely}.
$$
Then, for every $\eta>0$,
$$
\Pr\left[
\left\|\frac1N\sum_{i=1}^N X_i\right\|\ge\eta
\right]
\le
2d\exp\left(-\frac{N\eta^2}{8L^2}\right).
$$

\end{theorem}

The same conclusion holds conditionally on fixed parameters, provided
the matrices are independent under the conditional distribution and
satisfy the same norm bounds. In that case, $\overline\Gamma$ denotes
the conditional mean. If the bound is uniform in the fixed parameters,
averaging gives the same unconditional failure bound.

\section{The problem and the \texorpdfstring{$\QMAT$}{$\QMAT$} upper bound}
\label{sec:theproblem}
The black-box problem underlying our separation is a separable Hamiltonian problem, similar to the sparse separable Hamiltonian problem studied by Chailloux and Sattath~\cite{chailloux2012complexity}. 
Given a Hamiltonian $H$ acting on a bipartite Hilbert space $\calH_A \otimes \calH_B$, define its minimum separable energy by 
\[
\lambda_{\rm sep}(H) \coloneqq \min_{\norm{\ket\alpha} = \norm{\ket\beta} = 1} \bra{\alpha}\bra{\beta} H \ket{\alpha} \ket{\beta}.
\]
The general task is to distinguish between
$\lambda_{\rm sep}(H)\le a$ or $\lambda_{\rm sep}(H)\ge b$, for constants $a < b$, promised that one of these two cases holds. 

Our access model is different from the sparse Hamiltonian setting:
the Hamiltonian $H$ need not itself be sparse.
Instead, it is specified as an average of a large collection of simple constraints arising from signed edges of a regular graph.

\paragraph{The bipartite vertex space}
Let \(\calX_A\) and \(\calX_B\) be two copies of a set of size \(2^\ell\), and define
\[
\calH_A \coloneqq \operatorname{span}\{\ket{x_A}:x_A\in\calX_A\},
\qquad
\calH_B \coloneqq \operatorname{span}\{\ket{x_B}:x_B\in\calX_B\}.
\]
Thus, \(\calH_A\) and \(\calH_B\) are each \(2^\ell\)-dimensional Hilbert spaces, which we may regard as the state spaces of two \(\ell\)-qubit registers.
The computational basis of the joint space $\calH_A\otimes\calH_B$ is naturally indexed by pairs $(x_A,x_B)\in\calX_A\times\calX_B$.

We can now define the vertex set $V\coloneqq\calX_A\times\calX_B$. 
A \textit{vertex} $v=(x_A,x_B)\in V$ corresponds to the basis state $\ket v \coloneqq \ket{x_A}_A\ket{x_B}_B$.
We call elements of \(\calX_A\sqcup\calX_B\) \emph{coordinates}, and treat each set as disjoint, with their type $A$ or $B$ included in their description.
If
\[
\ket{\alpha}_A
=
\sum_{x_A\in\calX_A} a_{x_A}\ket{x_A}_A,
\qquad
\ket{\beta}_B
=
\sum_{x_B\in\calX_B} b_{x_B}\ket{x_B}_B,
\]
then their joint state is
\[
\ket{\alpha}_A\otimes\ket{\beta}_B
= \sum_{(x_A,x_B)\in V}
a_{x_A}b_{x_B}\ket{x_A}_A\ket{x_B}_B.
\]
Thus, viewing the basis states of $\calH_A \otimes \calH_B$ as vertices of $V$, a product state assigns to each vertex $v=(x_A, x_B)$ a coefficient of the factored form $a_{x_A}b_{x_B}$.

\paragraph{Signed cycle covers}
Let $N\coloneqq 2^k$.
For each $p \in [N]$, we specify a permutation $\pi_p: V \to V$ together with a sign function $\sigma_p: V \to \{\pm 1\}$.
One should regard $\pi_p$ as a directed cycle cover of $V$: each vertex $v$ has a directed edge from $v$ to $\pi_p(v)$, carrying sign $\sigma_p(v)$.\footnote{Note that we allow self-loops and parallel edges.}

For each \(p\), define the corresponding operator \(U_p\) by
\[
U_p\ket v
=
\sigma_p(v)\ket{\pi_p(v)}.
\]
Since \(\pi_p\) is a permutation and $\sigma_p(v) \in \{\pm 1\}$, $U_p$ is unitary, but it need not be Hermitian. 
We therefore define
\[
\Gamma_p \coloneqq \frac{U_p + U_p^\dagger}{2}.
\]
The operator $\Gamma_p$ is the (normalized) signed adjacency matrix of the undirected multigraph obtained by forgetting the orientation of the cycle cover specified by $\pi_p$.
In particular, a directed edge
\[
v\xrightarrow{\ \sigma_p(v)\ }\pi_p(v)
\]
contributes
\[
\frac{\sigma_p(v)}{2} \left( \ket{\pi_p(v)}\!\bra v + \ket v\!\bra{\pi_p(v)} \right)
\]
to $\Gamma_p$.
Associated with the $p$-th cycle cover is the constraint Hamiltonian
\[
H_p \coloneqq \frac{\id-\Gamma_p}{2} = \frac{(\id-U_p)^\dagger(\id-U_p)}{4}.
\]
Note that $H_p \succeq 0$ and the zero-energy states of $H_p$ are exactly the $+1$ eigenvectors of $U_p$.
Although the Hamiltonian depends only on the \emph{undirected} signed graph, it is convenient to specify each cycle cover by a permutation for two reasons. First, this will give a simple oracle interface in which one may query either the successor or predecessor of a vertex. Second, the hard instances used in our lower bound will be constructed from random permutations.

Our Hamiltonian is obtained by averaging these constraints over all $N$ cycle covers:
\begin{equation}\label{eq:signed-h}
H
\coloneqq
\frac{1}{N}\sum_{p\in[N]}H_p.
\end{equation}
Defining
\[
\Gamma \coloneqq \frac{1}{N}\sum_{p\in[N]}\Gamma_p,
\]
we have
\[
H=\frac{\id-\Gamma}{2}.
\]
Graph-theoretically, after forgetting orientations, each of the \(N\) cycle covers induces a signed \(2\)-regular multigraph on \(V\), and their union forms a signed \(2N\)-regular multigraph. The operator \(\Gamma\) is precisely the normalized signed adjacency matrix of this multigraph.

\paragraph{The promise problem}
We can now state the promise problem that underlies our separation.

\begin{problem}[Separable Signed-Graph Hamiltonian]
\label[problem]{problem:sign-matching}

Given oracle access to the signed permutations
\(\{(\pi_p,\sigma_p)\}_{p\in[N]}\), define \(H\) by
\cref{eq:signed-h}. Decide whether
\[
\begin{array}{ll}
\textup{YES case:}
&
\lambda_{\mathrm{sep}}(H)\leq \dfrac{3}{16},
\\[8pt]
\textup{NO case:}
&
\lambda_{\mathrm{sep}}(H)\geq \dfrac14,
\end{array}
\]
promised that one of these is the case.\footnote{
Because \(H=(\id-\Gamma)/2\), this problem can equivalently be phrased
in terms of deciding whether
\(\max_{\norm{\alpha}=\norm{\beta}=1}
\bra{\alpha,\beta}\Gamma\ket{\alpha,\beta}\)
is at least $5/8$ or at most $1/2$. 
}
\end{problem}

The minimum separable energy in \cref{problem:sign-matching} has a simple interpretation in terms of satisfying the signed edge constraints. 
For a product state
\[
\ket{\alpha}_A\otimes\ket{\beta}_B
= \sum_{(x_A,x_B)\in V}
a_{x_A}b_{x_B}\ket{x_A}_A\ket{x_B}_B,
\]
suppose $\pi_p(x_A,x_B)=(y_A,y_B)$.
Then
\[
\bra{\alpha,\beta}H\ket{\alpha,\beta}
=
\frac{1}{4N}
\sum_{p\in[N]}
\sum_{(x_A,x_B)\in V}
\abs{
a_{x_A}b_{x_B}
-
\sigma_p(x_A,x_B)a_{y_A}b_{y_B}
}^2.
\]
Thus, a positive edge asks the amplitudes at its endpoints to agree, while a negative edge asks them to differ by a sign.

Consequently, \(\lambda_{\rm sep}(H)\) asks how well one can simultaneously satisfy the signed edge constraints using a product-state labeling of the vertices $(x_A,x_B)\longmapsto a_{x_A}b_{x_B}$.
The YES case contains such a labeling with small average squared violation, whereas in the NO case every product-state labeling incurs noticeably larger violation.
A single $\QMA$ witness on $\calH_A\otimes\calH_B$ may be entangled, and hence its coefficients need not factor as $a_{x_A}b_{x_B}$. 
Thus, the natural test on this witness only certifies low energy over arbitrary bipartite states rather than low \emph{separable} energy.
The challenge in proving the separation is to show that no low-query $\QMA$ verifier can circumvent this distinction by using a more sophisticated verification procedure.

\paragraph{Oracle access}
For \(b_p(v)\in\{0,1\}\) defined by $\sigma_p(v)=(-1)^{b_p(v)}$,
define
\begin{equation}\label{eq:oracle}
O(p,z,v)
\coloneqq
\begin{cases}
\bigl(\pi_p(v),b_p(v)\bigr), & z=0,\\[4pt]
\bigl(\pi_p^{-1}(v),b_p(\pi_p^{-1}(v))\bigr), & z=1,
\end{cases},
\end{equation}
where $z \in \{0,1\}$ specifies whether we apply the forward or inverse permutation.

Arthur is given quantum query access to \(O\) through the usual reversible implementation 
\[
U_O\ket{p,z,v,w}
=
\ket{p,z,v,w\oplus O(p,z,v)}.
\]
Thus, a query specifies a permutation index \(p\), a direction $z$, and a vertex \(v\). It returns either the successor or predecessor of $v$ in the $p$-th cycle cover, together with the sign of the corresponding edge. 
We count one application of \(U_O\) as one query.\footnote{For the final classical-oracle separation, we encode the \(2\ell+1\)-bit output of \(O\) into a Boolean oracle in the standard way. The two query models are polynomially equivalent, and we defer this encoding to the diagonalization argument.}

\paragraph{The $\QMAT$ upper bound}
Known upper bounds do not seem to apply to our problem, but the containment in $\QMAT$ is not hard. We conclude this section with the details. 

\begin{proposition}\label[proposition]{prop:qma2-upper}
The Separable Signed-Graph Hamiltonian problem (\cref{problem:sign-matching}) is solvable in \(\QMAT\) using one query to the oracle in \cref{eq:oracle}.
\end{proposition}
\begin{proof}
    The two Merlins send the $\ell$-qubit states $\ket{\alpha}_A$ and $\ket{\beta}_B$, and let $\ket{\psi} \coloneqq \ket{\alpha}_A\ket{\beta}_B$. 
    Write $\ket{\psi} = \sum_{v \in V} \psi_v \ket{v}$.
    To make his single query, Arthur needs to initialize his query register $\ket{p,z,v,w}$. 
    Arthur samples from $[N]$ uniformly at random and stores the sample in the $\ket{p}$ register. 
    He prepares the direction qubit $\ket{z}$ in the $\ket{+}$ state. 
    The vertex register contains $\ket{\psi}$, and the answer register is initialized to $\ket{0^{2\ell + 1}}$. 

    Ignoring the $\ket{p}$ register, after the query the state is 
\[
\frac{1}{\sqrt2}
\sum_{v\in V}
\psi_v
\left(
\ket0\ket v\ket{\pi_p(v)}\ket{b_p(v)}
+
\ket1\ket v\ket{\pi_p^{-1}(v)}
\ket{b_p(\pi_p^{-1}(v))}
\right).
\]
Arthur then applies a $\mathrm{CZ}$ gate between the direction qubit and the sign-bit register. 
This multiplies the $\ket{1}$ branch by the phase $(-1)^{b_p(\pi_p^{-1}(v))} = \sigma_p(\pi_p^{-1}(v))$.
Next, Arthur applies a controlled-$\mathrm{SWAP}$ between the vertex and neighbor registers, with the direction qubit as the control.
Thus the \(\ket1\) branch becomes
\[
\sum_{v\in V}
\psi_v
(-1)^{b_p(\pi_p^{-1}(v))}
\ket1
\ket{\pi_p^{-1}(v)}
\ket v
\ket{b_p(\pi_p^{-1}(v))}.
\]

Now set $u=\pi_p^{-1}(v)$, so that $v=\pi_p(u)$.
Since $(-1)^{b_p(u)}=\sigma_p(u)$,
we may rewrite the \(\ket1\) branch as
\[
\sum_{u\in V}
\sigma_p(u)\psi_{\pi_p(u)}
\ket1
\ket u
\ket{\pi_p(u)}
\ket{b_p(u)}.
\]
Renaming \(u\) back to \(v\), the full state is therefore
\[
\frac{1}{\sqrt2}
\sum_{v\in V}
\left(
\psi_v\ket0
+
\sigma_p(v)\psi_{\pi_p(v)}\ket1
\right)
\ket v
\ket{\pi_p(v)}
\ket{b_p(v)}.
\]
The important point is that the two branches now carry the same edge label $\ket v\ket{\pi_p(v)}\ket{b_p(v)}$.

Arthur measures the direction qubit in the
\(\{\ket+,\ket-\}\) basis and rejects on outcome \(\ket-\).
For each \(v\), the amplitude on \(\ket-\) is
\[
\frac12
\left(
\psi_v-\sigma_p(v)\psi_{\pi_p(v)}
\right).
\]
Moreover, the retained edge labels corresponding to distinct \(v\) are orthogonal. Hence
\[
\Pr[\mathrm{reject}\mid p]
=
\frac14
\sum_{v\in V}
\abs{
\psi_v-\sigma_p(v)\psi_{\pi_p(v)}
}^2.
\]
By the definition of \(H_p\),
\[
\Pr[\mathrm{reject}\mid p]
=
\bra\psi H_p\ket\psi.
\]
Averaging over the uniformly random choice of \(p\) gives
\[
\Pr[\mathrm{reject}]
=
\frac1N\sum_{p\in[N]}
\bra\psi H_p\ket\psi
=
\bra\psi H\ket\psi,
\]
and therefore
\[
\Pr[\mathrm{accept}]
=
1-\bra\psi H\ket\psi.
\]

In a YES instance, there exists a product state satisfying $\lambda_{\rm sep}(H) \le \tfrac{3}{16}$,
and hence Arthur accepts with probability at least $\tfrac{13}{16}$.
In a NO instance, every product state satisfies
$\lambda_{\rm sep}(H) \ge \tfrac14$,
so every pair of unentangled witnesses is accepted with probability at most $\tfrac34$.

We note that using the amplification theorem of Harrow and Montanaro~\cite{HM13}, for any polynomial $r(n)$, this constant-gap protocol can be amplified to achieve completeness $1-\exp(-r(n))$ and soundness $\exp(-r(n))$, with polynomial overhead in the proof length, running time, and number of oracle queries. 
\end{proof}

\section{The hard YES and NO distributions}
\label{sec:construction}

We now construct two distributions, $\calD_+$ and $\calD_-$, over instances of Separable Signed-Graph Hamiltonian.
We show that $\lambda_{\rm sep}(H) = 0$ for every $H \sim \calD_+$, whereas with high probability over $H \sim \calD_-$, we have $\lambda_{\rm sep}(H) \ge 1/4$;
\cref{sec:qma-lower-bound} shows that no polynomial-query $\QMA$ verifier can distinguish the two families.

We begin by establishing notation and conventions common to both distributions.
Recall that $N = 2^k$, and assume $k \ge 2$.
Let $\calX_A$ and $\calX_B$ be two disjoint sets, each of size $N^3$. 
In the notation of \cref{sec:theproblem}, where $|\calX_A| = |\calX_B| = 2^{\ell}$, we have that $\ell = 3k$.

For each $p \in [N]$, both distributions will specify a permutation $\pi_p: V \to V$ and a sign $\sigma_p(x) \in \{\pm 1\}$ on each directed edge $x \rightarrow \pi_p(x)$.
These determine the signed cycle covers, and hence the Hamiltonian, as in \cref{sec:theproblem}. 
Both the YES and NO distributions use a common random product labeling of the vertices, from which their edge signs will be defined.
Independently, choose sign functions $h_A: \calX_A \to \{\pm 1\}$ and $h_B: \calX_B \to \{\pm 1\}$ uniformly at random, and define 
\[
\chi(x) \coloneqq h_A(x_A) h_B(x_B),
\]
for all $x = (x_A, x_B) \in V$.
Thus, $\chi: V \to \{\pm 1\}$ factors across the two coordinates of the vertex space.

\subsection{The hard YES distribution}
\label{subsec:hard-yes}

We describe the distribution $\calD_+$ over YES instances of Separable Signed-Graph Hamiltonian.
For each $p\in[N]$, independently choose a uniformly random permutation $\pi_p: V \to V$. 
For every directed edge $x \to \pi_p(x)$, define its sign by 
\[
\sigma_p(x) \coloneqq \chi(x) \chi(\pi_p(x)).
\]
Let $\calD_+$ denote the resulting distribution over signed-permutation instances.

The product labeling $\chi$ satisfies every edge constraint exactly, so every instance in the support of $\calD_+$ is a YES instance.

\begin{proposition}
\label[proposition]{clm:hard_yes:yes}
Every instance in the support of $\calD_+$ satisfies 
\[
\lambda_{\rm sep}(H) = 0.
\]
\end{proposition}
\begin{proof}
Define
\[
\ket{h_A}
\coloneqq
\frac{1}{\sqrt{N^3}}
\sum_{x_A\in\calX_A}
h_A(x_A)\ket{x_A},
\qquad
\ket{h_B}
\coloneqq
\frac{1}{\sqrt{N^3}}
\sum_{x_B\in\calX_B}
h_B(x_B)\ket{x_B}.
\]
Then, we have
\[
\ket{h_A}\ket{h_B}
=
\frac{1}{N^3}\sum_{x\in V}\chi(x)\ket x.
\]

For every \(p\in[N]\),
\[
\sigma_p(x)\chi(x)
=
\chi(x)\chi(\pi_p(x))\chi(x)
=
\chi(\pi_p(x)),
\]
and hence $U_p\ket{h_A,h_B}
=
\ket{h_A,h_B}$.
Therefore every \(H_p\) annihilates this product state, and so $\lambda_{\rm sep}(H)=0$, as claimed.
\end{proof}

\subsection{The hard NO distribution}
\label{subsec:hard-no}

We now define the distribution $\calD_-$.
With high probability, its instances are \emph{frustration-free}, yet every product state has energy at least $1/4$ (\cref{proposition_goodNO}).

\paragraph{Hidden Labels and Blocks}
To define the NO distribution, we first describe a partition of the vertex set $V$ into labeled blocks. 
These labels are not available through oracle access; they exist only ``in our mind'' as hidden parameters used to generate the instance. 
Let $[N] \coloneqq \{1,2,\dots,N\}$ be our underlying set of labels,
and let $\calC \coloneqq [N]^2$ denote the set of ordered pairs of labels, including pairs $(a,a)$ with a repeated label.

The basic building block of our partition is a four-element set $Q \subseteq [N]$, together with its $24$ orderings.

We next assign each element of $\calX_A$ and $\calX_B$ a pair of labels. 
More precisely, 
independently for every $x_A \in \calX_A$ and $x_B \in \calX_B$, choose $c_A(x_A) \in \calC$ and $c_B(x_B) \in \calC$ uniformly at random. Let $
\theta:=(c_A,c_B)
$ denote the complete pair-label assignment.
Thus, every vertex $x = (x_A, x_B) \in V$ inherits four labels $c(x) \coloneqq (c_A(x_A), c_B(x_B)) \in [N]^4$, two from each coordinate. Note that these four labels need not be distinct.

If the four labels of $x$ are distinct, they form a four-element set $Q \subseteq [N]$, and $c(x)$ is one of the $24$ orderings of $Q$.
For a four-element set $Q \subseteq [N]$, define 
\[
\calB_Q
\coloneqq
\left\{
x\in V:
c(x)\text{ is an ordering of }Q
\right\}.
\]
All remaining vertices, whose four labels are not all distinct, are collected into a single remainder block,
\[
\calB_{\mathrm{remainder}}
\coloneqq
\left\{
x\in V:
\text{the entries of }c(x)\text{ are not all distinct}
\right\}.
\]
The sets 
$\{\calB_Q : Q\in\tbinom{[N]}4\}$ and 
$\calB_{\mathrm{remainder}}$
therefore form a partition of \(V\).
We write $\theta(x)$ for the block containing vertex $x$. 
\paragraph{Permutations} We can now specify the permutations and signs associated with the NO distribution.
For each $p \in [N]$, independently choose a uniformly random permutation of every nonempty block $\calB_Q$, as well as a uniformly random permutation of $\calB_{\rm remainder}$. 
Taking the union of these blockwise permutations defines a permutation $\pi_p: V \to V$.
Thus, every $\pi_p$ preserves the hidden block decomposition: if $x \in \calB_Q$, then $\pi_p(x) \in \calB_Q$, and similarly $\pi_p(\calB_{\rm remainder}) = \calB_{\rm remainder}$.
Within a block $\calB_Q$, the permutation is otherwise completely random. In particular, it may map a vertex whose labels form one ordering of $Q$ to a vertex with any of the other $23$ orderings of $Q$.

\paragraph{Signs} We now explain how signs are assigned to the edges. 
We first define sign labels $s_{\theta}(x,y)$ for 2 vertices $x,y$ lying in the same block. 
These will be used to define the final edge signs. 
Note that since every $\pi_p$ preserves the block decomposition, these are the only values of $s$ that will be used.
We set $s_\theta{}(x,y) = -1$ for $x,y \in \calB_{\rm remainder}$.

For a vertex $x$ in a four-label block, let $\varepsilon(x)$ be the sign of the permutation that sorts $c(x)$.\footnote{For example, if $c(x)=(1,3,2,4)$, then sorting $c(x)$ takes a single transposition, so $\varepsilon(x)=-1$.}
For $x,y \in \calB_Q$, we set 
\[
s_\theta(x,y) \coloneqq \varepsilon(x)\varepsilon(y).
\]

Finally, writing $y = \pi_p(x)$, the actual edge sign is 
\[
\sigma_p(x) \coloneqq \chi(x) \chi(y) s_{\theta}(x,y).
\]
Together with the blockwise random permutations specified earlier, this defines $\calD_-$.

\paragraph{Zero-energy states}
The NO instances remain exactly satisfiable. Fix a nonempty block \(\calB_Q\), and define
\[
\phi_Q(x)\coloneqq \chi(x)\varepsilon(x).
\]
For an edge \(x\to y=\pi_p(x)\) inside this block, our definition gives
\[
\sigma_p(x)=\phi_Q(x)\phi_Q(y).
\]
Hence $\phi_Q(x)=\sigma_p(x)\phi_Q(y)$,
so the state whose amplitudes are proportional to \(\phi_Q(x)\) on \(\calB_Q\), and zero elsewhere, satisfies every edge constraint exactly. 
We write $\ket{\psi_Q}$ for this state, normalized.

Write $Q=\{q_1<q_2<q_3<q_4\}$.
The normalized totally antisymmetric state on these four labels is
\begin{equation}
\label{eq:ideal-state}
\ket{\mathrm{Alt}_Q} \coloneqq \frac1{\sqrt{24}} \sum_{\pi\in S_4} \sgn(\pi) \ket{q_{\pi(1)},q_{\pi(2)}}_A\ket{q_{\pi(3)},q_{\pi(4)}}_B,
\end{equation}
written in the orthonormal label bases $\{\ket{I}_P : I \in \calC\}$, $P\in\{A,B\}$.
Thus every ordering of the four labels appears once, with a positive or negative sign according to the parity of the permutation producing that ordering; for the ordering $c(x)$, this sign is $\varepsilon(x)$.
As \(Q\) ranges over \(\binom{[N]}4\), the states \(\ket{\mathrm{Alt}_Q}\) form the standard orthonormal basis of the four-register antisymmetric subspace, so
\[
P_N \coloneqq \sum_{Q\in\binom{[N]}4} \ket{\mathrm{Alt}_Q}\!\bra{\mathrm{Alt}_Q}
\]
is the projector onto this subspace.
By the product-overlap bound for the four-register antisymmetric projector (see e.g.~\cite[Lemma~4.1]{BGH+26}),
\[
h_{\rm Sep}(P_N)\le\frac13.
\]
In other words, no product state can have squared overlap greater than $1/3$ with the antisymmetric subspace.

\paragraph{$\calD_-$ has large minimum separable energy}

Despite having zero-energy states, we now show that, with high probability, an instance drawn from $\calD_-$ has no low-energy product state. 

For \(P\in\{A,B\}\) and \(I\in\mathcal C\), define
\[
C_I^P(\theta):=\{x_P\in\mathcal X_P:c_P(x_P)=I\}.
\]
Let \(K:=|\mathcal C|=N^2\), \(\mu:=N^3/K=N\), and
\(\varepsilon_0:=1/40\). Define the set of balanced labelings
\begin{equation}
\Theta_{\mathrm{size}}
:=\left\{\theta:
(1-\varepsilon_0)\mu\le |C_I^P(\theta)|\le(1+\varepsilon_0)\mu
\text{ for every }P\in\{A,B\},\ I\in\mathcal C
\right\}.
\label{eq:class-size-condition}
\end{equation}

\begin{lemma}[Class sizes]
\label{lem:class-sizes}
The labeling sampled in the NO construction satisfies
\[
\Pr_\theta[\theta\notin\Theta_{\mathrm{size}}]
\le 4K\exp(-\varepsilon_0^2\mu/3).
\]
For every \(\theta\in\Theta_{\mathrm{size}}\), every nonempty block of
its partition has size at least \((1-\varepsilon_0)^2\mu^2\).
\end{lemma}
\begin{proof}
Each \(|C_I^P(\theta)|\) is binomial with mean \(\mu\). By
\cref{lem:chernoff}, the probability that it violates the class-size
bounds is at most \(2\exp(-\varepsilon_0^2\mu/3)\). A union bound over
the \(2K\) classes proves the probability estimate.

Each four-label block contains \(C_I^A(\theta)\times C_J^B(\theta)\)
whenever the concatenation of \(I\) and \(J\) is an ordering of its labels. The remainder block contains
\(C_I^A(\theta)\times C_I^B(\theta)\). On
\(\Theta_{\mathrm{size}}\), each such product has size at least
\((1-\varepsilon_0)^2\mu^2\).
\end{proof}

\begin{proposition}
\label{proposition_goodNO}
With probability $1-\exp(-\Omega(N))$ over $H\sim\calD_-$, we have
$\lambda_{\rm sep}(H)\ge1/4$.
\end{proposition}

\begin{proof}
By \cref{lem:class-sizes}, \(\Theta_{\mathrm{size}}\) fails with
probability at most \(\exp(-\Omega(N))\). Fix
\(\theta\in\Theta_{\mathrm{size}}\) and any coordinate signs \(h_A,h_B\).
The only remaining randomness is in the permutations $\pi_p$. 
Let $\overline{\Gamma} \coloneqq \E[\Gamma\mid c_A,c_B,h_A,h_B]$ denote the signed adjacency matrix averaged over these random permutations.
For the remainder of the proof, all expectations are taken over the remaining randomness in the permutations, conditioned on the fixing of $c_A, c_B, h_A,$ and $h_B$; we suppress this from the notation.

Recall the orthonormal label basis states $\{\ket{I}_P : I \in \calC\}$ and the states $\ket{\mathrm{Alt}_Q}$ from \cref{eq:ideal-state}.
To account for the unequal class sizes in the actual construction, define local maps
\[
R_P\ket I
\coloneqq
\frac1{\sqrt\mu}
\sum_{x_P:c_P(x_P)=I}h_P(x_P)\ket{x_P},
\qquad P\in\{A,B\},
\]
and define $R\coloneqq R_A\otimes R_B$.
Since these maps have orthogonal columns and 
\[
\norm{R(\ket{I}_A\ket{J}_B)}^2 = \frac{\abs{C_I^A}\abs{C_J^B}}{\mu^2} \in [(1-\varepsilon_0)^2, (1+\varepsilon_0)^2],
\]
we have that $\norm{R} \le 1 + \varepsilon_0$ and $\norm{R\ket{\mathrm{Alt}_Q}}\ge1-\varepsilon_0$ for every $Q$.

Recall that for each nonempty block $\calB_Q$, we defined the normalized zero-energy state
\begin{equation}
\label{eq:psi-q}
\ket{\psi_Q}
= \frac{1}{\sqrt{\abs{\calB_Q}}} \sum_{x\in\calB_Q}
\chi(x)\varepsilon(x)\ket{x}.
\end{equation}
Grouping this sum by the labels of the two coordinates shows that
\begin{equation}
\label{eq:psi-q-R}
\ket{\psi_Q}
=
\frac{R\ket{\mathrm{Alt}_Q}}
{\norm{R\ket{\mathrm{Alt}_Q}}}.
\end{equation}

We next compute the averaged signed adjacency matrix $\overline\Gamma$. 
Fix a four-label block $\calB_Q$ and a vertex $x \in \calB_Q$. 
Recall that 
\[
U_p \ket{x} = \sigma_p(x) \ket{\pi_p(x)},
\]
and that on $\calB_Q$, 
\[
\sigma_p(x) = \phi_Q(x) \phi_Q(\pi_p(x)).
\]
Hence
\[
U_p\ket{x} = \phi_Q(x)\phi_Q(\pi_p(x)) \ket{\pi_p(x)}.
\]
For fixed $x$, the vertex $\pi_p(x)$ is uniformly distributed over $\calB_Q$. 
Therefore, averaging over the random permutation $\pi_p$ gives 
\[
\E_{\pi_p}[U_p\ket{x}] = \frac{\phi_Q(x)}{\abs{\calB_Q}} \sum_{y\in\calB_Q} \phi_Q(y)\ket{y} = \frac{\phi_Q(x)}{\sqrt{\abs{\calB_Q}}} \ket{\psi_Q}, 
\]
where the last line uses \cref{eq:psi-q}.
We also have
\[
\braket{\psi_Q}{x} = \frac{\phi_Q(x)}{\sqrt{\abs{\calB_Q}}},
\]
and hence
\[
\E_{\pi_p}[U_p\ket{x}] = \ket{\psi_Q}\!\braket{\psi_Q}{x}.
\]
Since this holds for every $x\in\calB_Q$,
\[
\E_{\pi_p}[U_p]\big|_{\calB_Q} = \ket{\psi_Q}\!\bra{\psi_Q}.
\]
Recall that
\[
\Gamma_p = \frac{U_p+U_p^\dagger}{2}.
\]
Since $\ket{\psi_Q}\bra{\psi_Q}$ is Hermitian, it follows that
\[
\E_{\pi_p}[\Gamma_p]\big|_{\calB_Q} = \ket{\psi_Q}\!\bra{\psi_Q}.
\]
On the remainder block, the sign rule contains an additional factor of $-1$. The same calculation therefore gives a negative semidefinite rank-one contribution. Discarding this contribution yields
\[
\overline\Gamma 
\preceq \sum_{Q\in\binom{[N]}4} \ket{\psi_Q}\!\bra{\psi_Q}. 
\]
Then, using \cref{eq:psi-q-R} and the fact that $\norm{R\ket{\mathrm{Alt}_Q}} \ge 1 -\varepsilon_0$ for all $Q$, we get
\[
\overline\Gamma  \preceq 
\sum_Q \frac{ R\ket{\mathrm{Alt}_Q}\!\bra{\mathrm{Alt}_Q}R^\dagger }{ \norm{R\ket{\mathrm{Alt}_Q}}^2 }
\preceq \frac1{(1-\varepsilon_0)^2} R \left( \sum_Q \ket{\mathrm{Alt}_Q}\!\bra{\mathrm{Alt}_Q} \right) R^\dagger.
\]
Recalling that
\[
P_N = \sum_Q \ket{\mathrm{Alt}_Q}\!\bra{\mathrm{Alt}_Q},
\]
we conclude that
\[
\overline\Gamma \preceq \frac{RP_NR^\dagger}{(1-\varepsilon_0)^2}.
\]

Observe that $R$ and $R^\dagger$ do not act across the bipartition and therefore preserve product states.
As explained earlier, $P_N$ is the projector onto the antisymmetric subspace, and thus $h_{\rm Sep}(P_N) \le 1/3$.
Therefore, for every normalized product state, 
\[
\bra{\alpha,\beta}\overline\Gamma\ket{\alpha,\beta}
\le
\frac{(1+\varepsilon_0)^2}{(1-\varepsilon_0)^2}h_{\rm Sep}(P_N)
\le
\frac13\left(\frac{41}{39}\right)^2
<\frac38.
\]

Finally, conditional on the hidden labels and coordinate signs, the
$\Gamma_p$ are independent operators, because they depend on different permutations that were chosen independently and uniformly at random. 
Recall also that the $\Gamma_p$ are $N^6 \times N^6$ Hermitian operators with $\norm{\Gamma_p} \le 1$. 
Applying \cref{thm:matrix-hoeffding} to
$X_p=\Gamma_p-\overline\Gamma$, for which $\norm{X_p}\le2$, gives
\[
\Pr\left[
\norm{\Gamma-\overline\Gamma}>\frac18
\;\middle|\;c_A,c_B,h_A,h_B
\right]
\le 2N^6\exp(-N/2048)
=\exp(-\Omega(N)).
\]
On the event that all pair-label classes satisfy the size bounds and that $\norm{\Gamma - \overline\Gamma} \le 1/8$, every normalized product state satisfies
\[
\bra{\alpha,\beta}H\ket{\alpha,\beta}
\ge
\frac12\left(1-\frac38-\frac18\right)
=\frac14.
\]
Taking the minimum over product states proves the claim.
\end{proof}

\input{lower-bound}

\printbibliography
\appendix
\input{appendix}

\end{document}

%% file: lower-bound.tex
\section{The \texorpdfstring{\(\QMA\)}{QMA} lower bound}
\label{sec:qma-lower-bound}

We use the distributions from \cref{sec:construction}, with \(N=2^k\) for input length \(k\). A NO parameter choice is \emph{good} if it
satisfies the class-size condition in \eqref{eq:class-size-condition}
and the graph-concentration condition in the proof
of \cref{proposition_goodNO}. That proposition proves that a good choice
defines a NO instance and that the probability of a bad choice is
\(\exp(-\Omega(N))\)

\begin{theorem}
\label{thm:query-lower-bound}
For every polynomial \(p\), the following holds for all sufficiently
large \(k\). No \(\QMA\) verifier with at most \(p(k)\) witness qubits
and \(p(k)\) oracle queries has completeness \(2/3\) on every instance
in \(\operatorname{supp}\mathcal D_+\) and soundness \(1/3\) on every
good instance from \(\mathcal D_-\).
\end{theorem}
The proof has three parts. In section \ref{subsec:recording}, we purify the oracle distributions and introduce the
compressed representation of the oracle register, in which the information learned by a few-query
verifier is a superposition of small records. In section \ref{subsec:transfer-certificate}, we state our main technical tool,
the transfer theorem (Theorem \ref{thm:transfer-certificate}), and use it to prove Theorem \ref{thm:query-lower-bound}. The remaining sections construct the transfer map and prove
the transfer theorem: Section \ref{subsec:database-transfer} makes the compressed records belonging to YES distributions compatible with records belonging to the NO distribution. Section  \ref{sec:remaining-completion-estimates} shows that this process introduces only small errors and remains consistent as new query information is revealed. Section  \ref{sec:generic-transfer} translates these guarantees from compressed records back to actual oracle states, proving the transfer theorem.

\input{oracle-records}

\subsection{The transfer theorem and the witness argument}
\label{subsec:transfer-certificate}

We now state the transfer theorem, which converts states of the YES oracle register into states of
the NO oracle register, and use it to
prove \cref{thm:query-lower-bound}. 

Let \(\Omega_-^{\mathrm{good}}\) be the set of good NO parameters
defined at the start of this section  . Define
\[
 P_{\mathrm{bad}}
 :=\sum_{\omega^-\in\Omega_-\setminus\Omega_-^{\mathrm{good}}}
 |\omega^-\rangle\langle\omega^-|.
\]
to be the projector onto bad NO parameters. It commutes with \(O_-\).

\begin{theorem}[Transfer between oracle-parameter spaces]
\label{thm:transfer-certificate}
For all input lengths $k$, there is a family of linear contractions
\[
 T=T_k:\mathcal H_+\longrightarrow\mathcal H_-.
\]
such that, for every polynomially bounded integer-valued \(R=R(k)\ge0\),
there are negligible functions \(\eta_R^{\text{norm}}=\eta_R^{\text{norm}}(k)\),
\(\varepsilon^{\text{query}}_R=\varepsilon^{\text{query}}_R(k)\), and \(\gamma^{\text{bad}}_R=\gamma^{\text{bad}}_R(k)\)
such that for all sufficiently large \(k\) and
every finite-dimensional  register \(\mathcal E\), 
every \(\phi\in\mathcal H_+\otimes\mathcal E\) with
\(J_+\phi\in\mathsf{Rec}_{+,\le R}\otimes\mathcal E\) satisfies
\begin{align}
 (1-\eta_R^{\text{norm}})\|\phi\|^2
 &\le\|T\phi\|^2\le\|\phi\|^2,
 \label{eq:transfer-norm}\\
 \|(O_-T-TO_+)\phi\|
 &\le\varepsilon^{\text{query}}_R\|\phi\|,
 \label{eq:transfer-query}\\
 \|P_{\mathrm{bad}}T\phi\|^2
 &\le\gamma^{\text{bad}}_R\|T\phi\|^2.
 \label{eq:transfer-bad}
\end{align}
\end{theorem}

\begin{proof}[Proof of \cref{thm:query-lower-bound}]
Suppose that the QMA verifier has witness length \(m\), makes at most
\(q\) oracle queries, and has completeness \(2/3\) and soundness \(1/3\).
For a fixed oracle \(\omega\), let \(A(\omega)\) be its acceptance
operator such that \(0\preceq A(\omega)\preceq\id\), and the acceptance
probability on a normalized witness \(\ket w\) is
\(\langle w|A(\omega)|w\rangle\). Define the controlled acceptance operator 
\[
 A_\pm:=\sum_{\omega_\pm}\ket{\omega_\pm}\!\bra{\omega_\pm}
                         \otimes A(\omega_\pm).
\]
The operator \(A(\omega)\) can be implemented as a postselected linear map by first running the verifier,
projecting onto acceptance, reversing the verifier, and projecting
the work register onto its initial state. Hence it uses at most \(2q\)
oracle queries. 
We first show that each application of $A_{+}$ shrinks the norm of a state with a small record size by a factor of about $1/3$. Next, we will show that this indeed contradicts completeness of the verifier. 
\begin{claim}
\label{claim:acceptance-bound}
Let \(B=\operatorname{poly}(k)\) be an integer with \(B\ge12q\).
For every vector \(\phi\) with
\(J_+\phi\in\mathsf{Rec}_{+,\le B-12q}\otimes\mathcal E\),
we have, for all sufficiently large \(k\),
\[
 \|A_+\phi\|\le b_B\|\phi\|,
 \qquad
 b_B:=\frac{\frac13\sqrt{1+8\gamma^{\text{bad}}_B}+2q\varepsilon^{\text{query}}_B}
                 {\sqrt{1-\eta^{\text{norm}}_B}}
       =\frac13+o(1).
\]
\end{claim}
\begin{proof}

By \cref{prop:compressed-records}, the encoding under \(J_+\) of
every intermediate YES state has record size at most \(B\).
Between oracle queries the verifier acts only on its own registers,
while \(T\) acts only on the oracle register. Each intermediate
operation is a contraction. Applying Equation \eqref{eq:transfer-query} at each
query and adding the errors gives us
\[
 \|(A_-T-TA_+)\phi\|\le2q\varepsilon^{\text{query}}_B\|\phi\|.
\]

Next, soundness implies that for every good NO parameter choice,
\(A(\omega^-)\preceq\frac13\id\), and furthermore, 
for every NO parameter choice, \(A(\omega^-)\preceq\id\).
Equation~\eqref{eq:transfer-bad} therefore gives
\[
 \begin{aligned}
 \|A_-T\phi\|^2
 &\le\frac19\|(\id-P_{\mathrm{bad}})T\phi\|^2
       +\|P_{\mathrm{bad}}T\phi\|^2\\
 &=\frac19\|T\phi\|^2
       +\frac89\|P_{\mathrm{bad}}T\phi\|^2\\
 &\le\frac{1+8\gamma^{\text{bad}}_B}{9}\|T\phi\|^2
 \le\frac{1+8\gamma^{\text{bad}}_B}{9}\|\phi\|^2.
 \end{aligned}
\]
Hence
\[
 \|A_-T\phi\|
 \le\frac13\sqrt{1+8\gamma^{\text{bad}}_B}\,\|\phi\|.
\]

Moreover, \(J_+A_+\phi\) has record size at most \(B\), so
\eqref{eq:transfer-norm} gives
\[
 \sqrt{1-\eta^{\text{norm}}_B}\,\|A_+\phi\|\le\|TA_+\phi\|.
\]
Combining these inequalities yields
\[
 \begin{aligned}
 \sqrt{1-\eta^{\text{norm}}_B}\,\|A_+\phi\|
 &\le\|A_-T\phi\|+2q\varepsilon^{\text{query}}_B\|\phi\|\\
 &\le\left(\frac13\sqrt{1+8\gamma^{\text{bad}}_B}
                  +2q\varepsilon^{\text{query}}_B\right)\|\phi\|.
 \end{aligned}
\]
Dividing by \(\sqrt{1-\eta^{\text{norm}}_B}\) proves the claim.
\end{proof}

\begin{claim}
\label{claim:completeness-contradiction}
 Claim \ref{claim:acceptance-bound} contradicts completeness for all
sufficiently large \(k\).
\end{claim}
\begin{proof}
Let \(\mathcal W\) be the \(m\)-qubit witness register and let
\(\mathcal R\) be an \(m\)-qubit reference register. Define
\[
 |\Phi_m\rangle
 :=2^{-m/2}\sum_{z\in\{0,1\}^m}
 |z\rangle_{\mathcal W}|z\rangle_{\mathcal R},
 \qquad
 \phi_0:=|\mathcal D_+\rangle|\Phi_m\rangle.
\]
Let \(r:=m+1\), \(B:=12qr\), and
\[
 \phi_j:=(A_+\otimes\id_{\mathcal R})^j\phi_0.
\]
 Since \(J_+|\mathcal D_+\rangle=|\emptyset\rangle\),
\cref{prop:compressed-records} gives
\[
 J_+\phi_j\in\mathsf{Rec}_{+,\le12qj}
                   \otimes\mathcal W\otimes\mathcal R.
\]
For \(j<r\), we have \(12qj\le B-12q\), so
Claim \ref{claim:acceptance-bound} gives
\[
 \|\phi_{j+1}\|\le b_B\|\phi_j\|.
\]
Since \(\|\phi_0\|=1\), it follows that
\(\|\phi_r\|^2\le b_B^{2r}\). on the other hand,  by the definition of \(\phi_0\),
\[
 \begin{aligned}
 \|\phi_r\|^2
 &=2^{-m}\mathbb E_{\omega^+\sim\mathcal D_+}
       \sum_{z\in\{0,1\}^m}\|A(\omega^+)^r|z\rangle\|^2\\
 &=2^{-m}\mathbb E_{\omega^+\sim\mathcal D_+}
       \operatorname{tr}\!\left(A(\omega^+)^{2r}\right).
 \end{aligned}
\]
Completeness implies \(\lambda_{\max}(A(\omega^+))\ge2/3\)
for every \(\omega^+\in\operatorname{supp}\mathcal D_+\).
Since \(A(\omega^+)\succeq0\),
\[
 \operatorname{tr}\!\left(A(\omega^+)^{2r}\right)
 \ge\left(\frac23\right)^{2r}.
\]
Therefore
\[
 2^{-m}\left(\frac23\right)^{2r}\le b_B^{2r}.
\]
Taking \(2r\)-th roots and using \(r=m+1\) gives
\[
 \frac23\le2^{m/(2r)}b_B\le\sqrt2\,b_B.
\]
But \(b_B=1/3+o(1)\), which is impossible for all sufficiently
large \(k\).
\end{proof}
Claims~\ref{claim:acceptance-bound} and~\ref{claim:completeness-contradiction}
prove the theorem.
\end{proof}

\input{database-completion}
\input{completion-estimates}
\input{oracle-transfer}

%% file: oracle-records.tex
\subsection{Oracle registers and compressed records}
\label{subsec:lower-bound-setup}
\label{subsec:recording}
Sampling an oracle from $\mathcal{D}_\pm$ and answering oracle queries can be modeled as sampling the parameters required to define the oracle, storing them in a register, and then answering oracle queries by evaluating a fixed operation controlled on the stored oracle parameters. A powerful technique first introduced by \cite{Zha19} is to purify this oracle sampling process. Instead of first
sampling an oracle from $\mathcal D_+$ or $\mathcal D_-$, we may store its
classical parameters coherently in an oracle register inaccessible to the verifier.  
We then change the representation of this register by a fixed isometry, which we
call compression.
 Since this operation is applied locally to the hidden register it is undetected by the verifier. We apply  the compression isometry of Carolan and Majenz
\cite{CM26} to the permutation components while encoding the
coordinate signs in the Fourier basis. The
hidden pair-labeling $\theta$ (for hard NO instances) stays the same. In this compressed representation, the oracle register is in a superposition of records,
which consist of partial permutation tables and sets of coordinates, and a verifier making few
queries only produces small records. We state below the
properties of this compressed representation that will be used in the lower
bound.  

\paragraph{Oracle registers.} A YES parameter and a NO parameter have the respective forms
\[
 \omega^+=((\pi_p)_{p\in[N]},h_A,h_B),
 \qquad
 \omega^-=((\pi_p)_{p\in[N]},h_A,h_B,\theta).
\]
Here \(\theta=(c_A,c_B)\) specifies the blocks from
\cref{subsec:hard-no}. We view \(\mathcal D_\pm\) as distributions on
parameters. Let \(\Omega_\pm:=\operatorname{supp}\mathcal D_\pm\)
and define
$
 \mathcal H_\pm:=\operatorname{span}\{\ket{\omega^\pm}:
                         \omega^\pm\in\Omega_\pm\}.$ The basis states are orthonormal.

The purified YES distribution is
$$
|\mathcal D_+\rangle
:=
\sum_{\omega^+\in\Omega_+}
\sqrt{\mathcal D_+(\omega^+)}
\,|\omega^+\rangle.
$$
For each parameter choice $\omega^\pm$, let $O_{\omega^\pm}$ denote
the corresponding oracle query from Equation~\eqref{eq:oracle}.
Define
$$
O_\pm
:=
\sum_{\omega^\pm\in\Omega_\pm}
|\omega^\pm\rangle\langle\omega^\pm|
\otimes O_{\omega^\pm}.
$$
Thus $O_\pm$ applies the oracle specified by the parameter register.
\paragraph{Records.}
We now define the basis states used for the compressed representation.
A \textit{permutation database} is defined as a tuple
$$
D=(D_1,\ldots,D_L),
$$
where each $D_j$ is a partial bijection on some $X_j \subseteq V $ and $L \in \mathbb N$. Thus $D_j$
consists of entries
$
x\mapsto y
$
with no repeated input and no repeated output. We refer to each $D_j$ as a \textit{table}. Define
$
|D|:=\sum_{j=1}^L |D_j|.
$ A \textit{sign database} is defined as a subset $S \subseteq \mathcal X_A\sqcup\mathcal X_B$ of coordinates. These subsets will later be thought of as
indexing a Fourier basis vector of the register storing the sign functions \(h_A\) and \(h_B\).  
A \textit{record} is now defined as a tuple $(D,S)$ of a permutation database and a sign database.
Its size is
$
|D|+|S|
$. 

A YES permutation database is \(D=(D_p)_{p\in[N]}\), with each \(D_p\)
a partial bijection on \(V\). For fixed \(\theta\), a NO permutation
database is \(D=(D_{p,B})\), with each \(D_{p,B}\) a partial bijection on vertices in 
\(B\), indexed by \(p\in[N]\) and the nonempty blocks \(B\) of
\(\theta\). 

Let \(\mathsf{Rec}_+\) and \(\mathsf{Rec}_-(\theta)\) be spaces spanned by orthonormal
bases \(\{|D,S\rangle\}\) indexed by the respective databases and all
\(S\subseteq\mathcal X_A\sqcup\mathcal X_B\). Define
\[
\mathsf{Rec}_-
:=
\bigoplus_\theta
|\theta\rangle\otimes\mathsf{Rec}_-(\theta),
\]
where \(\theta\) ranges over all pair-labelings. A subscript \(\le R\)
denotes the subspace spanned by records satisfying \(|D|+|S|\le R\).
The empty YES record \(|\emptyset\rangle\) has every table empty and
\(S=\emptyset\).

Write \(\mathcal B_\theta(v)\) for the block containing \(v\). An
entry \((p;x\mapsto y)\) belongs to \(D\) if \(x\mapsto y\) belongs to
\(D_p\) in the YES case or \(D_{p,\mathcal B_\theta(x)}\) in the NO case.
Define the coordinate support
\[
\operatorname{coords}(D)
:=
\bigcup_{(p;x\mapsto y)\in D}
\{x_A,x_B,y_A,y_B\},
\]
where \(x=(x_A,x_B)\) and \(y=(y_A,y_B)\).

\paragraph{Compression and well formed records.}
We now map the oracle-parameter spaces into the record spaces defined
above. First we define compression of the registers holding the coordinate-sign parameters, $\ket{h_A, h_B}$. Let
\(\mathcal X:=\mathcal X_A\sqcup\mathcal X_B\). Write \(h\) for the
combined sign assignment: \(h(u)=h_A(u)\) when \(u\in\mathcal X_A\),
and \(h(u)=h_B(u)\) when \(u\in\mathcal X_B\). For a set
\(S\subseteq\mathcal X\), write \(h(S)\) for the product of its
coordinate signs:
\[
h(S):=\prod_{u\in S}h(u),
\qquad h(\emptyset):=1.
\]
Writing \(M:=|\mathcal X|\), the map to the record space applies the Fourier transform 
\[
|h_A,h_B\rangle
\longmapsto
2^{-M/2}\sum_{S\subseteq\mathcal X}h(S)|S\rangle.
\]
If we imagine $S$ written as a list of membership bits, this is simply the Hadamard transform on the \(M\) independent coordinate signs.
In order to compress the registers $\ket{(\pi_p)_{p \in [N]}}$ holding the permutation parameters,  we use recent techniques from the literature on compressed permutation oracles.
For a finite set \(X\), let
\[
 J_X:\operatorname{span}\{\ket\pi:\pi\in\operatorname{Sym}(X)\}
 \longrightarrow
 \operatorname{span}\{\ket I:I\text{ is a partial bijection on }X\}
\]
be the compression isometry of \cite[Theorem~4.2]{CM26}.
Define \(J_+:\mathcal H_+\longrightarrow\mathsf{Rec}_+\) by
\[
 \begin{aligned}
 &J_+\ket{(\pi_p)_{p\in[N]},h_A,h_B}\\
 &\quad:=
 \left(\bigotimes_{p\in[N]}J_V\ket{\pi_p}\right)
 \otimes
 \left(2^{-M/2}\sum_{S\subseteq\mathcal X}h(S)\ket S\right).
 \end{aligned}
\]
For NO parameters, \(\pi_p|_B:B\to B\) is the restriction of
\(\pi_p\) to a block \(B\) of \(\theta\).
Define \(J_-:\mathcal H_-\longrightarrow\mathsf{Rec}_-\) by
\[
 \begin{aligned}
 &J_-\ket{(\pi_p)_{p\in[N]},h_A,h_B,\theta}\\
 &\quad:=\ket\theta\otimes
 \left(\bigotimes_{p\in[N]}
       \bigotimes_{\substack{B\text{ a block of }\theta\\B\ne\varnothing}}
       J_B\ket{\pi_p|_B}\right)
 \otimes
 \left(2^{-M/2}\sum_{S\subseteq\mathcal X}h(S)\ket S\right).
 \end{aligned}
\]
A vector \(\psi\in\mathsf{Rec}_\pm\) is \emph{well formed} if
\(\psi\in\operatorname{im}J_\pm\).
\paragraph{Decompression.} Since $J_\pm$ is an isometry, an oracle query $O_\pm$ can be exactly represented in the compressed representation as $J_\pm O_\pm J^\dag_\pm$. However, we also define a query procedure that acts directly on compressed records and approximatesthe exactly represented query $J_\pm O_\pm J^\dag_\pm$. 
Let \(X\) be a finite set of size \(s\).
For a forward query at \(x\in X\), the \textit{decompression operator}
\(\mathrm{Dec}_x^\rightarrow\) exchanges
\[
 \ket I\quad\text{and}\quad
 \frac1{\sqrt{s-|I|}}\sum_{y\in X\setminus\operatorname{ran}(I)}
       \ket{I\cup\{x\mapsto y\}}
\]
for every partial bijection \(I\) on \(X\) with
\(x\notin\operatorname{dom}(I)\).
For an inverse query at \(y\in X\), the decompression operator
\(\mathrm{Dec}_y^\leftarrow\) exchanges
\[
 \ket I\quad\text{and}\quad
 \frac1{\sqrt{s-|I|}}\sum_{x\in X\setminus\operatorname{dom}(I)}
       \ket{I\cup\{x\mapsto y\}}
\]
for every partial bijection \(I\) on \(X\) with
\(y\notin\operatorname{ran}(I)\). \(\mathrm{Dec}_y^\leftarrow\)  and \(\mathrm{Dec}_x^\rightarrow\) 
swap the two specified vectors in each  of their  specified subspaces and act as identity on all  orthogonal vectors. Thus both operators are Hermitian unitaries.

An \textit{address} is a triple \(a=(p,z,v)\), with \(p\in[N]\),
\(z\in\{0,1\}\), and \(v\in V\).
It specifies a forward query when \(z=0\) and an inverse query when
\(z=1\). The address register has basis \(\ket a=\ket{p,z,v}\).
The answer register consists of a vertex register \(\ket w\) and a
sign qubit $\ket{t}$, which holds the sign of the queried edge.

The operator \(\mathrm{Dec}_{\pm,a}\) applies \(\mathrm{Dec}_v^\rightarrow\) for \(z=0\)
and \(\mathrm{Dec}_v^\leftarrow\) for \(z=1\).
It acts on the register for \(D_p\) in the YES case and for
\(D_{p,\mathcal B_\theta(v)}\) in the NO case.
Finally, we define the operator \(\mathrm{Dec}_{\pm}\) that applies decompression controlled by the address register as follows:
\[
 \mathrm{Dec}_{\pm}:=\sum_a \mathrm{Dec}_{\pm,a}\otimes\ket a\!\bra a.
\]

\paragraph{Queries on records.}
For an entry \(e=(p;x\mapsto y)\), let \(\partial e\) be the set of
coordinates that occur an odd number of times among
\(x_A,x_B,y_A,y_B\).
The sign \(\chi(x)\chi(y)\) equals \(h(\partial e)\).
We represent the sign qubit $\ket{t}$ in the fourier basis, 
\[
 \ket{\widehat t}:=H \ket{t},
 \qquad t\in\{0,1\}.
\]
XORing a bit \(b\) into this register multiplies introduces the phase
\((-1)^{bt}\). Thus the sign \(\chi(x)\chi(y)\) contributes the
phase \(h(\partial e)^t\).
In the NO case, the sign \(s_\theta(x,y)\) also contributes the
phase \(s_\theta(x,y)^t\).

Fix an address \(a=(p,z,v)\).
Suppose \(D\) contains \(e=(p;v\mapsto y)\) for \(z=0\), or
\(e=(p;x\mapsto v)\) for \(z=1\).
This entry is unique. Let \(v'_a=y\) in the first case and
\(v'_a=x\) in the second.
Define
\[
 \begin{aligned}
 Z_+\ket{D,S}\ket{a,w,\widehat t}
 &=\ket{D,S\triangle(t\partial e)}
                       \ket{a,w\oplus v'_a,\widehat t},\\
 Z_-\ket{\theta,D,S}\ket{a,w,\widehat t}
 &=s_\theta(x,y)^t\ket{\theta,D,S\triangle(t\partial e)}
                       \ket{a,w\oplus v'_a,\widehat t}.
 \end{aligned}
\]
Here \(0\partial e=\varnothing\), \(1\partial e=\partial e\),
and \(s_\theta(x,y)\) is the sign \(s(x,y)\) from
\cref{subsec:hard-no}. If no such entry exists, \(Z_\pm\) acts as
the identity. In other words, the permutation queries are answered by reading off an entry from the permutation database if it exists. The sign query is implemented likewise by an XOR based on the sign stored in the $h$ vectors; since \(h(S)h(\partial e)=h(S\triangle\partial e)\), multiplication
by \(h(\partial e)\) is equivalent to coherently sending \(\ket S\) to
\(\ket{S\triangle\partial e}\) in the Fourier basis.

The query operation on records is
\begin{equation}
 C_\pm:=\mathrm{Dec}_{\pm} Z_\pm \mathrm{Dec}_{\pm}.
 \label{eq:recorded-query}
\end{equation}
The operator \(Z_\pm\) permutes the orthonormal basis vectors and
multiplies them by signs, so it is unitary.
It also satisfies \(Z_\pm^2=\id\). Thus \(Z_\pm\), and hence
\(C_\pm\), is a Hermitian unitary.

For the parameter space \(\mathcal H_-\) and the record space \(\mathsf{Rec}_-\), let 

\[
P_{\mathrm{size}}
:=
\left(\sum_{\theta\in\Theta_{\mathrm{size}}}
|\theta\rangle\langle\theta|\right)\otimes I,
\]

where $\Theta_{\mathrm{size}}$ is defined in \eqref{eq:class-size-condition}.

By definition of $J_{-}$, we can conclude that 
\[
 J_-P_{\mathrm{size}}=P_{\mathrm{size}}J_-,
 \qquad
 J_-^\dagger P_{\mathrm{size}}=P_{\mathrm{size}}J_-^\dagger.
\]
The following proposition states the properties of compressed records that we will use.
\begin{proposition}[Properties of compressed records]
\label{prop:compressed-records}
\label[proposition]{prop:cm-interface}

Let $s_+:=N^6, s_-:=(1-\varepsilon_0)^2\mu^2.$
Let \(\mathcal E\) be the joint Hilbert space of the query
registers and any reference registers.
There exist absolute constants \(C_0,C_{\mathrm{rec}}>0\) such that
the following statements hold for every $\tau \in \{0,1\}$,  \(R\ge0\), and $e_R(s):=C_{\mathrm{rec}}\frac{R+1}{\sqrt s}$: 
\begin{enumerate}[label=\textup{(\roman*)}]
\item \label{item:record-isometry}\textbf{Encoding.}
The map \(J_{
\tau
}\) is an isometry. Moreover,
\(J_+\ket{\mathcal D_+}=\ket\emptyset\).

\item \label{item:record-growth}\textbf{Record growth.}
For every \(\xi\in\mathsf{Rec}_{\tau,\le R}\otimes\mathcal E\),
\[
 C\xi\in\mathsf{Rec}_{\tau,\le R+6}\otimes\mathcal E.
\]
For every \(\phi\in\mathcal H\otimes\mathcal E\),
\[
 J\phi\in\mathsf{Rec}_{\tau, \le R}\otimes\mathcal E
 \quad\Longrightarrow\quad
 JO\phi\in\mathsf{Rec}_{\tau,\le R+6}\otimes\mathcal E.
\]

\item \textbf{Comparison of \(C\) and \(O\).}
Suppose \(s_{\tau}\ge C_0(R+13)^2\).
For every \(\phi\in\mathcal H\otimes\mathcal E\) such that
\(J\phi\in\mathsf{Rec}_{\tau, \le R}\otimes\mathcal E\),
\begin{equation}
 \|(CJ-JO)\phi\|\le e_R(s)\|\phi\|.
 \label{eq:recording-query-forward}
\end{equation}
For every \(\xi\in\mathsf{Rec}_{\tau,\le R}\otimes\mathcal E\),
\begin{equation}
 \|(J^\dagger C-OJ^\dagger)\xi\|
 \le e_{R+6}(s)\|\xi\|.
 \label{eq:recording-query-adjoint}
\end{equation}
In the NO case, require \(P_{\mathrm{size}}\phi=\phi\) in
\eqref{eq:recording-query-forward} and
\(P_{\mathrm{size}}\xi=\xi\) in
\eqref{eq:recording-query-adjoint}.
\end{enumerate}
\end{proposition}
\begin{proof}
For (i), \cite[Theorem~4.2]{CM26} proves that \(J_X\) is an isometry
and maps the uniform superposition of permutations to the empty table vector.
Tensor products preserve isometry, as does the direct sum over
\(\theta\) in the NO case.
The Fourier transform on the coordinate-sign register is unitary and
maps the uniform superposition of sign assignments to \(\ket\emptyset\).
Thus \(J\) is an isometry and
\(J_+\ket{\mathcal D_+}=\ket\emptyset\).

For (ii), each occurrence of \(\mathrm{Dec}\) in \(C=\mathrm{Dec}\,Z\,\mathrm{Dec}\) changes the number
of table entries by at most one.
The operator \(Z\) changes the sign database by symmetric difference
with a set of at most four coordinates.
This proves the assertion for \(C\).

For the assertion about \(JO\phi\), \cref{app:query-interface}
implements \(O\) with two calls to the permutation operator \(Q_x\)
defined there and one operation on the coordinate-sign register.
Each call to \(Q_x\) increases encoded table size by at most one,
by \cite[Lemma~4.5]{CM26}.
The operation on the coordinate-sign register takes symmetric difference
with a set of at most four coordinates. Thus the total increase is at most six.

For (iii), \cref{app:query-interface} applies
\cite[Theorem~4.2]{CM26} to each of the two calls to \(Q_x\).
Their errors add to give \eqref{eq:recording-query-forward}.
The appendix proves this bound under \(s\ge C_0(r+7)^2\).
It then takes adjoints and applies the bound at \(r=R+6\) to prove
\eqref{eq:recording-query-adjoint}; the assumption
\(s\ge C_0(R+13)^2\) suffices.
For NO vectors, apply these bounds for each
\(\theta\in\Theta_{\mathrm{size}}\), where every restriction
\(\pi_p|_B\) has domain size at least \(s_-\), and add squared norms
over the orthogonal states \(\ket\theta\).
\end{proof}

%% file: database-completion.tex
\subsection{Completion of the records}
\label{subsec:database-transfer}
At this point, the information the verifier has learned from a hard-YES oracle is represented
by a superposition of short records $|D,S\rangle$, where $D$ contains the
permutation entries that have been exposed and $S$ records the dependence on
the random signs.  We now want to reinterpret the \emph{same recorded
information} as information that could have come from a hard-NO oracle.

The only obstacle is the additional hidden structure of the NO distribution.
A NO instance comes with a pair-label assignment
$\theta=(c_A,c_B)$, which partitions the vertices into the hidden blocks of
Section \ref{subsec:hard-no}, and every NO permutation must preserve these blocks.  A YES
permutation database $D$ was created without regard to this partition.  Thus, if we simply
sampled $\theta$ independently, an entry $x\mapsto y$ of $D$ might place
$x$ and $y$ in different NO blocks, in which case the entry could
not occur in a hard-NO instance.

The completion map $W$, which we will construct later in the section, resolves this without changing the permutation
entries of $D$.  For each record $(D,S)$, $W$ creates a superposition over hidden
pair-label assignments $\theta$, giving zero amplitude to assignments that
are incompatible with the entries of $D$ and reweighting the compatible
ones.  The reweighting also inserts the additional sign factors that
distinguish NO edges from YES edges.  Thus the map keeps the physical
information already recorded by the verifier fixed, while supplying a
hidden NO explanation for that information.

The main task of this section is to choose these reweighted amplitudes consistently.
They must preserve the norm of every record and must interact correctly
with the sign associated with each entry of $D$.  This is subtle because different
entries of $D$ may share coordinates, so the conditions they impose on $\theta$
are not independent.  

Formally, let $\nu$ denote the product distribution of the hidden
pair-labels $\theta$.  For every record $(D,S)$ and pair-labeling $\theta$, we define
a real coefficient $a_{D,S}(\theta)$ such that $
\mathbb E_{\theta\sim\nu}|a_{D,S}(\theta)|^2=1.$


\subsubsection{Removing and reinserting entries}
\label{subsec:peeling}

The Peeling Algorithm (Algorithm \ref{alg:peeling}) removes an entry $x \mapsto y$ from an input record $(D,S)$ when the two  coordinates of its endpoint $y$- $(y_A,y_B)$,
occur in no other remaining entry in $D$,  and have the same membership in
 $S$. We call such an oriented entry \emph{removable}, and call $y$ its \emph{removed endpoint}.

\begin{algorithm}[H]
\caption{Peeling \((D,S)\)}
\label{alg:peeling}
\textbf{Input:} A record \((D,S)\) %

\textbf{Initialize:} \(E\gets D\), \(S_{\text{cur}}\gets S\), and \(L\gets()\).

\textbf{While} there exists an entry of $E$ that can be oriented as
$$
e=(p;x\mapsto y),
\qquad
x=(x_A,x_B),\quad y=(y_A,y_B),
$$
such that

\begin{enumerate}
    \item $x_A,x_B,y_A,y_B$ are distinct;
    \item $y_A$ and $y_B$ occur in no other entry of $E$;
    \item $y_A$ and $y_B$ have the same membership in $S_{\text{cur}}$,
\end{enumerate}
do the following:

\begin{itemize}
    \item Let
    $
    b_e
    :=
    \mathbf 1[y_A\in S_{\text{cur}}]
    =
    \mathbf 1[y_B\in S_{\text{cur}}].
    $

    \item Append $(e,b_e)$ to $L: L \leftarrow L \cup (e,b_e)$.

    \item Remove $e$ from $E$:
    $
    E\leftarrow E\setminus\{e\}.
    $

    \item Set
    \[
 S_{\text{cur}}\gets
 \begin{cases}
 S_{\text{cur}},&b_e=0,\\
 S_{\text{cur}}\triangle\partial e,&b_e=1.
 \end{cases}
\]
\end{itemize}

\textbf{End while.}

\textbf{Output:} \(C=E\), \(F=S_{\text{cur}}\setminus \operatorname{coords}(C)\), and \(L\).
\end{algorithm}
 
After Algorithm \ref{alg:peeling} terminates, \(C\) is the residual  database. Thus, for every
\[
e=(p;x\leftrightarrow y)\in C,
\qquad
x=(x_A,x_B),\quad y=(y_A,y_B),
\]
 condition 1, 2 or 3 must be violates. More concretely,
\[
\begin{aligned}
&|\{x_A,x_B,y_A,y_B\}|<4,\\
&\text{or both of the following hold:}\\[1mm]
&\qquad
\left[
\begin{array}{c}
\text{one of }x_A,x_B\text{ occurs in another entry of }C,\\
\text{or }\mathbf 1[x_A\in S_{\text{cur}}]\neq \mathbf 1[x_B\in S_{\text{cur}}]
\end{array}
\right],\\[2mm]
&\qquad
\left[
\begin{array}{c}
\text{one of }y_A,y_B\text{ occurs in another entry of }C,\\
\text{or }\mathbf 1[y_A\in S_{\text{cur}}]\neq \mathbf 1[y_B\in S_{\text{cur}}].
\end{array}
\right]
\end{aligned}
\]
The set $F$ consists of the elements that remain in
$S_{\text{cur}}$ but do not appear in any entry of the residual database $C$. These are sign
dependencies that are not attached to any remaining entry.
Thus,
$$
F
=
\left\{
u\in S_{\text{cur}}:
u\notin
\{x_A,x_B,y_A,y_B:
(p;x\leftrightarrow y)\in C\}
\right\}.
$$
\subsubsection{The completion amplitudes}
Recall that $K = N^2$. For a labelling $\theta$ and a vertex $x$, define the \emph{block hit probability} as follows:
\[
\beta_{\theta}(x):= \begin{cases}
    24/K^2 & \text{if } \exists Q \text{ s.t } x \in B_Q\\
    1-(N)_4/K^2 & \text{if } x \in B_{\text{remainder}}
\end{cases}
\]
This is the probability that a vertex $y$, whose coordinates receive independent and uniformly
random labels, lands in the same block as $x$. Indeed, if $x \in B_Q$ for some 4 element set $Q \subseteq [N]$, the four labels $c(y)$ must form
one of the $24$ orderings of $Q$, and if $x \in B_{\text{remainder}}$, then the four labels of $y$ must not be all distinct, and exactly
$(N)_4:=N(N-1)(N-2)(N-3)$ of the $K^2$ sequences in $[N]^4$ have distinct entries.
\begin{definition}[Completion Amplitudes]
   Fix a record $(D,S)$,  one execution of
Algorithm~\ref{alg:peeling} with output $(C,F,L)$, and a pair labelling $\theta$.
\begin{enumerate}
    \item For all $(e, b_e)\in L$ with $e= (p; x \rightarrow y)$ oriented as in Algorithm \ref{alg:peeling}, let 
    \[
    w_e(\theta):= \frac{\mathbb{1}[y \in \theta(x)]}{\sqrt{\beta_{\theta}(x)}}, \lambda_e(\theta):= w_e(\theta) \cdot s_{\theta}(x,y)^{b_e}
    \]
    \item Let \(v=|\operatorname{coords}(C)|\), and define
$$
\mathcal G_C
:=
\left\{
\theta:
\begin{array}{l}
c_A(x_A)=(1,2)
\text{ for every }x_A\in\mathcal X_A\text{ appearing in }C,\\
c_B(x_B)=(3,4)
\text{ for every }x_B\in\mathcal X_B\text{ appearing in }C
\end{array}
\right\}.
$$ and $$
\gamma_C(\theta)
=
K^{v/2}\mathbf 1[\theta\in\mathcal G_C]
$$
 \item Let $f =|F|$. For each $u\in F$, write its pair-label as
$(a_u,b_u)$.  Define
$$
\mathcal G_F:=\{a_{u_1},b_{u_1} \dots a_{u_{f}},b_{u_{f}} \} \text{ are all distinct }
$$

On $\mathcal G_F$, let $\beta_\theta(F)$ be the sign of the permutation
that sorts the sequence
$
a_{u_1},b_{u_1},\ldots,a_{u_f},b_{u_f},
$
for any ordering $u_1,\ldots,u_f$ of $F$. Define
$$
\gamma_F(\theta)
=
\sqrt{\frac{K^f}{(N)_{2f}}}\,
\beta_\theta(F)\mathbf 1[\theta\in\mathcal G_F].
$$
where $(N)_{2f}:= N(N-1) \cdots (N-2f+1)$, and $\gamma_F(\theta)=1$ if $F = \emptyset$.
\end{enumerate}
The completion amplitude is
\[a_{D,S}(\theta):= \gamma_C(\theta)\gamma_F(\theta)\prod_{(e,b_e)\in L}\lambda_e(\theta)\] 
\end{definition}
We will now interpret/ provide some intuition for each factor in the completion amplitude $a_{D,S}(\theta)$. 
\begin{itemize}
    \item The weight $w_e$ is the amplitude analogue of the classical reinsertion step in the technical overview. Conditioned on all other labels, the two pair-labels of the removed endpoint $y$ are fresh and
uniformly random, and $w_e$ keeps exactly those labelings that place $y$ in the block of $x$. Next, if  $b_e=1$, then  both $y_A$ and $y_B$ belong to the current sign set 
$S_{\text{cur}}$. Peeling toggles $\partial_e$, i.e., it removes the coordinates $y_A,y_B$ and toggles $x_A,x_B$. This is the same update that occurs when the permutation $p$ is queried at $x$ on an oracle from the YES distribution. Now, the factor $s_{\theta}(x,y)$ accounts for the difference in cases that the permutation $p$ is queried at $x$ on an oracle from the YES vs the NO distribution.   
    \item The factor $\gamma_C$ places every vertex in $C$ in the same block, $B_{1234}$, with the same labels $c(x)=(1,2,3,4)$. Since $\varepsilon(x)=1$ for this ordering, the sign factors introduced by YES and NO queries are the same for vertices in $C$.
    \item The factor $\gamma_F$ assigns pair labelings to all the coordinates leftover in the signed coordinate set $S_{\text{cur}}$. 
\end{itemize}
\begin{lemma}\label{lem:peeling_ind}[Independence of the peeling choices]
    For every record $(D,S)$  such that $|\text{coords}(D)\cup S|<\frac{N}{2}$, the completion amplitude $a_{D,S}(\theta)$ does not depend on the order in which the entries are removed, nor on the orientations chosen in Algorithm \ref{alg:peeling}.
\end{lemma}
\begin{proof}
 Fix a temporary state $(E,S_{\text{cur}})$, a removable oriented entry
$e=(p;x\mapsto y)$, and its removal bit
$
b_e:=\mathbf 1[y_A\in S_{\text{cur}}]=\mathbf 1[y_B\in S_{\text{cur}}].
$
Let
$$
R_e(E,S_{\text{cur}})
:=
\bigl(E\setminus\{e\},\,S_{\text{cur}}\triangle(b_e\partial e)\bigr),
\qquad
\lambda_e(\theta)
:=
w_e(\theta)s_\theta(x,y)^{b_e}.
$$
Thus peeling $e$ multiplies the amplitude by
$\lambda_e(\theta)$ and replaces $(E,S_{\text{cur}})$ by $R_e(E,S_{\text{cur}})$.

We prove by induction on $|E|$ that the resulting amplitude is
independent of the peeling choices. This follows immediately when
no entry is removable. Suppose first that two distinct entries $e,f\in E$ are removable.
Since the removed coordinates of either entry occur in no other entry,
$$
\partial e\cap\{(y_f)_A,(y_f)_B\}
=
\partial f\cap\{(y_e)_A,(y_e)_B\}
=
\varnothing.
$$
Hence peeling one does not change the removal bit of the other, and
$$
R_fR_e(E,S_{\text{cur}})
=
R_eR_f(E,S_{\text{cur}})
=
\left(
E\setminus\{e,f\},
S_{\text{cur}}\triangle(b_e\partial e)
 \triangle(b_f\partial f)
\right).
$$
Moreover,
$$
\lambda_e(\theta)\lambda_f(\theta)
=
\lambda_f(\theta)\lambda_e(\theta).
$$
The induction hypothesis therefore gives the same amplitude in either
order.

It remains to compare the two orientations of one entry
$
e=(p;x\leftrightarrow y).
$
If both orientations are removable, all four coordinates of $e$ are distinct
and occur in no other entry. Let
$$
b_x
:=
\mathbf 1[x_A\in S_{\text{cur}}]
=
\mathbf 1[x_B\in S_{\text{cur}}],
\qquad
b_y
:=
\mathbf 1[y_A\in S_{\text{cur}}]
=
\mathbf 1[y_B\in S_{\text{cur}}].
$$
If $b_x=b_y$, then
$$
R_{x\mapsto y}(E,S_{\text{cur}})
=
R_{y\mapsto x}(E,S_{\text{cur}}),\qquad
\lambda_{x\mapsto y}(\theta)
=
\lambda_{y\mapsto x}(\theta)
=
w_e(\theta)s_\theta(x,y)^{b_x}.
$$

Now suppose $b_x\neq b_y$. Without loss of generality,  assume
$
b_x=0,
b_y=1.
$
Since the four coordinates of $e$ occur nowhere in $E\setminus\{e\}$,
$$
\left.
(S_{\text{cur}}\triangle\partial e)
\right|_{\operatorname{coords}(E\setminus\{e\})}
=
\left.
S_{\text{cur}}
\right|_{\operatorname{coords}(E\setminus\{e\})}.
$$
Hence the two runs may use identical subsequent peeling steps and the same residual database $C$. Let $F_{x\to y}$ and $F_{y\to x}$ denote the outputs sets $F$
obtained when $e$ is first peeled in the orientations $x\to y$ and
$y\to x$, respectively. Since none of the four coordinates of $e$ occurs in
the remaining database, there is a common set $F_0$, disjoint from
these four coordinates, such that
$$
F_{x\to y}
=
F_0\cup\{x_A,x_B\},
\qquad
F_{y\to x}
=
F_0\cup\{y_A,y_B\}.
$$
Indeed, in the orientation $x\to y$ the removal bit is
$b_y=1$, so toggling $\partial e$ removes $y_A,y_B$ from $S_{\text{cur}}$
and inserts $x_A,x_B$; in the orientation $y\to x$ the removal bit
is $b_x=0$, so $S_{\text{cur}}$ is unchanged.

If 
$w_e(\theta)=0$, both contributions vanish. If
$x,y\in\mathcal B_{\mathrm{remainder}}$, then
$
\gamma_{F_{x\to y}}(\theta)
=
\gamma_{F_{y\to x}}(\theta)
=
0
$, since the four labels of $x$ are not all distinct, and neither are those of $y$.
Finally, if $x,y\in\mathcal B_Q$ for some
$Q\in\binom{[N]}4$, then
$
\gamma_{F_{x\to y}}(\theta)
=
s_\theta(x,y)\gamma_{F_{y\to x}}(\theta).
$
Indeed, list the coordinates of $F_0$ first. The label sequences of $F_{x\to y}$ and $F_{y\to x}$ then agree except in their last four entries, which are $c(x)$ and $c(y)$, respectively.
Since $c(x)$ and $c(y)$ are orderings of the same set $Q$, the event $\mathcal G_{F_{x\to y}}$ holds exactly when $\mathcal G_{F_{y\to x}}$ does, and the two sorting signs differ by the sign of the rearrangement taking $c(x)$ to $c(y)$, which is $\varepsilon(x)\varepsilon(y)=s_\theta(x,y)$.
Therefore
$$
w_e(\theta)s_\theta(x,y)\gamma_{F_{x\to y}}(\theta)
=
w_e(\theta)s_\theta(x,y)^2\gamma_{F_{y\to x}}(\theta)
=
w_e(\theta)\gamma_{F_{y\to x}}(\theta).
$$
Thus the two orientations also give the same amplitude. This completes the induction.

\end{proof}

\subsubsection{Properties of the amplitudes}

\begin{lemma}[Completion amplitudes]
\label{lem:completion-amplitude_prop}
  
For every record $(D,S)$ with $|\text{coords}(D)\cup S|< N/2$, the amplitudes
$
a_{D,S}(\theta)$ %
satisfy  the following
properties.
\begin{enumerate}
\item[\rm (i)] \label{item:1}\textbf{Block compatibility.}
    If $a_{D,S}(\theta)\neq 0$, then the two endpoints of every
    entry of $D$ lie in the same NO block.

    \item[\rm (ii)] \label{item:2}\textbf{Normalization.}
    $$
    \mathbb E_{\theta\sim\nu}
    |a_{D,S}(\theta)|^2
    =1.
    $$

    \item[\rm (iii)]  \label{item:3}\textbf{Sign identity for entries of $D$.}
    For every $e=(p;x\mapsto y)\in D$,
    $$
    a_{D,S\triangle\partial e}(\theta)
    =
    s_\theta(x,y)\,a_{D,S}(\theta).
    $$

    \item[\rm (iv)] \label{item:4}\textbf{Fresh-entry identity.}
    Let $e=(p;x\mapsto y)\notin I$, suppose $I\cup\{e\}$ is a
    database, and suppose
    $$
    \{y_A,y_B\}
    \cap
    \bigl(\operatorname{coords}(I)\cup S\cup\{x_A,x_B\}\bigr)
    =
    \varnothing.
    $$
   Then
    $$
    a_{I\cup\{e\},S}(\theta)
    =
    w_e(\theta)a_{I,S}(\theta).
    $$
\end{enumerate}
\end{lemma}

\begin{proof}

.

\noindent \emph{ Block Compatibility.}
Let $(C,F,L)$ be the output of any execution of
Algorithm~\ref{alg:peeling}. By lemma \ref{lem:peeling_ind}, we can choose the order of execution freely. Since
$
D
=
C\sqcup\{e:(e,b_e)\in L\},
$
it suffices to consider these two sets of entries separately.

If
$
a_{D,S}(\theta)\neq 0,
$
then by definition
$
\gamma_C(\theta)\neq 0
\text{ and }
w_e(\theta)\neq 0
\text{ for every }(e,\sigma_y)\in L.
$
$\gamma_C(\theta)\neq 0$ implies that every vertex appearing in $C$
belongs to $\mathcal B_{\{1,2,3,4\}}$. $w_e(\theta)\neq 0$ implies, by the
definition of $w_e$, that the two endpoints of every peeled entry
belong to the same NO block. Hence every entry of $D$ has compatible
endpoints.

\medskip
\noindent\emph{Normalization.}
Write the peeled entries in removal order as
$$
L=((e_1,b_1),\ldots,(e_r,b_r)),
\qquad
e_i=(p_i;x_i\mapsto y_i).
$$
Since \(|s_{\theta}(x_i,y_i)^{b_i}|^2=1\),
$$
|a_{D,S}(\theta)|^2
=
|\gamma_C(\theta)|^2
|\gamma_F(\theta)|^2
\prod_{i=1}^r w_{e_i}(\theta)^2.
$$

For each $i$, let
$
Y_i:=\{(y_i)_A,(y_i)_B\}.
$
At the moment $e_i$ is peeled, the coordinates in $Y_i$ occur in no later
database entry. Moreover, after the peeling update they are absent
from the new sign set: if $b_i=0$ they were absent already,
and if $b_i=1$ they are removed by toggling $\partial e_i$.
Hence, after the factors
$w_{e_1}^2,\ldots,w_{e_{i-1}}^2$ have been averaged out, no remaining
factor except $w_{e_i}^2$ depends on the pair-labels of the coordinates in
$Y_i$.

Therefore, if
$
H_i(\theta)
:=
|\gamma_C(\theta)|^2|\gamma_F(\theta)|^2
\prod_{j=i+1}^r w_{e_j}(\theta)^2,
$
then $H_i$ is independent of the pair-labels on $Y_i$, and
$$
\mathbb E_\nu
\left[
H_i w_{e_i}^2
\right]
=
\mathbb E_\nu
\left[
H_i\,
\mathbb E_\nu
\left[
w_{e_i}^2
\mid
\theta|_{Y_i^c}
\right]
\right]
=
\mathbb E_\nu[H_i].
$$
Applying this successively for $i=1,\ldots,r$ gives
$
\mathbb E_\nu|a_{D,S}(\theta)|^2
=
\mathbb E_\nu
\left[
|\gamma_C(\theta)|^2|\gamma_F(\theta)|^2
\right].
$

Finally,
$
\operatorname{coords}(C)\cap F=\varnothing,
$
so $\gamma_C$ and $\gamma_F$ depend on disjoint sets of independent
pair-labels. Hence
$$
\mathbb E_\nu
\left[
|\gamma_C|^2|\gamma_F|^2
\right]
=
\mathbb E_\nu|\gamma_C|^2\,
\mathbb E_\nu|\gamma_F|^2
=
1.
$$

\noindent \emph{Sign identity}
Fix
$
e=(p;x\mapsto y)\in D.
$
We will compare the peeling procedures for
$
(D,S)
\text{ and }
(D,S\triangle\partial e).
$
By Lemma~\ref{lem:peeling_ind}, we may choose the peeling
order in both runs.

As long as $e$ remains in the database, let $f\neq e$ be removable.
The output coordinates of $f$ occur in no other remaining entry, and hence
are disjoint from $\partial e$. Thus $f$ is removable in both runs
with the same removal bit. We therefore peel the same entries in the
two runs until either $e$ is removed or both runs terminate.

Suppose first that $e$ is removed. If its removal bit in the first run
is $b_e$, then by construction, its removal bit in the second run is $1-b_e$.
Moreover, after removing $e$ the two sign sets agree:
$$
(S_{\text{cur}}\triangle\partial e)
\triangle
((1-b_e)\partial e)
=
S_{\text{cur}}\triangle(b_e\partial e).
$$
Let $B(\theta)$
denote the subsequent factors, then
$$
a_{D,S}(\theta)
=
B(\theta)\,
w_e(\theta)s_\theta(x,y)^{b_e},
$$
while
$$
a_{D,S\triangle\partial e}(\theta)
=
B(\theta)\,
w_e(\theta)s_\theta(x,y)^{1-b_e}.
$$
Since
$
s_\theta(x,y)^2=1,
$
we have
$
s_\theta(x,y)^{1-b_e}
=
s_\theta(x,y)\,
s_\theta(x,y)^{b_e},
$
and therefore
$
a_{D,S\triangle\partial e}(\theta)
=
s_\theta(x,y)a_{D,S}(\theta).
$

If $e$ is never removed, then $e\in C$ in both runs. By definition,  \(S_{\text{cur}}\) and \(S_{\text{cur}}'\) differ only on
$
\partial e\subseteq\operatorname{coords}(C),
$
so the resulting sets $F$ are identical. Hence the two amplitudes are
equal. If $\gamma_C(\theta)=0$, then we are done; otherwise the definition
of $\gamma_C$ gives
$
s_\theta(x,y)=1.
$
Thus the same identity holds in this case.

\noindent \medskip
\noindent\emph{Fresh-entry identity:}
Let
$
e=(p;x\mapsto y)\notin I
$
satisfy the hypotheses of the lemma. Then
$$
\{y_A,y_B\}
\cap
\bigl(
\operatorname{coords}(I)\cup S\cup\{x_A,x_B\}
\bigr)
=
\varnothing
$$
implies that, in the record $(I\cup\{e\},S)$, $
y_A,y_B
\text{ occur in no other entry}$ and $
\mathbf 1[y_A\in S]
=
\mathbf 1[y_B\in S]
=
0.
$
Thus $e$ is removable in the orientation $x\mapsto y$ with
$
b_e=0.
$
By the independence of the peeling choices, we can peel $e$ first. Then
$
(E,S_{\text{cur}})
=
(I\cup\{e\},S)
\longmapsto
(I,S),
$
and the peeled-entry factor is
$
w_e(\theta)s_\theta(x,y)^0
=
w_e(\theta).
$
Therefore
$$
a_{I\cup\{e\},S}(\theta)
=
w_e(\theta)a_{I,S}(\theta).
$$
\end{proof}

\subsubsection{The completion map}

Call a record $(D,S)$ \emph{admissible} if
$
|\operatorname{coords}(D)\cup S|<N/2.
$
If $a_{D,S}(\theta)\neq0$, Lemma \ref{lem:completion-amplitude_prop} implies that every entry of
$D$ has both endpoints in one NO block. In this case, define the NO database
\(D_\theta\) by
\[
(D_\theta)_{p,B}
:=\{x\mapsto y\in D_p:x,y\in B\}
\]
for every \(p\in[N]\) and every nonempty block \(B\) of \(\theta\).
Since every entry of $D$ has both endpoints in one block, \(D_\theta\) retains
exactly the physical entries of \(D\).
For every admissible record $(D,S)$, define
\begin{equation}
 W|D,S\rangle
 =\sum_{\theta\in\Theta_{\mathrm{size}}}
    \sqrt{\nu(\theta)}\,a_{D,S}(\theta)
    |\theta,D_\theta,S\rangle.
 \label{eq:W-definition}
\end{equation}
Here \(\Theta_{\mathrm{size}}\) is defined in \eqref{eq:class-size-condition}, 
terms with zero coefficient are omitted, and we 
set \(W|D,S\rangle=0\)
for all records $(D,S)$ which are not admissible. %
This defines the  completion map:
$$
W:\mathsf{Rec}_+\longrightarrow\mathsf{Rec}_-.
$$
\begin{lemma}[Exact properties of completion]
\label{lem:W-exact-properties}
The completion map
$
W:\mathsf{Rec}_+\to\mathsf{Rec}_-
$
has the following properties.

\begin{enumerate}
    \item For every $\xi\in\mathsf{Rec}_+$,
    $
    \|W\xi\|\le \|\xi\|.
    $

    \item 
    $
    Z_-W
    =
    WZ_+.
    $
\end{enumerate}
\end{lemma}

\begin{proof}
For (i), if $(D,S)$ is admissible, then by Lemma \ref{lem:completion-amplitude_prop},
$$
\begin{aligned}
\|W|D,S\rangle\|^2
&=
\sum_{\theta\in\Theta_{\mathrm{size}}}
\nu(\theta)|a_{D,S}(\theta)|^2
\\
&\le
\mathbb E_{\theta\sim\nu}
|a_{D,S}(\theta)|^2
=
1.
\end{aligned}
$$
If $(D,S)$ is not admissible, then $W|D,S\rangle=0$.
Moreover, for $(D,S)\neq(D',S')$, the images
$W|D,S\rangle$ and $W|D',S'\rangle$ are orthogonal.
Hence, for
$
\xi=\sum_{D,S}\alpha_{D,S}|D,S\rangle,
$
we have
$$
\|W\xi\|^2
=
\sum_{D,S}
|\alpha_{D,S}|^2
\|W|D,S\rangle\|^2
\le
\sum_{D,S}|\alpha_{D,S}|^2
=
\|\xi\|^2.
$$
For (ii), fix a basis state $|D,S\rangle$ and a query
$(p,z,v)$. If the entry is absent in $D$, then $Z_+$ and $Z_-$
both act as the identity, and hence
$
Z_-W|D,S\rangle
=
WZ_+|D,S\rangle.
$

Suppose instead that the entry is
$
e=(p;x\mapsto y)\in D.
$
Let
$$
v'
=
\begin{cases}
y,&z=0,\\
x,&z=1.
\end{cases}
$$
For $t\in\{0,1\}$, Lemma \ref{lem:completion-amplitude_prop} gives
$$
a_{D,S\triangle(t\partial e)}(\theta)
=
s_\theta(x,y)^t a_{D,S}(\theta).
$$
Therefore
$$
\begin{aligned}
&Z_-W|D,S\rangle
|p,z,v,w,\widehat t\rangle
\\
&=
\sum_{\theta\in\Theta_{\mathrm{size}}}
\sqrt{\nu(\theta)}\,
s_\theta(x,y)^t a_{D,S}(\theta)
|\theta,D_\theta,S\triangle(t\partial e)\rangle
|p,z,v,w\oplus v',\widehat t\rangle
\\
&=
\sum_{\theta\in\Theta_{\mathrm{size}}}
\sqrt{\nu(\theta)}\,
a_{D,S\triangle(t\partial e)}(\theta)
|\theta,D_\theta,S\triangle(t\partial e)\rangle
|p,z,v,w\oplus v',\widehat t\rangle
\\
&=
WZ_+|D,S\rangle
|p,z,v,w,\widehat t\rangle.
\end{aligned}
$$

Since
$
\partial e\subseteq\operatorname{coords}(D),
$
we also have
$
\operatorname{coords}(D)\cup
\bigl(S\triangle(t\partial e)\bigr)
=
\operatorname{coords}(D)\cup S,
$
so admissibility is unchanged. Hence the identity also holds when
$W|D,S\rangle=0$.

Thus
$$
Z_-W=WZ_+.
$$
\end{proof}

%% file: completion-estimates.tex
\subsection{Estimates for the transfer}
\label{sec:remaining-completion-estimates}
Section~\ref{subsec:database-transfer} defined the completion map $W$ on  records
and proved that
$$
Z_-W=WZ_+.
$$
We now  use this to prove  the transfer theorem.

We establish the following.   First, the definition of $W$ only contains 
pair-labelings in $\Theta_{\rm size}$ from \eqref{eq:class-size-condition},
so we must show that restricting pair labellings to lie in $\Theta_{\rm size}$
 loses negligible norm.  Second, we to establish that the map $W$ is approximately compatible with the
decompression operators $\mathrm{Dec}_+$ and $\mathrm{Dec}_-$.    Finally,we show that undoing the compression isometry $J$ puts negligible weight 
on NO parameters that fail the graph-concentration condition. Throughout this section, $R$ denotes a nonnegative integer. 

Recall
$$
K=N^2,
\qquad
\mu=\frac{N^3}{K}=N,
\qquad
\varepsilon_0=\frac1{40}.
$$
For a basis record $(D,S)$, set
$
H(D,S):=\operatorname{coords}(D)\cup S.
$
By Lemma \ref{lem:completion-amplitude_prop}, before restricting to $\Theta_{\rm size}$,
$
\mathbb E_{\theta\sim\nu}
\left[
|a_{D,S}(\theta)|^2
\right]
=1.
$
Moreover, $a_{D,S}(\theta)$ depends only on the pair-labels assigned
to the coordinates in $H(D,S)$.  Finally, distinct basis
records have orthogonal images under $W$, since $W$ retains their
permutation entries and sign sets.
\subsubsection{Loss from the class-size restriction}

The coefficients $a_{D,S}(\theta)$ are normalized with respect to the
full product distribution $\nu$, whereas $W$ is defined only for 
$\theta\in\Theta_{\rm size}$.  In this subsection, we prove that this restriction only loses
negligible squared norm.
\label{sec:completion-class-bound}

\begin{lemma}[Loss from the class-size restriction]
\label[lemma]{lem:completion-class-loss} 
Let  \(\rho_N:=4K\exp(-\varepsilon_0^2\mu/24)\).
If \(N\ge C_*(R+13)\) for some large enough constant $C_*$, then every
\(\xi\in\mathsf{Rec}_{+,\le R}\otimes\mathcal E\) satisfies
\begin{equation}
 (1-\rho_N)\|\xi\|^2\le\|W\xi\|^2\le\|\xi\|^2.
 \label{eq:W-near-isometry}
\end{equation}
\end{lemma}
\begin{proof}
Fix a basis record \(\ket{D,S}\) of size at most \(R\).
Let \(H:=\operatorname{coords}(D)\cup S\). Each permutation entry
uses at most four coordinates, so it follows that,
$
 |H|\le4|D|+|S|\le4R<N/2, |H|\le\varepsilon_0\mu/2.
$
This implies that $(D,S)$ is admissible.

Let \(F:=\{\theta\notin\Theta_{\mathrm{size}}\}\), and
fix a labeling \(\eta\) on \(H\). Let \(\theta^\eta\) denote $\theta|_H=\eta$, and it  is  distributed according to \(\nu\) conditioned on \(\theta|_H=\eta\).  For $P \in\{A,B\}$, let
$
X_{P,I}(\theta):=|\{x_P : \theta(x_P) = I\}|,
Y_{P,I}(\theta^{\eta}):=|\{y_P : \theta^{\eta}(y_P) = I\}|$.
For all $\theta$,
$
|Y_{P,I}-X_{P,I}|\le |H|.
$
Hence, using \(|H|\le\varepsilon_0\mu/2\),
\[
\begin{aligned}
|Y_{P,I}-\mu|>\varepsilon_0\mu
&\implies
|X_{P,I}-\mu|
\ge |Y_{P,I}-\mu|-|Y_{P,I}-X_{P,I}|\\
&>\varepsilon_0\mu-|H|
\ge \frac{\varepsilon_0\mu}{2}.
\end{aligned}
\]
Therefore,
\[
\begin{aligned}
\Pr_{\theta\sim \nu}[F\mid\theta|_H=\eta]
&=\Pr[\theta^\eta\notin\Theta_{\mathrm{size}}]\\
&\le
\sum_{\substack{P\in\{A,B\}\\ I\in\mathcal C}}
\Pr\!\left[|X_{P,I}-\mu|>\frac{\varepsilon_0\mu}{2}\right]\\
&\le
2K\cdot 2\exp\!\left(-\frac{(\varepsilon_0/2)^2\mu}{3}\right)\\
&=
4K\exp\!\left(-\frac{\varepsilon_0^2\mu}{12}\right)\\
&\le
4K\exp\!\left(-\frac{\varepsilon_0^2\mu}{24}\right)
=\rho_N.
\end{aligned}
\]
Next, observe that 
$$
\begin{aligned}
1-\|W|D,S\rangle\|^2
&=
\mathbb E_{\theta\sim\nu}
\left[
\mathbf 1_{F}
|a_{D,S}(\theta)|^2
\right]
\\
&=
\mathbb E_{\theta|_H}
\left[
|a_{D,S}(\theta)|^2
\Pr[
F
\mid
\theta|_H
]
\right]
\\
&\le
\rho_N
\mathbb E_{\theta\sim\nu}
|a_{D,S}(\theta)|^2
=
\rho_N.
\end{aligned}
$$
Lemma \ref{lem:W-exact-properties} already gives
$
\|W|D,S\rangle\|\le1.
$

Finally, the images of distinct basis records under $W$ are
orthogonal, which implies that 
$$
(1-\rho_N)\|\xi\|^2
\le
\|W\xi\|^2
\le
\|\xi\|^2.
$$
\end{proof}

\subsubsection{Completion and decompression}
\label{sec:completion-decompression-bound}
We next compare decompression before and after applying the completion map $W$.  

\begin{lemma}[Completion and decompression]
\label{lem:completion-decompression}
There is an absolute constant \(c_{\rm dec}>0\) such that, 
if \(N\ge C_*(R+13)\), every address \(a\) and every
\(\xi\in\mathsf{Rec}_{+,\le R+6}\otimes\mathcal E\) satisfy
\begin{equation}
 \|(\mathrm{Dec}_{-,a}W-W\mathrm{Dec}_{+,a})\xi\|
 \le\delta_R\|\xi\|,
 \qquad
 \delta_R:=c_{\rm dec}\sqrt{\frac{R+1}{N}}.
 \label{eq:W-decompression}
\end{equation}

\end{lemma}
\begin{proof}
We prove the statement for a forward query; the inverse case is the
same after exchanging inputs and outputs. Fix a forward address $p\in[N]$ and
$x\in V$, and let
$
d:=R+6.
$

First, we delete the entry at $(p,0,x)$ from $(D,S)$, and denote the remaining
record by \((I,S)\). The records $(D,S)$ with the same remaining record  $(I,S)$ span
a subspace \(\mathcal V_{I,S}\) with basis
\begin{equation}
 \ket\perp:=\ket{I,S},\qquad
 \ket y:=\ket{I\cup\{(p;x\mapsto y)\},S}
 \quad(\forall y\notin\operatorname{im}(I_p)).
 \label{eq:fixed-residual-space}
\end{equation}
Here \(x\notin\operatorname{dom}(I_p)\).

 Let \(\Pi_{I,S}\) be the orthogonal projector onto the group \(\mathcal V_{I,S}\): the span of records whose remainder after deleting the queried entry is exactly \((I,S)\). For an arbitrary input \(\xi\), define 
$
\zeta_{I,S}:=\Pi_{I,S}\xi.
$
Since $\mathcal V_{I,S}$ and $\mathcal V_{I',S'}$ are orthogonal for $(I,S)\neq (I',S')$, we can write
\[
\xi=\sum_{I,S}\zeta_{I,S},
\qquad
\|\xi\|^2=\sum_{I,S}\|\zeta_{I,S}\|^2.
\]
Now recall that we want to prove
\[
\|E_a\xi\|\le\delta_R\|\xi\|,
\qquad
E_a=\mathrm{Dec}_{-,a}W-W\mathrm{Dec}_{+,a}.
\]
So suffices to prove that, for every \((I,S)\),
\[
\|E_a\zeta\|\le\delta_R\|\zeta\|
\quad
\text{for every }\zeta\in\mathcal V_{I,S}
\]

Fix such a pair $(I,S)$, and the corresponding subspace $\mathcal{V}_{I,S}$. Let
\[
 H:=\operatorname{coords}(I)\cup S\cup\{x_A,x_B\},\qquad
 t:=|H|\le4d+2,\qquad
 Y:=(\mathcal X_A\setminus H)
\times
(\mathcal X_B\setminus H),.
\]
Let
$$
n_A:=|\mathcal X_A\setminus H|,
\qquad
n_B:=|\mathcal X_B\setminus H|,
\qquad
n:=|Y|=n_An_B.
$$
For a unit vector $v$ orthogonal to $\bot$, let $S_v$ denote
the operator that swaps $\bot$ and $v$ and acts as identity on the orthogonal complement. Now recall what  the YES Decompression operator, $\mathrm{Dec}_{+,a}$, does: it simply applies $S_v$, where $v$ is the uniform superposition over all unused outputs. 

For the moment, we will replace  $\mathrm{Dec}_{+,a}$ with  $\mathrm{Dec}^Y_{+,a}:= S_{u_0}$, where 
$
u_0
:=
\frac1{\sqrt n}
\sum_{y\in Y}|y\rangle
$. Thus $\mathrm{Dec}^Y_{+,a}$ swaps \(|\perp\rangle\) with the smaller uniform superposition over the set of fresh outputs $Y$.

To define a similar fresh output NO decompression operator,  fix a pair labeling $\theta$ and  define
$
Y_\theta
:=
Y\cap\mathcal B_\theta(x),
$ where $B_{\theta}(x)$ is the block containing $x$ under the pair labelling $\theta$.
Now, similarly to $\mathrm{Dec}^Y_{+,a}$ , we can define the operator $\mathrm{Dec}^Y_{-,a}$ .  On the event
$
G:=\{\theta\in\Theta_{\rm size}\},
$
let $\mathrm{Dec}^Y_{-,a}= S_{u_\theta}$, where
$
u_\theta
:=
\frac1{\sqrt{|Y_\theta|}}
\sum_{y\in Y_\theta}|y\rangle.
$
Thus $\mathrm{Dec}^Y_{-,a}$ swaps \(|\perp\rangle\) with the smaller uniform superposition over the set of fresh outputs that are in the same NO block as $x$.   Outside $G$, we define $\mathrm{Dec}^Y_{-,a}$ to be the identity.

We will for now prove that, for every
$\zeta\in\mathcal V_{I,S}\otimes\mathcal E$,
\begin{equation}
\left\|
\left(
\mathrm{Dec}^Y_{-,a}W
-
W\mathrm{Dec}^Y_{+,a}
\right)\zeta
\right\|
=
O(N^{-1/2})\|\zeta\|.
\label{eq:completion-restricted-bound}
\end{equation}
Let $ 
\zeta=z|\bot\rangle+f$ where $f=\sum_{y\in Y}f_y|y\rangle$ and $ z,f_y\in\mathbb C.
$
We first analyze $W \mathrm{Dec}_{+,a}^{Y}\zeta$.
For each \(y\in Y, \) let $
I_y:=I\cup\{(p;x\mapsto y)\}.
$
\[\ket{y}=
|I_y,S\rangle
=
\frac1{\sqrt n}u_0
+
\left(|I_y,S\rangle-\frac1{\sqrt n}u_0\right).
\] 
Hence
\[
\mathrm{Dec}^{Y}_{+,a}|I_y,S\rangle
=
\frac1{\sqrt n}|I,S\rangle
+|I_y,S\rangle-\frac1{\sqrt n}u_0.
\]
Now,
\[
\begin{aligned}
\mathrm{Dec}^{Y}_{+,a}\zeta\\&= \mathrm{Dec}^{Y}_{+,a}
\left(z\ket{\bot}+\sum_{y\in Y}f_y\ket{y}\right)\\
&=
\mathrm{Dec}^{Y}_{+,a}
\left(z|I,S\rangle+\sum_{y\in Y}f_y|I_y,S\rangle\right)
\\
&=
z\,\mathrm{Dec}^{Y}_{+,a}|I,S\rangle
+\sum_{y\in Y}f_y\,\mathrm{Dec}^{Y}_{+,a}|I_y,S\rangle
\\
&=
z u_0+
\sum_{y\in Y}f_y
\left(
\frac1{\sqrt n}|I,S\rangle
+|I_y,S\rangle-\frac1{\sqrt n}u_0
\right)
\\
&=
z u_0
+\left(\frac1{\sqrt n}\sum_{y\in Y}f_y\right)|I,S\rangle
\\
&\qquad
+\sum_{y\in Y}f_y|I_y,S\rangle
-\left(\frac1{\sqrt n}\sum_{y\in Y}f_y\right)u_0
\\
&=
z u_0+f_0|I,S\rangle+f-f_0u_0
\\
&=
f_0|I,S\rangle+f+(z-f_0)u_0.
\end{aligned}
\]
Applying $W$, 
\[
\begin{aligned}
W\mathrm{Dec}^{Y}_{+,a}\zeta
&=
W\bigl(f_0|I,S\rangle+f+(z-f_0)u_0\bigr)
\\
&=
f_0W|I,S\rangle
+\sum_{y\in Y}f_yW|I_y,S\rangle
+\frac{z-f_0}{\sqrt n}\sum_{y\in Y}W|I_y,S\rangle
\\
&=
f_0W|I,S\rangle
+\sum_{y\in Y}
\left(f_y+\frac{z-f_0}{\sqrt n}\right)W|I_y,S\rangle.
\end{aligned}
\]
Expanding \(W\), then using the guarantee from Lemma \ref{lem:completion-amplitude_prop}:
\[
\begin{aligned}
W\mathrm{Dec}^{Y}_{+,a}\zeta
&=
\sum_\theta\sqrt{\nu(\theta)}
\Biggl[
f_0a_{I,S}(\theta)|\theta,I_\theta,S\rangle
\\
&\qquad+
\sum_{y\in Y}
\left(f_y+\frac{z-f_0}{\sqrt n}\right)
a_{I_y,S}(\theta)|\theta,(I_y)_\theta,S\rangle
\Biggr]
\\
&=
\sum_\theta\sqrt{\nu(\theta)}\,a_{I,S}(\theta)
\Biggl[
f_0|\theta,I_\theta,S\rangle
\\
&\qquad+
\sum_{y\in Y}
\left(f_y+\frac{z-f_0}{\sqrt n}\right)
\frac{\mathbf1[y\in Y_\theta]}{\sqrt{\beta_\theta(x)}}
|\theta,(I_y)_\theta,S\rangle
\Biggr]
\\
&=
\sum_\theta\sqrt{\nu(\theta)}\,a_{I,S}(\theta)
\,|\theta\rangle\otimes
\Biggl[
f_0|I_\theta,S\rangle
\\
&\qquad+
\frac1{\sqrt{\beta_\theta(x)}}
\sum_{y\in Y_\theta}f_y|(I_y)_\theta,S\rangle
\\
&\qquad+
\frac{z-f_0}{\sqrt{\beta_\theta(x)n}}
\sum_{y\in Y_\theta}|(I_y)_\theta,S\rangle
\Biggr].
\end{aligned}
\]
Finally, using
\[
\frac1{\sqrt{\beta_\theta(x)n}}
\sum_{y\in Y_\theta}|(I_y)_\theta,S\rangle
=
\kappa_\theta u_\theta,
\] where $\kappa_\theta= \frac{\sqrt{|Y_{\theta}|}}{\sqrt{\beta_{\theta}(x)n}}$
we obtain
\[
\begin{aligned}
W\mathrm{Dec}^{Y}_{+,a}\zeta
=
\sum_\theta
&\sqrt{\nu(\theta)}\,a_{I,S}(\theta)
\,|\theta\rangle\otimes
\Biggl[
f_0|I_\theta,S\rangle
\\
&+
\frac1{\sqrt{\beta_\theta(x)}}
\sum_{y\in Y_\theta}f_y|(I_y)_\theta,S\rangle
\\
&+\kappa_\theta(z-f_0)u_\theta
\Biggr].
\end{aligned}
\]
Now we will analyze $\mathrm{Dec}_{-,a}^{Y}W\zeta$.
For \(y\in Y_\theta\), the NO decompression satisfies
\[
\mathrm{Dec}^{Y}_{-,a}|\theta,I_\theta,S\rangle
=
|\theta\rangle\otimes u_\theta,
\]
and
\[
\begin{aligned}
&\mathrm{Dec}^{Y}_{-,a}|\theta,(I_y)_\theta,S\rangle\\
&\quad=
|\theta\rangle\otimes
\left(
\frac1{\sqrt{|Y_\theta|}}|I_\theta,S\rangle
+|(I_y)_\theta,S\rangle
-\frac1{\sqrt{|Y_\theta|}}u_\theta
\right).
\end{aligned}
\]
Expanding \(W\), then applying NO decompression:
\[
\begin{aligned}
\mathrm{Dec}^{Y}_{-,a}W\zeta
&=
\mathrm{Dec}^{Y}_{-,a}
\sum_\theta\sqrt{\nu(\theta)}\,a_{I,S}(\theta)
\Biggl[
z|\theta,I_\theta,S\rangle
\\
&\qquad+
\frac1{\sqrt{\beta_\theta(x)}}
\sum_{y\in Y_\theta}f_y|\theta,(I_y)_\theta,S\rangle
\Biggr]
\\
&=
\sum_\theta\sqrt{\nu(\theta)}\,a_{I,S}(\theta)
\Biggl[
z\,\mathrm{Dec}^{Y}_{-,a}|\theta,I_\theta,S\rangle
\\
&\qquad+
\frac1{\sqrt{\beta_\theta(x)}}
\sum_{y\in Y_\theta}f_y
\,\mathrm{Dec}^{Y}_{-,a}|\theta,(I_y)_\theta,S\rangle
\Biggr]
\\
&=
\sum_\theta\sqrt{\nu(\theta)}\,a_{I,S}(\theta)
\,|\theta\rangle\otimes
\Biggl[
z u_\theta
\\
&\qquad+
\frac1{\sqrt{\beta_\theta(x)}}
\sum_{y\in Y_\theta}f_y
\left(
\frac{|I_\theta,S\rangle}{\sqrt{|Y_\theta|}}
+|(I_y)_\theta,S\rangle
-\frac{u_\theta}{\sqrt{|Y_\theta|}}
\right)
\Biggr]
\\
&=
\sum_\theta\sqrt{\nu(\theta)}\,a_{I,S}(\theta)
\,|\theta\rangle\otimes
\Biggl[
z u_\theta
\\
&\qquad+
\frac1{\sqrt{\beta_\theta(x)}}
\sum_{y\in Y_\theta}f_y|(I_y)_\theta,S\rangle
\\
&\qquad+
\left(
\frac1{\sqrt{\beta_\theta(x)|Y_\theta|}}
\sum_{y\in Y_\theta}f_y
\right)
\bigl(|I_\theta,S\rangle-u_\theta\bigr)
\Biggr].
\end{aligned}
\]
Finally, rewrite the scalar coefficient:
\[
\begin{aligned}
\frac1{\sqrt{\beta_\theta(x)|Y_\theta|}}
\sum_{y\in Y_\theta}f_y
&=
\frac1{\kappa_\theta\sqrt n}
\sum_{y\in Y}
\frac{\mathbf1[y\in Y_\theta]}{\beta_\theta(x)}f_y
\\
&=
\frac1{\kappa_\theta}
\Biggl[
\frac1{\sqrt n}\sum_{y\in Y}f_y
\\
&\qquad+
\frac1{\sqrt n}\sum_{y\in Y}
\left(
\frac{\mathbf1[y\in Y_\theta]}{\beta_\theta(x)}-1
\right)f_y
\Biggr]
\\
&=
\frac{f_0+g_\theta}{\kappa_\theta},
\end{aligned}
\]
where
\[
g_\theta:=
\frac1{\sqrt n}\sum_{y\in Y}
\left(
\frac{\mathbf1[y\in Y_\theta]}{\beta_\theta(x)}-1
\right)f_y.
\]
Therefore,
\[
\begin{aligned}
\mathrm{Dec}^{Y}_{-,a}W\zeta
=
\sum_\theta
&\sqrt{\nu(\theta)}\,a_{I,S}(\theta)
\,|\theta\rangle\otimes
\Biggl[
z u_\theta
\\
&+
\frac1{\sqrt{\beta_\theta(x)}}
\sum_{y\in Y_\theta}f_y|(I_y)_\theta,S\rangle
\\
&+
\frac{f_0+g_\theta}{\kappa_\theta}
\bigl(|I_\theta,S\rangle-u_\theta\bigr)
\Biggr].
\end{aligned}
\]

Now observe that
\[
\begin{aligned}
&\bigl(\mathrm{Dec}^{Y}_{-,a}W-W\mathrm{Dec}^{Y}_{+,a}\bigr)\zeta
\\
&\quad=
\sum_{\theta\in\Theta_{\mathrm{size}}}
\sqrt{\nu(\theta)}\,a_{I,S}(\theta)
|\theta\rangle\otimes e_\theta.
\end{aligned}
\]
where 
\[
\begin{aligned}
e_\theta
&=
z u_\theta
+\frac{f_0+g_\theta}{\kappa_\theta}
\bigl(|I_\theta,S\rangle-u_\theta\bigr)
\\
&\qquad
-f_0|I_\theta,S\rangle
-\kappa_\theta(z-f_0)u_\theta
\\
&=
(1-\kappa_\theta)(z-f_0)u_\theta
\\
&\qquad+
\left(
\frac{f_0+g_\theta}{\kappa_\theta}-f_0
\right)
\bigl(|I_\theta,S\rangle-u_\theta\bigr)
\\
&=
(1-\kappa_\theta)(z-f_0)u_\theta
\\
&\qquad+
\frac{g_\theta+(1-\kappa_\theta)f_0}{\kappa_\theta}
\bigl(|I_\theta,S\rangle-u_\theta\bigr).
\end{aligned}
\]
This implies that
\[
\begin{aligned}
&\left\|
\bigl(\mathrm{Dec}^{Y}_{-,a}W-W\mathrm{Dec}^{Y}_{+,a}\bigr)\zeta
\right\|^2
\\
&\quad=
\sum_{\theta\in\Theta_{\mathrm{size}}}
\nu(\theta)|a_{I,S}(\theta)|^2\|e_\theta\|^2.
\end{aligned}
\]
Now we will bound this expression. Since
\[
\|u_\theta\|=1,
\qquad
\||I_\theta,S\rangle-u_\theta\|=\sqrt2,
\]
the triangle inequality gives
\[
\begin{aligned}
\|e_\theta\|
&\le
|1-\kappa_\theta|\,|z-f_0|+
\frac{\sqrt2}{\kappa_\theta}
\left(
|g_\theta|+|1-\kappa_\theta|\,|f_0|
\right).
\end{aligned}
\]
Recall that
\[
g_\theta
=
\frac1{\sqrt n}\sum_{y\in Y}
\left(\frac{\mathbf1_y}{\beta}-1\right)f_y.
\]
We will compute $\mathbb{E}[g_{\theta}]^2$, and then show how this will be used to bound the quantity above.  Define the covariance matrix
\[
\begin{aligned}
\Sigma_{yy'}
&:=
\frac1n\,
\mathbb E\left[
\left(\frac{\mathbf1_y}{\beta}-1\right)
\left(\frac{\mathbf1_{y'}}{\beta}-1\right)
\right]
\\
&=
\frac{\operatorname{Cov}(\mathbf1_y,\mathbf1_{y'})}
{\beta^2n}.
\end{aligned}
\]
The second equality uses \(\mathbb E\mathbf1_y=\beta\).

Next, observe that
\[
\begin{aligned}
\mathbb E|g_\theta|^2
&=
\frac1n\sum_{y,y'}
\mathbb E\left[
\left(\frac{\mathbf1_y}{\beta}-1\right)
\left(\frac{\mathbf1_{y'}}{\beta}-1\right)
\right]
\overline{f_y}f_{y'}
\\
&=
\sum_{y,y'}\Sigma_{yy'}\overline{f_y}f_{y'}
\\
&=
\langle f|\Sigma|f\rangle
\\
&\le
\|\Sigma\|\sum_y|f_y|^2
\\
&=
\|\Sigma\|\,\|f\|^2.
\end{aligned}
\tag{1}
\]
Now, we will bound $\mathbb E(\kappa_\theta-1)^2$. 
Since \(|Y_\theta|=\sum_y\mathbf1_y\),
\[
\begin{aligned}
\kappa_\theta^2-1
&=
\frac{|Y_\theta|}{\beta n}-1\\
&=
\frac1n\sum_y
\left(\frac{\mathbf1_y}{\beta}-1\right).
\end{aligned}
\]
Also, \(\kappa_\theta\ge0\), so
\[
|\kappa_\theta-1|
=
\frac{|\kappa_\theta^2-1|}{\kappa_\theta+1}
\le|\kappa_\theta^2-1|.
\]
Consequently,
\[
\begin{aligned}
\mathbb E(\kappa_\theta-1)^2
&\le
\mathbb E(\kappa_\theta^2-1)^2
\\
&=
\frac1{n^2}\sum_{y,y'}
\mathbb E\left[
\left(\frac{\mathbf1_y}{\beta}-1\right)
\left(\frac{\mathbf1_{y'}}{\beta}-1\right)
\right]
\\
&=
\frac1n\sum_{y,y'}\Sigma_{yy'}
\\
&=
\langle u_0|\Sigma|u_0\rangle
\\
&\le
\|\Sigma\|.
\end{aligned}
\]
The same matrix therefore controls both errors, so 
we will bound $\|\Sigma\|$. 
Suppose first that \(x\) lies in a four-label block. Recall that \(K=N^2\) and the probability that $x$ lies the block defined by some labeling $\theta$ is
$
\beta_\theta(x)=\frac{24}{K^2}.$
For \(y,y'\in Y\), there are three cases.
If \(y=y'\), then
$
\mathbb E[\mathbf1_y\mathbf1_{y'}]
=
\mathbb E\mathbf1_y
=
\frac{24}{K^2}.$

If $y\neq y'$, then by definition this implies that $(y_{A},y_B)\neq (y'_{A},y'_B)$. The second case we consider is one  where exactly one of the  coordinates is equal. more concretely,  either $y_A=y'_A$ and $y_B\neq y'_B$  or $y_B=y'_B$ and $y_A\neq y'_A$. The shared coordinate can receive any of the $12$ ordered pairs of distinct elements of $Q$. Each choice leaves two allowed labels for each of the other two coordinates, namely the two orderings of the remaining two elements of $Q$. Thus
$
\mathbb E[\mathbf1_y\mathbf1_{y'}]
=
12\cdot\frac{4}{K^3}=\frac{48}{K^3}.
$
If they share no coordinates, their labels are independent, so
\[
\mathbb E[\mathbf1_y\mathbf1_{y'}]
=
\beta^2.
\]
Substituting into
\[
\Sigma_{yy'}
=
\frac{\mathbb E[\mathbf1_y\mathbf1_{y'}]-\beta^2}
{\beta^2n}
\]
gives
\[
\Sigma_{yy'}
=
\begin{cases}
(K^2/24-1)/n,&y=y',\\
(K/12-1)/n,&\text{one shared coordinate},\\
0,&\text{no shared coordinates}.
\end{cases}
\]
Each row of the matrix has  one diagonal entry and \(n_A+n_B-2\) entries corresponding to one shared coordinate. All entries are nonnegative. Therefore,
\[
\begin{aligned}
\|\Sigma\|
&\le
\max_y\sum_{y'}\Sigma_{yy'}
\\
&=
\frac{
K^2/24-1+(n_A+n_B-2)(K/12-1)
}{
n_An_B
}
\\
&\le
\frac{K^2}{24n_An_B}
+\frac K{12}\left(\frac1{n_A}+\frac1{n_B}\right)
\\
&=
O(N^{-2})+O(N^{-1})
\\
&=
O(1/N).
\end{aligned}
\tag{3}
\]
Here we used \(K=O(N^2)\) and \(n_A,n_B\ge N^3/2\).

For the remainder block, \(\beta=1-(N)_4/K^2\ge5/N\) for \(N\ge11\).
Indicators of outputs that share no coordinate are independent, so the corresponding entries of \(\Sigma\) vanish,
and the diagonal entries are \((\beta^{-1}-1)/n=O(N^{-5})\).
Now suppose that \(y\) and \(y'\) share exactly one coordinate.
Given the label of the shared coordinate, \(\mathbf1_y\) and \(\mathbf1_{y'}\) are independent with the same conditional mean \(q\),
where \(q=1\) if that label has the form \((a,a)\), and \(q=(5N-6)/N^2\) otherwise.
Since \(\mathbb E q=\beta\),
\[
\operatorname{Cov}(\mathbf1_y,\mathbf1_{y'})
=
\operatorname{Var}(q)
\le
\mathbb E q^2
\le
\frac1N+\frac{25}{N^2},
\]
so these entries of \(\Sigma\) are nonnegative and at most \(\frac{N^2}{25}\bigl(\frac1N+\frac{25}{N^2}\bigr)\frac1n=O(N/n)\).
Taking the largest row sum as before,
\[
\|\Sigma\|
\le
\frac{\beta^{-1}-1}{n}
+(n_A+n_B-2)\,O\!\left(\frac Nn\right)
=
O(N^{-2}).
\]
Putting everything together, 
\[
\mathbb E|g_\theta|^2
\le O(\frac {1}{N})\|f\|^2,
\qquad
\mathbb E(\kappa_\theta-1)^2
\le O(\frac {1}{N}).
\]
We also need a lower bound on the denominator $\kappa_{\theta}$ in order to bound the final expression. For all pair labels $\theta \in \Theta_{\text{size}}$, the number of coordinates consistent with $\theta$ is $\geq (1-\varepsilon_0)\cdot
\mu$. Not counting the coordinates in $H$, the number of consistent coordinates that are fresh is $\geq (1-\varepsilon_0)\cdot
\mu-t$. So we can conclude that $|Y_{\theta}|\geq 
|Y_\theta|
\ge
\beta K^2\bigl((1-\varepsilon_0)\mu-t\bigr)^2.
$
Using \(n\le N^6=K^2\mu^2\),
\[
\begin{aligned}
\kappa_\theta^2
&=
\frac{|Y_\theta|}{\beta n}
\\
&\ge
\frac{
\beta K^2\bigl((1-\varepsilon_0)\mu-t\bigr)^2
}{
\beta K^2\mu^2
}
\\
&=
\left(1-\varepsilon_0-\frac t\mu\right)^2
\\
&\ge\frac14.
\end{aligned}
\]
 Therefore,
\[
\kappa_\theta^{-1}\le2
\qquad
(\theta\in\Theta_{\mathrm{size}}).
\]
Now, let's recall that we had previously obtained:
\[
\begin{aligned}
\|e_\theta\|
\le{}&
|1-\kappa_\theta|\,|z-f_0|
\\
&+
\frac{\sqrt2}{\kappa_\theta}
\left(
|g_\theta|+|1-\kappa_\theta|\,|f_0|
\right).
\end{aligned}
\]
  \(\theta\in\Theta_{\mathrm{size}}\),
  Using the lower bound on $\kappa_{\theta}$ established above,
\[
\begin{aligned}
\|e_\theta\|
\le{}&
|1-\kappa_\theta|\,|z-f_0|
\\
&+
2\sqrt2\,|g_\theta|
+
2\sqrt2\,|1-\kappa_\theta|\,|f_0|.
\end{aligned}
\]
Thus
\[
\begin{aligned}
&\left(
\mathbb E[
\mathbf1_{\Theta_{\mathrm{size}}}\|e_\theta\|^2]
\right)^{1/2}
\\
&\quad\le
|z-f_0|\sqrt{\mathbb E(1-\kappa_\theta)^2}
\\
&\qquad+
2\sqrt2\sqrt{\mathbb E|g_\theta|^2}
\\
&\qquad+
2\sqrt2\,|f_0|
\sqrt{\mathbb E(1-\kappa_\theta)^2}
\\
&\quad\le
\sqrt{\|\Sigma\|}
\left(
|z-f_0|+2\sqrt2\|f\|+2\sqrt2|f_0|
\right).
\end{aligned}
\]
By Cauchy–Schwarz,
\[
\begin{aligned}
|f_0|^2
&=
\left|\frac1{\sqrt n}\sum_y f_y\right|^2
\\
&\le\sum_y|f_y|^2
=\|f\|^2.
\end{aligned}
\]
Also,
\[
\begin{aligned}
|z-f_0|^2
&\le2(|z|^2+|f_0|^2)\\
&\le2(|z|^2+\|f\|^2)\\
&=2\|\zeta\|^2.
\end{aligned}
\]
Therefore,
\[
\begin{aligned}
\left(
\mathbb E[
\mathbf1_{\Theta_{\mathrm{size}}}\|e_\theta\|^2]
\right)^{1/2}
&\le
5\sqrt2\,\sqrt{\|\Sigma\|}\,\|\zeta\|.
\end{aligned}
\]
Squaring gives
\[
\mathbb E[
\mathbf1_{\Theta_{\mathrm{size}}}\|e_\theta\|^2]
\le
50\|\Sigma\|\,\|\zeta\|^2
\le
\frac cN\|\zeta\|^2.
\]
Putting everything together, 
\[
\begin{aligned}
&\left\|
(\mathrm{Dec}^{Y}_{-,a}W-W\mathrm{Dec}^{Y}_{+,a})\zeta
\right\|^2
\\
&\quad=
\mathbb E_\nu\!\left[
|a_{I,S}(\theta)|^2
\mathbf1_{\Theta_{\mathrm{size}}}\|e_\theta\|^2
\right]
\\
&\quad\le
\frac cN\|\zeta\|^2
\,\mathbb E_\nu|a_{I,S}(\theta)|^2
\\
&\quad=
\frac cN\|\zeta\|^2,
\end{aligned}
\]

Hence 
\[
\left\|
(\mathrm{Dec}^{Y}_{-,a}W-W\mathrm{Dec}^{Y}_{+,a})\zeta
\right\|
\le
\sqrt{\frac cN}\,\|\zeta\|.
\]
Now we will look at entries that are not fresh, i.e. \(y\notin Y\). In this case, $\mathrm{Dec}^{Y}_{\pm,a}$ and $W$ does not change  \(y\). Therefore,
\[
\begin{aligned}
\mathrm{Dec}^{Y}_{-,a}W|I_y,S\rangle
&=W|I_y,S\rangle\\
&=W\mathrm{Dec}^{Y}_{+,a}|I_y,S\rangle.
\end{aligned}
\]
Thus the difference is zero on these nonfresh answer records. Combining this with the bound already proved on the fresh subspace gives
\[
\left\|
(\mathrm{Dec}^{Y}_{-,a}W-W\mathrm{Dec}^{Y}_{+,a})\zeta
\right\|
\le
\sqrt{\frac cN}\,\|\zeta\|
\]
for every \(\zeta\in\mathcal V_{I,S}\).
It remains to compare $\mathrm{Dec}^{Y}_{\pm,a}$ with $\mathrm{Dec}_{\pm,a}$, which represent the full decompression operators. 
Suppose \(v\) is uniform over \(M\) answer records, and \(v'\) is uniform over a subset containing \((1-\alpha)M\) of them. Then
\[
\begin{aligned}
\langle v|v'\rangle
&=
\frac{(1-\alpha)M}
{\sqrt M\sqrt{(1-\alpha)M}}\\
&=\sqrt{1-\alpha}.
\end{aligned}
\]
Consequently,
\[
\begin{aligned}
\|v-v'\|^2
&=2-2\langle v|v'\rangle\\
&=2-2\sqrt{1-\alpha}\\
&\le2\alpha.
\end{aligned}
\]
Observe that the swap between \(|\bot\rangle\) and \(v\) can be written as
\[
S_v=I-2|r_v\rangle\langle r_v|,
\qquad
r_v=\frac{|\bot\rangle-v}{\sqrt2}.
\]
Therefore,
\[
\begin{aligned}
\|S_v-S_{v'}\|
&=
2\bigl\|
|r_v\rangle\langle r_v|
-|r_{v'}\rangle\langle r_{v'}|
\bigr\|\\
&\le4\|r_v-r_{v'}\|\\
&=2\sqrt2\,\|v-v'\|\\
&\le4\sqrt\alpha.
\end{aligned}
\]
So deleting a fraction \(\alpha\) of the available answers changes the decompression operator by at most \(4\sqrt\alpha\).
\[
t_A=|H\cap A|,
\qquad
t_B=|H\cap B|,
\qquad
t_A+t_B=t.
\]
Since \(|A|=|B|=N^3\), the fraction of all outputs that are not fresh is
\[
\begin{aligned}
1-\frac{(N^3-t_A)(N^3-t_B)}{N^6}
&=
\frac{t_A+t_B}{N^3}
-\frac{t_At_B}{N^6}\\
&\le\frac t{N^3}.
\end{aligned}
\]
Thus, 
\[
\|\mathrm{Dec}_{+,a}-\mathrm{Dec}^{Y}_{+,a}\|
\le
4\sqrt{\frac t{N^3}}
\] on \(\mathcal V_{I,S}\).
Now we will compare $\mathrm{Dec}_{-,a}$ and $\mathrm{Dec}^{Y}_{-,a}$. 
Fix \(\theta\in\Theta_{\mathrm{size}}\). Its block is a union of cells
\[
C_J^A(\theta)\times C_{J'}^B(\theta),
\]
where each class has size at least \((1-\varepsilon_0)\mu\).

Within any such cell, the fraction of outputs having a coordinate in \(H\) is at most
\[
\begin{aligned}
\frac{|H\cap C_J^A|}{|C_J^A|}
+\frac{|H\cap C_{J'}^B|}{|C_{J'}^B|}
&\le
\frac{t_A+t_B}{(1-\varepsilon_0)\mu}\\
&=
\frac t{(1-\varepsilon_0)\mu}.
\end{aligned}
\]
Hence \[
\begin{aligned}
\|(\mathrm{Dec}_{-,a}-\mathrm{Dec}^{Y}_{-,a})W\zeta\|
&\le
4\sqrt{\frac t{(1-\varepsilon_0)\mu}}\,
\|W\zeta\|\\
&\le
4\sqrt{\frac t{(1-\varepsilon_0)\mu}}\,
\|\zeta\|.
\end{aligned}
\]
Finally, 
\[
\begin{aligned}
&\|(\mathrm{Dec}_{-,a}W-W\mathrm{Dec}_{+,a})\zeta\|\\
&\quad\le
\|(\mathrm{Dec}_{-,a}-\mathrm{Dec}^{Y}_{-,a})W\zeta\|\\
&\qquad+
\|(\mathrm{Dec}^{Y}_{-,a}W-W\mathrm{Dec}^{Y}_{+,a})\zeta\|\\
&\qquad+
\|W(\mathrm{Dec}^{Y}_{+,a}-\mathrm{Dec}_{+,a})\zeta\|\\
&\quad\le
\left(
4\sqrt{\frac t{(1-\varepsilon_0)\mu}}
+\sqrt{\frac cN}
+4\sqrt{\frac t{N^3}}
\right)\|\zeta\|.
\end{aligned}
\]
Now use
\[
\mu=N,
\qquad
t\le4(R+6)+2=4R+26.
\]
Since \(\varepsilon_0\) is a fixed constant, this gives
\[
\|(\mathrm{Dec}_{-,a}W-W\mathrm{Dec}_{+,a})\zeta\|
\le
c_{\mathrm{dec}}\sqrt{\frac{R+1}{N}}\,
\|\zeta\|.
\]
This proves the required bound on each fixed subspace \(\mathcal V_{I,S}\).
For a general input, as described above, we can decompose
\[
\xi=\sum_{I,S}\xi_{I,S},
\qquad
\xi_{I,S}\in\mathcal V_{I,S}.
\]
Thus
\[
\begin{aligned}
&\|(\mathrm{Dec}_{-,a}W-W\mathrm{Dec}_{+,a})\xi\|^2\\
&\quad=
\sum_{I,S}
\|(\mathrm{Dec}_{-,a}W-W\mathrm{Dec}_{+,a})\xi_{I,S}\|^2\\
&\quad\le
c_{\mathrm{dec}}^2\frac{R+1}{N}
\sum_{I,S}\|\xi_{I,S}\|^2\\
&\quad=
c_{\mathrm{dec}}^2\frac{R+1}{N}\|\xi\|^2.
\end{aligned}
\]
Taking square roots,
\[
\|(\mathrm{Dec}_{-,a}W-W\mathrm{Dec}_{+,a})\xi\|
\le
c_{\mathrm{dec}}\sqrt{\frac{R+1}{N}}\,\|\xi\|.
\]
\end{proof}

\subsubsection{Weight on bad NO parameters}
\label{sec:weighted-graph-bound}

Soundness applies on the range of \(\id-P_{\mathrm{bad}}\).
We therefore bound the norm on bad NO parameters after applying
\(J_-^\dagger\).

\begin{lemma}[Weight on graph-concentration failures]
\label{lem:weighted-graph-failure}
For an integer \(0\le R\le N/64\), let \(\xi\) be any vector in
\(\mathsf{Rec}_{-,\le R}\otimes\mathcal E\) with
\(P_{\mathrm{size}}\xi=\xi\). Then
\begin{equation}
 \|P_{\rm bad}J_-^\dagger\xi\|^2
 \le\gamma_R\|J_-^\dagger\xi\|^2,
 \qquad
 \gamma_R:=2N^6(N+1)^R\exp(-N/8192).
 \label{eq:bad-no-mass}
\end{equation}

\end{lemma}
\begin{proof}
First fix \(\theta\in\Theta_{\mathrm{size}}\) and an let  \(h\) denote the signs assigned to all the coordinates. Let \(Z_p\) denote the collection of NO permutations
at index \(p\).
Let \(\Gamma_p\) be the signed adjacency matrix from
\cref{sec:theproblem}. Each \(\Gamma_p\) is Hermitian and has norm
at most \(1\). Using this, we can define
\[
 \Gamma:=\frac1N\sum_{p=1}^N\Gamma_p,\qquad
 \overline\Gamma:=\mathbb E\Gamma,\qquad
 F:=\{\|\Gamma-\overline\Gamma\|>1/8\}.
\]
Here the expectation is over the permutations, with \(\theta\) and \(h\) fixed.
For \(B\subseteq[N]\), define \(Z_B:=(Z_p)_{p\in B}\).
Suppose \(|B|\le R\), and fix any assigned values \(z_B\).
To sample conditioned on  \(Z_B=z_B\), we simply sample all \(Z_p\) independently, then replace \(Z_B\) with \(z_B\).
Let \(\Gamma'\) be the resulting average adjacency matrix. Since each replaced
summand changes by at most \(2\) in operator norm,
\[
 \|\Gamma'-\Gamma\|\le\frac{2|B|}{N}\le\frac1{32}.
\]
If \(\|\Gamma'-\overline\Gamma\|>1/8\), then
\(\|\Gamma-\overline\Gamma\|>3/32>1/16\).
The original matrices \(\Gamma_p-\mathbb E\Gamma_p\) are independent,
centered, and have norm at most \(2\).
\Cref{thm:matrix-hoeffding}, with dimension \(N^6\), implies
\begin{equation}
 p_{\text{fail}}=\Pr(F\mid Z_B=z_B)
 \le\Pr(\|\Gamma-\overline\Gamma\|>1/16)
 \le\beta:=2N^6e^{-N/8192}.
 \label{eq:graph-failure-after-prescriptions}
\end{equation}

For a NO database \(D\), define
$
T(D):=\{p:\text{some table }D_{p,B}\text{ is nonempty}\}$ and observe that $|T(D)|\le|D|\le R.
$
Then we can decompose $\xi $ into records 
$
\xi=\sum_{|T|\le R}\xi_T,
$
where \(\xi_T\) contains exactly the basis records with \(T(D)=T.\)Let 
$
\psi_T:=J_-^\dagger\xi_T.
$. 
Then$
J_-^\dagger\xi=\sum_{|T|\le R}\psi_T.
$
So it suffices to bound $\|P_{\rm bad}J_-^\dagger \psi_T\|^2$ for each such $\psi_T$.
We can expand $
\psi_T
=
\sum_{\theta,h,z_T}
c_T(\theta,h,z_T)
|\theta,h,z_T\rangle
\otimes
\bigotimes_{p\notin T}|u_p^\theta\rangle,
$, where $c_T(\theta,h,z_T)\in \mathbb{C}, \ket{z_T}=\bigotimes_{p\in T}|\pi_p\rangle$, and $\ket{u_p^\theta}$ is the uniform superposition over permutations indexed by $p \notin T$. 
Thus \[
\begin{aligned}
\|P_{\mathrm{bad}}\psi_T\|^2
&=
\sum_{\theta,h,z_T}
|c_T(\theta,h,z_T)|^2
\Pr[F\mid\theta,h,Z_T=z_T]
\\
&\le
p_{\mathrm{fail}}
\sum_{\theta,h,z_T}|c_T(\theta,h,z_T)|^2
\\
&=
p_{\mathrm{fail}}\|\psi_T\|^2.
\end{aligned}
\]
Let $L= \sum_{j=1}^{R}\binom{N}{j}$ denote the number of choices of  sets \(T\). Then
Cauchy–Schwarz gives
\[
\begin{aligned}
\|P_{\mathrm{bad}}J_-^\dagger\xi\|^2
&=
\left\|\sum_TP_{\mathrm{bad}}\psi_T\right\|^2\\
&\le
L\sum_T\|P_{\mathrm{bad}}\psi_T\|^2\\
&\le
L\,p_{\mathrm{fail}}\sum_T\|\psi_T\|^2\\
&=
L\,p_{\mathrm{fail}}\|J_-^\dagger\xi\|^2\\
&\le
2N^6(N+1)^Re^{-N/8192}
\|J_-^\dagger\xi\|^2.
\end{aligned}
\]
This completes the proof. 
\end{proof}

%% file: oracle-transfer.tex
\subsection{Proof of the transfer theorem}
\label{sec:generic-transfer}
\label{sec:generic-oracle-transfer}
In this section, we put everything together to construct the transfer map $
\mathcal T:=J_-^\dagger WJ_+:
\mathcal H_+\longrightarrow\mathcal H_-.
$ and prove Theorem \ref{thm:transfer-certificate}.
First, some notation. For an address $a=(p,z,v)$, let $P_{\pm,a}$ be the orthogonal projector onto the span of the basis records whose database contains $(p, v\rightarrow y)$ if $z=0$, and $(p, x\rightarrow v)$ if $z=1$. By construction of the completion map $W$, an entry $a$ is present in some $W\ket{D,S}$ iff it is present in $D$, and furthermore $P_{-,a}W=WP_{+,a}$ for all $a$.

\begin{proposition}[Well-formedness and decoding]
\label{prop:record-decoding}
For $\tau\in\{+,-\}$, let
$\Pi_\tau:=J_\tau J_\tau^\dagger$, and define
\[
Q_{\tau,a}:=
(I-P_{\tau,a})\operatorname{Dec}_{\tau,a}P_{\tau,a}.
\]
Then, for every address $a$,
\begin{equation}
Q_{\tau,a}\Pi_\tau=0.
\label{eq:record-validity-cancellation}
\end{equation}
Moreover, for every
$\xi\in\mathrm{Rec}_\tau\otimes\mathcal E$,
\[
\|\xi\|^2-\|J_\tau^\dagger\xi\|^2
=\|(I-\Pi_\tau)\xi\|^2.
\]

In the NO case, let
$s_-:=(1-\varepsilon_0)^2\mu^2$ and suppose
$3R\le s_-$. Then every
$\xi\in\mathrm{Rec}_{-,\le R}\otimes\mathcal E$
satisfying $P_{\mathrm{size}}\xi=\xi$ obeys
\begin{equation}
\|(I-\Pi_-)\xi\|^2
\le 2\sum_a\|Q_{-,a}\xi\|^2.
\label{eq:recording-decoding-loss}
\end{equation}

\end{proposition}

\begin{proof}
The first identity is the cancellation property of
the permutation encoding \cite[Lemma~4.10]{CM26}.
Indeed, for a table $I$ of size $r$ containing
the entry $e$ queried by $a$,
\[
Q_{\tau,a}|I\rangle
=\frac{|I\setminus\{e\}\rangle}{\sqrt{s-r+1}},
\]
where $s$ is the size of the corresponding
permutation domain. This is precisely the
cancellation operator in \cite[Lemma~4.10]{CM26}.
The identity extends to our encoding by taking
tensor products over the permutation tables.

Since $J_\tau$ is an isometry, $\Pi_\tau$ is the
orthogonal projector onto its image and
$\|J_\tau^\dagger\xi\|=\|\Pi_\tau\xi\|$.
The second identity follows by orthogonal
decomposition.

For the NO bound, fix a balanced labeling
$\theta$. The encoding is a tensor product of
permutation encodings, with one factor for each
permutation and block. Every such block has size
at least $s_-$. By \cite[Lemma~4.11]{CM26}, the
one-permutation decoding estimate applies to
records of size at most $R$.

Writing $\Pi_j$ for the image projector of the
$j$-th permutation factor, we have
\[
I-\bigotimes_j\Pi_j
\preceq\sum_j(I-\Pi_j).
\]
Applying the one-permutation estimate to each
factor and summing gives
\[
\|(I-\Pi_{-,\theta})\xi_\theta\|^2
\le 2\sum_a\|Q_{-,a}\xi_\theta\|^2.
\]
The estimates extend from fixed table-size
sectors by orthogonality, since each
$Q_{-,a}$ removes one entry from a specified
table. Finally, summing squared norms over
the orthogonal labeling registers $\theta$
proves the claim.
\end{proof}

\begin{proof}[Proof of \cref{thm:transfer-certificate}]
Define
\begin{equation}
 \mathcal T:=J_-^\dagger WJ_+:\mathcal H_+\longrightarrow\mathcal H_-.
 \label{eq:physical-transfer}
\end{equation}

Fix a polynomially bounded integer-valued \(R=R(k)\ge0\).
Since \(N=2^k\), \(s_+=N^6\), and
\(s_-=(1-\varepsilon_0)^2\mu^2\ge N^2/2\),
\cref{prop:compressed-records,prop:record-decoding,lem:completion-class-loss,lem:completion-decompression,lem:weighted-graph-failure}
hold for all sufficiently large \(k\). For an input \(\phi\), let
$
\zeta:=J_+\phi\in\mathrm{Rec}_{+,\le R}.
$
We will first show that $(1-\eta_R^{\text{norm}})\|\phi\|^2 \le\|\mathcal T\phi\|^2\le\|\phi\|^2$. 
We first bound the part of \(W\zeta\) outside the well-formed NO subspace. Then
\[
\begin{aligned}
&(I-P_{-,a})\mathrm{Dec}_{-,a}P_{-,a}W\zeta\\
&\quad=
(I-P_{-,a})\mathrm{Dec}_{-,a}WP_{+,a}\zeta\\
&\quad=
(I-P_{-,a})
(\mathrm{Dec}_{-,a}W-W\mathrm{Dec}_{+,a})
P_{+,a}\zeta\\
&\qquad+
W(I-P_{+,a})\mathrm{Dec}_{+,a}P_{+,a}\zeta\\
&\quad=
(I-P_{-,a})
(\mathrm{Dec}_{-,a}W-W\mathrm{Dec}_{+,a})
P_{+,a}\zeta.
\end{aligned}
\]
The last equality follows from Equation \ref{eq:record-validity-cancellation}. Lemma \ref{lem:completion-decompression} therefore gives
\[
\|(I-P_{-,a})\mathrm{Dec}_{-,a}P_{-,a}W\zeta\|
\le
\delta_R\|P_{+,a}\zeta\|.
\]
Each recorded entry is counted once by its forward query and once by its inverse query. Hence
$
\sum_a\|P_{+,a}\zeta\|^2\le2R\|\zeta\|^2.
$
Applying Equation \ref{eq:recording-decoding-loss},
\[
\begin{aligned}
\|(I-\Pi_-)W\zeta\|^2
&\le
2\sum_a
\|(I-P_{-,a})\mathrm{Dec}_{-,a}P_{-,a}W\zeta\|^2\\
&\le
2\delta_R^2\sum_a\|P_{+,a}\zeta\|^2\\
&\le
4R\delta_R^2\|\zeta\|^2.
\end{aligned}
\tag{33}
\]
 Proposition \ref{prop:record-decoding} implies that for every  \(\xi\in\mathrm{Rec}_{\pm,\le R}\),
\[
\|\xi\|^2-\|J_\pm^\dagger\xi\|^2
=
\|(I-\Pi_\pm)\xi\|^2.
\]
Combining this observation with Lemma \ref{lem:completion-decompression} gives us 
\[
\begin{aligned}
\|\mathcal T\phi\|^2
&=\|J_-^\dagger W\zeta\|^2\\
&=\|W\zeta\|^2-\|(I-\Pi_-)W\zeta\|^2\\
&\ge
(1-\rho_N)\|\zeta\|^2
-4R\delta_R^2\|\zeta\|^2\\
&=
(1-\rho_N-4R\delta_R^2)\|\phi\|^2.
\end{aligned}
\tag{34}
\]
Together with \(\|\mathcal T\|\le1\), this proves Equation \ref{eq:transfer-norm} with
\[
\eta_R^{\mathrm{norm}}
:=
\rho_N+4R\delta_R^2.
\]

Next we will prove that $\|(O_-T-TO_+)\phi\|\le\varepsilon^{\text{query}}_R\|\phi\|.
$ Recall
$
\widetilde O_\pm
=
\mathrm{Dec}_\pm Z_\pm\mathrm{Dec}_\pm,
$ and let
$
\Delta:=\mathrm{Dec}_-W-W\mathrm{Dec}_+.
$
Lemma \ref{lem:completion-decompression} implies
\[
\begin{aligned}
\|\Delta\xi\|^2
&=
\sum_a
\|(\mathrm{Dec}_{-,a}W-W\mathrm{Dec}_{+,a})\xi_a\|^2\\
&\le
\delta_R^2\|\xi\|^2
\end{aligned}
\]
whenever the record size is at most \(R+6\).

Using \(Z_-W=WZ_+\),
\[
\begin{aligned}
\widetilde O_-W-W\widetilde O_+
&=
\mathrm{Dec}_-Z_-
\bigl(W\mathrm{Dec}_++\Delta\bigr)
-W\mathrm{Dec}_+Z_+\mathrm{Dec}_+\\
&=
\mathrm{Dec}_-Z_-\Delta
+\mathrm{Dec}_-WZ_+\mathrm{Dec}_+
-W\mathrm{Dec}_+Z_+\mathrm{Dec}_+\\
&=
\mathrm{Dec}_-Z_-\Delta
+\Delta Z_+\mathrm{Dec}_+.
\end{aligned}
\]
For input record size at most \(R\), both occurrences of \(\Delta\) act on records of size at most \(R+6\): \(\mathrm{Dec}_+\) increases the size by at most one, and \(Z_+\) changes \(S\) on at most four coordinates. Since the remaining factors are unitary,

\begin{equation}\label{eq:35}
\|(\widetilde O_-W-W\widetilde O_+)\zeta\|
\le
\|\Delta\zeta\|
+\|\Delta Z_+\mathrm{Dec}_+\zeta\|\\
\le2\delta_R\|\zeta\|.
\end{equation}

Now, for the  actual oracle queries, observe that

\begin{equation}\label{eq:36}
O_-\mathcal T-\mathcal TO_+
={}
(O_-J_-^\dagger-J_-^\dagger\widetilde O_-)\,WJ_+\\
+
J_-^\dagger(\widetilde O_-W-W\widetilde O_+)\,J_+\\
+
J_-^\dagger W(\widetilde O_+J_+-J_+O_+).
\end{equation}

Proposition \ref{prop:compressed-records} and Equation \ref{eq:35} imply
\[
\begin{aligned}
\|(O_-\mathcal T-\mathcal TO_+)\phi\|
\le
\bigl(
\operatorname{e}_{R+6}(s_-)
+2\delta_R
+\operatorname{e}_R(s_+)
\bigr)\|\phi\|.
\end{aligned}
\]
This proves Equation \ref{eq:transfer-query} with
\[
\eta_R^{\mathrm{query}}
:=
\operatorname{e}_{R+6}(s_-)
+2\delta_R
+\operatorname{e}_R(s_+).
\]
Finally, we will prove $\|P_{\mathrm{bad}}\mathcal T\phi\|^2\le\gamma^{\text{bad}}_R\|\mathcal T\phi\|^2.$

Apply Lemma \ref{lem:weighted-graph-failure} to \(W\zeta\):
\[
\begin{aligned}
\|P_{\mathrm{bad}}\mathcal T\phi\|^2
&=
\|P_{\mathrm{bad}}J_-^\dagger W\zeta\|^2\\
&\le
\eta_R^{\mathrm{bad}}\|J_-^\dagger W\zeta\|^2\\
&=
\eta_R^{\mathrm{bad}}\|\mathcal T\phi\|^2.
\end{aligned}
\]
This proves Equation \ref{eq:transfer-bad}, with
\[
\eta_R^{\mathrm{bad}}
=
2N^6(N+1)^Re^{-N/8192}.
\]
Finally, \(N=2^k\), \(R=\operatorname{poly}(k)\), \(s_+=N^6\), and \(s_-=\Omega(N^2)\). Thus
\[
\rho_N=e^{-\Omega(N)},
\qquad
\delta_R=O\!\left(\sqrt{\frac{R+1}{N}}\right),
\qquad
\operatorname{e}_r(s)=O\!\left(\frac{r+1}{\sqrt s}\right).
\]

Therefore all three error functions are negligible in \(k\), completing the proof.

\end{proof}

%% file: appendix.tex
\section{Properties of compressed records}
\label{app:cm-interface}

We prove the extensions of the CM results stated in
\cref{prop:cm-interface,prop:record-decoding}. Throughout this appendix,
\(J\) is the tensor product of the CM isometries and the sign Fourier
transform from \cref{subsec:recording}. Write \(\Pi_J:=JJ^\dagger\).
For the NO query and decoding estimates, assume
\(P_{\mathrm{size}}\xi=\xi\) on record inputs, or
\(P_{\mathrm{size}}\phi=\phi\) on physical inputs.
We apply those estimates for each \(\theta\in\Theta_{\mathrm{size}}\),
where every permutation domain has size at least \(s_-\), and then add
squared norms over \(\theta\).
All operators below retain \(\theta\), so they preserve these subspaces.
The identities and record-growth statements require no restriction on
\(\theta\).

\subsection{Well formed vectors and decoding}

For one permutation on a set of size \(s\), CM's image projector
preserves each table-size sector
\cite[Lemma~4.10, Equation~(38)]{CM26}. Their cancellation identities
\cite[Lemma~4.10, Equation~(39)]{CM26} and decoding estimate
\cite[Lemma~4.11, Equation~(45)]{CM26} use normalized sums of coefficients
obtained by removing an input or output. These are precisely our operators
\((\id-P_a)\mathrm{Dec}_aP_a\): for a forward address \(x\) and an \(r\)-entry
table \(D\),
\[
 (\id-P_x)\mathrm{Dec}_xP_x\ket D
 =\begin{cases}
 (s-r+1)^{-1/2}\ket{D\setminus\{x\mapsto y\}},&x\mapsto y\in D,\\
 0,&x\notin\operatorname{dom}(D).
 \end{cases}
\]
Indeed, the coefficient of the smaller table in \(\mathrm{Dec}_x\ket D\) is its
inner product with the normalized sum of fresh extensions. The final
projection removes every component that still contains \(x\).
Exchanging inputs and outputs proves the inverse-address identity.
The denominator agrees with
\cite[Equations~(42)--(43)]{CM26}. Thus their conclusions are exactly
our cancellation identity and decoding bound for one permutation.
The constant-two bound requires \(r\le(s+1)/3\), which follows from
our size assumption. At \(r=0\), the empty table is well formed and all
terms vanish.

To pass to several permutations, fix every table size and the sign
subset. Write \(\Pi_j\) for the projector onto the image of the compression
map for factor \(j\), restricted to that sector.
These projections act on distinct factors, and
\[
 \id-\bigotimes_j\Pi_j\preceq\sum_j(\id-\Pi_j).
\]
On a common eigenvector, the left side is one if any projection has
eigenvalue zero; the right side counts those zero eigenvalues. This
proves the operator inequality. Apply the single-factor estimate to
each term and add. Distinct size sectors have orthogonal outputs under
each fixed operator \((\id-P_a)\mathrm{Dec}_aP_a\), because it removes one entry
from the specified factor. We can therefore add the inequalities over
all size sectors and sign subsets. This proves
\eqref{eq:recording-decoding-loss}, since
\(\|J^\dagger\xi\|^2=\|\Pi_J\xi\|^2\).
Cancellation follows by applying the single-factor identity to each
factor. The same reasoning proves that \(\Pi_J\) preserves the sizes
of all tables and the sign subset.

Each estimate remains valid with additional registers: expand in an
orthonormal basis of those registers and add squared norms. For vectors
in the NO record space, apply the decoding estimate for each
\(\theta\in\Theta_{\mathrm{size}}\), then add over the orthogonal
\(\theta\) register. The cancellation identity holds for each \(\theta\).
For fixed \(\theta\), the physical addresses enumerate each independent
permutation's inputs and outputs exactly once.
This proves \cref{prop:record-decoding}.

\subsection{Signed queries and the adjoint estimate}
\label{app:query-interface}

\paragraph{The imported permutation query.}
Fix a forward input \(x\) for a permutation on a set of size \(s\).
A helper register has basis \(\ket\perp\) and \(\ket y\), where \(y\)
ranges over possible answers. Let \(\tau_y\) exchange \(\ket\perp\)
and \(\ket y\). The operator \(Q_x\) applies
\(\tau_{\pi(x)}\). Let \(L_x\) instead apply \(\tau_{D(x)}\) when
\(x\in\operatorname{dom}(D)\), and the identity otherwise. Define
\(\widetilde Q_x:=\mathrm{Dec}_xL_x\mathrm{Dec}_x\).

For every parameter-register vector \(\phi\) whose encoding \(J\phi\)
has at most \(r\) entries in this table,
\cite[Theorem~4.2, Equation~(12)]{CM26} gives
\begin{equation}
 \|(\widetilde Q_xJ-JQ_x)\phi\|
 \le b(r,s)\|\phi\|,
 \qquad
 b(r,s)=\frac{60}{\sqrt{s-r}}+\frac{r+1}{s-2r-1},
 \label{eq:CM-swap-bound}
\end{equation}
provided \(s\ge3\) and \(3r+2\le s\).
The estimate allows arbitrary reference registers and inverse queries.
The operator \(Q_x\) increases the encoded table size by at most one
\cite[Lemma~4.5, Equation~(16)]{CM26}.

\paragraph{The full signed answer.}
Write \(O_x,C_x,Z_x\) for \(O,C,Z\) at the fixed address under consideration.
Given address \(x\) and helper value \(y\), let \(M_x\) XOR the
endpoint \(y\) and the edge-sign bit into the answer register.
It leaves the address, helper, and oracle parameters unchanged, and acts
as the identity when the helper is \(\ket\perp\).
It uses only the two endpoints, their four coordinate signs, and, in
the NO case, \(\theta\). In the compressed representation it has a version
\(\widetilde M_x\) that toggles the corresponding sign subset. Hence
\[
 \widetilde M_xJ=JM_x,
 \qquad \widetilde M_x\mathrm{Dec}_x=\mathrm{Dec}_x\widetilde M_x.
\]
The first equality follows from the sign Fourier transform. The second
holds because \(\mathrm{Dec}_x\) acts only on the table and leaves the helper and
\(\theta\) unchanged.

With helper state \(\ket\perp\), the first \(Q_x\) writes
\(\pi(x)\) into the helper, \(M_x\) writes the signed answer, and
the second \(Q_x\) restores \(\ket\perp\). Thus \(Q_xM_xQ_x\)
implements \(O_x\).
On vectors in the record space, the preceding commutation identity gives
\[
 \widetilde Q_x\widetilde M_x\widetilde Q_x
 =\mathrm{Dec}_xL_x\widetilde M_xL_x\mathrm{Dec}_x=\mathrm{Dec}_xZ_x\mathrm{Dec}_x
\]
on a blank helper. The middle operation reads the stored entry and
restores the helper; it is the identity if the entry is absent. This
proves the second equality and identifies the operation with our
specific \(C_x\).

Insert \(J\) between the two swap calls. Subtracting the two routines
gives
\[
 \begin{aligned}
 &\widetilde Q_x\widetilde M_x\widetilde Q_xJ-JQ_xM_xQ_x\\
 &\quad=\widetilde Q_x\widetilde M_x(\widetilde Q_xJ-JQ_x)
       +(\widetilde Q_xJ-JQ_x)M_xQ_x.
 \end{aligned}
\]
If \(J\phi\) has record size at most \(r\), then \(JM_xQ_x\phi\)
has size at most \(r+5\): the swap adds at most one entry, and the
sign read toggles at most four coordinates. Equation~\eqref{eq:CM-swap-bound}
and the triangle inequality therefore bound the norm of this difference
by \((b(r,s)+b(r+5,s))\|\phi\|\). For sufficiently large absolute
\(C_0,C_{\rm rec}\), if \(s\ge C_0(r+7)^2\), this is at most
\(e_r(s)\|\phi\|\).
The same construction shows that \(JO_x\phi\) has record size at most \(r+6\).
For \(C_x=\mathrm{Dec}_xZ_x\mathrm{Dec}_x\), each \(\mathrm{Dec}_x\) changes table size by at most
one and \(Z_x\) toggles at most four coordinates, so its growth is also at
most six.

The registers of unqueried permutations are reference registers.
For a superposition of addresses, square the bound for each address and
sum over the orthogonal address register. Taking the maximum gives the
same estimate with \(s\) equal to a lower bound on the permutation
domain sizes, with no dependence on the number of permutations.
This is the argument of \cite[Corollary~4.17]{CM26}; it also permits
unequal domain sizes.
For NO query estimates, first fix \(\theta\in\Theta_{\mathrm{size}}\)
and then sum squared norms over \(\theta\). These observations prove
\eqref{eq:recording-query-forward} and both record-growth statements.

\paragraph{The adjoint on arbitrary vectors in the record space.}
Let \(Q_{\le r}\) project onto records of size at most \(r\), and put
\(\Delta=C\Pi_J-JOJ^\dagger\).
As proved above from \cite[Lemma~4.10]{CM26}, \(\Pi_J\) preserves each
table size and the sign database, and hence total record size.
The bound just proved for \(CJ-JO\) therefore gives
\(\|\Delta Q_{\le r}\|\le e_r(s)\) whenever
\(s\ge C_0(r+7)^2\).
For NO vectors, all operator norms in this argument are restricted to
\(\operatorname{im}P_{\mathrm{size}}\).
The XOR queries are Hermitian, so
\(\Delta^\dagger=\Pi_JC-JOJ^\dagger\).
Both terms send records of size at most \(R\) to well formed vectors of
record size at most \(R+6\). Consequently,
\[
 \begin{aligned}
 \Delta^\dagger Q_{\le R}
 &=Q_{\le R+6}\Delta^\dagger Q_{\le R},\\
 \|\Delta^\dagger Q_{\le R}\|
 &=\|Q_{\le R}\Delta Q_{\le R+6}\|
 \le e_{R+6}(s).
 \end{aligned}
\]
The bound here requires only \(s\ge C_0(R+13)^2\), as stated in the
proposition.
Finally, \(J^\dagger\) preserves the norm on the range of
\(\Delta^\dagger\), and
\(J^\dagger\Delta^\dagger=J^\dagger C-OJ^\dagger\).
This proves \eqref{eq:recording-query-adjoint} for every vector of record
size at most \(R\), whether or not it is well formed.

\section{Diagonalization argument}
\label{app:diagonalization}
\label{app:diag}

We use the query lower bound to construct the oracle, allowing us to
prove \cref{thm:main}. The argument is standard: we enumerate all
polynomial-time QMA oracle verifiers and defeat them one at a time.
At stage $i$, we choose a fresh input length $n_i$ and use
\cref{thm:query-lower-bound} to choose the $n_i$-th oracle component
on which the $i$-th verifier fails.

\begin{proof}[Proof of \cref{thm:main}]
We construct an input-length-indexed signed-permutation oracle
$$
O=(O^{(n)})_{n\ge 1},
$$
where $O^{(n)}$ is an oracle of the form in \cref{eq:oracle} for the
length-$n$ instance of \cref{problem:sign-matching}, with
\(k=n\) and \(\ell=3n\). As discussed
after \cref{eq:oracle}, we encode this as a Boolean oracle $A$ in the
standard way: the Boolean oracle takes as input $n$, a query
$(p,z,v)$ to $O^{(n)}$, and an index specifying one bit of the
answer, and returns that bit. We choose the encoding so that a query
to the $n$-th component has length at least $n$. The two query models
are polynomially equivalent.
We construct the components $O^{(n)}$ together with a unary promise
problem $L^A$. Enumerate all polynomial-time QMA oracle verifiers
$
V_1,V_2,\ldots
$
together with explicit polynomial bounds on their running time. For
each $V_i$, let $m_i(n)$ and $q_i(n)$ be polynomial bounds on its
witness length and number of Boolean oracle queries on input $1^n$,
respectively. We also let $a_i(n)$ be a polynomial bound on the
largest oracle component that $V_i$ can query on input $1^n$. Such a
bound exists because a query to the $r$-th oracle component has
length at least $r$.

We construct the oracle in stages. After stage $i$, we maintain a
number $B_i$ such that every oracle component of index at most $B_i$
has been permanently fixed. Initially, set
$
B_0=0.
$
Suppose stages $1,\ldots,i-1$ have been completed. We first choose a
fresh length
$
n_i>B_{i-1}.
$
Since $m_i(n)$ and $q_i(n)$ are polynomially bounded, we may choose
$n_i$ sufficiently large that \cref{thm:query-lower-bound} applies
to $V_i$ on input $1^{n_i}$. Here we use the polynomial equivalence
between Boolean queries and queries to \cref{eq:oracle}.

We now fix all other oracle components that $V_i$ could access on
this input. Set
$$
B_i:=\max\{B_{i-1},a_i(n_i),n_i\}.
$$
For every still-undefined component of index at most $B_i$, except
for $n_i$ itself, set
$$
\pi_p=\mathbb I,
\qquad
\sigma_p\equiv +1
$$
for every $p$. For this choice, $H=0$, so the component is a YES
instance of \cref{problem:sign-matching}.

At this point, every oracle component that $V_i$ can query on input
$1^{n_i}$ has been fixed except $O^{(n_i)}$.

Thus, with all other queried oracle components fixed, $V_i$ on input
$1^{n_i}$ becomes a QMA verifier for the length-$n_i$ instance of
\cref{problem:sign-matching}. It remains only to choose
$O^{(n_i)}$. By our choice of $n_i$ and
\cref{thm:query-lower-bound}, $V_i$ cannot satisfy completeness
$2/3$ on every instance in
$
\operatorname{supp}(\mathcal D_+)
$
and soundness $1/3$ on every good instance in
$
\operatorname{supp}(\mathcal D_-).
$
Hence there exists an instance in one of these two families on which
$V_i$ fails. Fix $O^{(n_i)}$ to be such an instance, and define
$$
1^{n_i}\in L^A
\quad\Longleftrightarrow\quad
O^{(n_i)}\in\operatorname{supp}(\mathcal D_+).
$$

This completes the construction at the $i$-th stage. Note that every
oracle component of index at most $B_i$ is now permanently fixed.
Since
$$
n_i>B_{i-1}
\qquad\text{and}\qquad
B_i\ge n_i,
$$
we have $B_i\to\infty$. Thus every oracle component is eventually
fixed. Since $V_i$ cannot query any component of index larger than
$$
a_i(n_i)\le B_i
$$
on input $1^{n_i}$, no later stage can change its computation on
this input. Thus $V_i$ remains incorrect on $1^{n_i}$ in the final
oracle.

At every input length $n$ that is not equal to any $n_i$, the
component $O^{(n)}$ is the trivial YES instance above, and we declare
$$
1^n\in L^A.
$$

Repeating the construction for every $i$ defines the Boolean oracle
$A$ and the unary promise problem $L^A$. By construction, for every
$i$, the verifier $V_i$ fails to decide $L^A$ on input $1^{n_i}$.
Since every polynomial-time QMA oracle verifier appears in the
enumeration,
$
L^A\notin\QMA^A.
$

It remains to show that
$
L^A\in\QMA(2)^A.
$
There is a single uniform verifier. On input $1^n$, the verifier
uses the $n$-th oracle component and runs the QMA(2) protocol of
\cref{prop:qma2-upper}.

If
$
1^n\in L^A,
$
then either $O^{(n)}$ is a hard-YES instance or it is the trivial
YES instance. In either case,
$
\lambda_{\mathrm{sep}}(H)=0,
$
where for the hard-YES family this follows from
\cref{clm:hard_yes:yes}.

If
$
1^n\notin L^A,
$
then $n=n_i$ for some $i$ and $O^{(n)}$ is a good hard-NO instance.
By \cref{proposition_goodNO},
$
\lambda_{\mathrm{sep}}(H)\ge\frac14.
$

Therefore the verifier of \cref{prop:qma2-upper}, followed by the
amplification stated there, decides $L^A$ with completeness at least
$2/3$ and soundness at most $1/3$. The standard Boolean encoding of
\cref{eq:oracle} incurs only polynomial query overhead. Hence
$
L^A\in\QMA(2)^A.
$

Combining the two conclusions,
$$
L^A\in\QMA(2)^A\setminus\QMA^A,
$$
and therefore
$$
\QMA^A\neq\QMA(2)^A.
$$
\end{proof}